\newif\ifsubmit  
\newif\ifllncs      
\newif\ifexabs      
\newif\ifblind      

\submittrue
\blindfalse

\ifllncs
  \documentclass[runningheads,a4paper]{llncs}

\else
  \documentclass[letterpaper,11pt]{article}
  \usepackage[in]{fullpage}
  \usepackage{setspace}
  \usepackage[margin=1in]{geometry}
\fi

\usepackage{iftex}
\ifPDFTeX
  \usepackage[utf8]{inputenc}
  \usepackage[noTeX]{mmap}
  \usepackage[T1]{fontenc}
\fi
\ifLuaTeX
  \usepackage{luatex85}
  \usepackage[noTeX]{mmap}
\fi

\usepackage{mdframed}
\usepackage{amsmath}
\usepackage{amsfonts}
\usepackage{amssymb}
\usepackage{amsthm}
\usepackage{color}

\newcommand{\qnote}[1]{%
  \noindent
  {\color{blue}[\textbf{Qipeng:} #1]}%
}

\allowdisplaybreaks[3]

\usepackage{appendix}
\usepackage{algorithm}
\usepackage{algpseudocode}
\usepackage{braket}
\usepackage{comment}
\usepackage{url}
\usepackage{bbm}
\usepackage{multicol}
\usepackage{mdframed}
\usepackage{multirow}

\usepackage[bookmarks]{hyperref}
\usepackage[nameinlink]{cleveref}
\hypersetup{
  colorlinks,
  allcolors=blue
}

\ifllncs

  \spnewtheorem{claim}{Claim}{\bfseries}{\rmfamily}
  \crefname{claim}{claim}{claims}
  \Crefname{claim}{Claim}{Claims}
\else
  \newtheorem{theorem}{Theorem}[section]
  \newtheorem{definition}[theorem]{Definition}
  \newtheorem{conjecture}[theorem]{Conjecture}
  \newtheorem{remark}[theorem]{Remark}
  \newtheorem{lemma}[theorem]{Lemma}
  \newtheorem{corollary}[theorem]{Corollary}

  \newtheorem{claim}[theorem]{Claim}
  \newtheorem*{remark*}{Remark}

  \newtheorem{ob}{Observation}
  \newtheorem*{theorem*}{Theorem}
  \newtheorem*{lemma*}{Lemma}

\newcommand{\email}[1]{\href{mailto:#1}{\texttt{#1}}}
\usepackage[style=alphabetic,minalphanames=3,maxalphanames=4,maxnames=99,backref=true]{biblatex}

  \DeclareFieldFormat{eprint:iacr}{Cryptology ePrint Archive: \href{https://ia.cr/#1}{\texttt{#1}}}
  \DeclareFieldFormat{eprint:iacrarchive}{Cryptology ePrint Archive: \href{https://eprint.iacr.org/archive/#1}{\texttt{#1}}}
  \AtEveryBibitem{
    \clearlist{address}
    \clearfield{date}
    \clearfield{isbn}
    \clearfield{issn}
    \clearlist{location}
    \clearfield{month}
    \clearfield{series}
 
    \ifentrytype{book}{}{
      \clearlist{publisher}
      \clearname{editor}
    }
  }
\fi

\usepackage{appendix}
\usepackage{algorithm}
\usepackage{algpseudocode}
\usepackage{braket}
\usepackage{comment}
\usepackage{url}
\usepackage{multicol}
\usepackage{tikz}
\usetikzlibrary{arrows,automata,positioning}
\usepackage{caption}
\usepackage[caption=false]{subfig}
\usepackage[font=small,labelfont=bf]{caption}
\usepackage{comment}
\usepackage{pifont}
\usepackage{quantikz}
\usetikzlibrary{patterns}
\usetikzlibrary{arrows.meta}

\usepackage{dsfont}

\usepackage{enumitem}
\setlist[description]{noitemsep}
\setlist[enumerate]{noitemsep}
\setlist[itemize]{noitemsep}

\usepackage{soul, xcolor, xparse}
\makeatletter
  \ExplSyntaxOn
    \cs_new:Npn \white_text:n #1
    {
      \fp_set:Nn \l_tmpa_fp {#1 * .01}
      \llap{\textcolor{white}{\the\SOUL@syllable}\hspace{\fp_to_decimal:N \l_tmpa_fp em}}
      \llap{\textcolor{white}{\the\SOUL@syllable}\hspace{-\fp_to_decimal:N \l_tmpa_fp em}}
    }
    \NewDocumentCommand{\whiten}{ m }
    {
      \int_step_function:nnnN {1}{1}{#1} \white_text:n
    }
  \ExplSyntaxOff
  
  \NewDocumentCommand{ \varul }{ D<>{5} O{0.2ex} O{0.1ex} +m } {%
    \begingroup
    \setul{#2}{#3}%
    \def\SOUL@uleverysyllable{%
      \setbox0=\hbox{\the\SOUL@syllable}%
      \ifdim\dp0>\z@
      \SOUL@ulunderline{\phantom{\the\SOUL@syllable}}%
      \whiten{#1}%
      \llap{%
        \the\SOUL@syllable
        \SOUL@setkern\SOUL@charkern
      }%
      \else
      \SOUL@ulunderline{%
        \the\SOUL@syllable
        \SOUL@setkern\SOUL@charkern
      }%
      \fi}%
    \ul{#4}%
    \endgroup
  }
\makeatother

\ifsubmit
    \newcommand{\qipeng}[1]{}
    \newcommand{\saachi}[1]{}
\else
    \newcommand{\qipeng}[1]{{\color{red} Qipeng: #1}}
    \newcommand{\saachi}[1]{{\color{blue} Saachi: #1}}
\fi

\newcommand{\As}{\mathcal{A}}

\renewcommand{\kappa}{\ell}

\newcommand{\rv}[1]{\mathbf{#1}}

\DeclareMathAlphabet\mathbfcal{OMS}{cmsy}{b}{n}

\DeclareFontFamily{U}{skulls}{}
\DeclareFontShape{U}{skulls}{m}{n}{ <-> skull }{}

\DeclareRobustCommand{\substack}[1]{\subarray{c}#1\endsubarray}

\newenvironment{boxfig}[2]{\begin{figure}[#1]\fbox{\begin{minipage}{0.97\linewidth}
                        \vspace{0.2em}
                        \makebox[0.025\linewidth]{}
                        \begin{minipage}{0.95\linewidth}
            {{
                        #2 }}
                        \end{minipage}
                        \vspace{0.2em}
                        \end{minipage}}
                        }
                        {\end{figure}}

\crefname{theorem}{theorem}{theorems}
\Crefname{theorem}{Theorem}{Theorems}

\crefname{corollary}{corollary}{corollaries}
\Crefname{corollary}{Corollary}{Corollaries}

\crefname{conjecture}{conjecture}{conjectures}
\Crefname{conjecture}{Conjecture}{Conjectures}

\title{
    Parallel Quantum Advantage with Limited Adaptivity Requires Structure \\

}

\ifblind
\author{
}

\else
\author{
     Qipeng Liu\footnote{UC San Diego,~\email{qipengliu0@gmail.com}} \and Saachi Mutreja\footnote{Columbia University,~\email{sm5540@columbia.edu}}
}
\fi

\date{}

\begin{document}

\maketitle

\begin{abstract}
Aaronson and Ambainis (Theory of Computing, 2014) conjectured that quantum query algorithms admit efficient almost-everywhere classical simulation: for any $T$-query quantum algorithm, its acceptance probability can be approximated on a $(1-\delta)$ fraction of inputs, up to $\epsilon$ additive error, using $\mathrm{poly}(T, 1/\epsilon, 1/\delta)$ classical queries. At a high level, the conjecture suggests that exponential quantum speedups are possible only on sufficiently structured inputs.

In this work, we make progress on this conjecture by proving it for quantum algorithms that make \emph{massively parallel quantum queries}. In contrast, Yamakawa and Zhandry (Journal of the ACM, 2024) showed that quantum algorithms restricted to parallel queries can still achieve exponential speedups over classical algorithms for sampling problems.
We establish our simulation theorem by proving the \emph{stronger} statement that parallel-query quantum algorithms cannot distinguish the uniform distribution over oracles from oracles drawn from so-called "dense distributions".  Our main technical contribution is a coupling theorem that relates the uniform distribution over oracles to oracles drawn from dense distributions.

We further extend this approach beyond the purely parallel setting. We  prove  simulation theorems for  quantum algorithms with \emph{constant rounds of adaptivity}, and   for quantum algorithms with a \emph{bounded quantum-query prefix} followed by a massively parallel quantum-query stage. 

\end{abstract}
\newpage
\tableofcontents
\newpage

\section{Introduction}

Quantum computers have offered significant speedups for many computational problems. However, superpolynomial quantum speedups have so far been confined to mathematical problems exhibiting underlying structure—such as periodicity, quadratic residue patterns, or similar regularities. In the context of quantum algorithms,''structured” typically refers to the task of extracting a global property of a massively long sequence, under the assumption that the sequence obeys some underlying regularity. Indeed, quantum computing has offered exponential speedups relative to structured periodic oracles ~\cite{simon1997power,shor1999polynomial}.  By contrast,  relative to unstructured oracles, quantum algorithms have only been able to offer polynomial speedups, a canonical example being the Grover search problem~\cite{grover1996fast}. Moreover, this speedup is optimal, as shown by Bennett, Bernstein, Brassard, and Vazirani \cite{BBBV97}. As far as we know, super polynomial quantum speedups have only been observed in quantum computing when inputs admit certain structured distributions\footnote{This holds for decision problems. Recently, Yamakawa and Zhandry~\cite{yamakawa2022verifiable} showed an exponential separation between quantum and classical algorithms for a sampling problem.}. 

A central question in quantum computing has been to understand when a quantum speedup comes from exploiting genuine structure in the input, and when the behavior of a quantum algorithm can already be predicted from a small amount of classical information on a typical input.  
Formalizing the observation, Aaronson and Ambainis~\cite{AA08} proposed the so-called \emph{simulation conjecture.} 
The simulation conjecture gives one clean formulation of this question.  
Informally, it asks the following question: 

\begin{center} Does every $T$-query quantum algorithm admit an almost-everywhere classical simulation using only $\mathrm{poly}(T)$ queries?
\end{center}
The formal statement of this conjecture is as follows:
\begin{conjecture}[Simulation Conjecture]\label{conj:sim_intro}
    Let $\mathcal{A}$ be a quantum algorithm making $T$ queries to a boolean input $x=(x_1, \dots x_N)$, and let  $\epsilon, \delta > 0$. There is a deterministic classical algorithm that makes $\text{poly}(T, 1/\epsilon, 1/\delta)$ queries to the $x_i's$ and approximates $\mathcal{A}'s$ acceptance probability with an additive error $\epsilon$ on a $(1-\delta)$ fraction of inputs. 
\end{conjecture}
\noindent This conjecture is not a worst-case statement. Rather, it asserts that on most inputs, a bounded-query quantum algorithm should leave behind enough low-complexity structure to be learned classically.

We prove the simulation conjecture for several natural classes of quantum query algorithms whose main access to the input occurs through a massively parallel query layer. Informally, our main result is the following.

\begin{theorem}[Main theorem, informal]
The simulation conjecture holds for the following classes of quantum algorithms:
\begin{enumerate}
    \item \textbf{Parallel Query Quantum algorithms:} Quantum algorithms that make a single layer of massively parallel quantum queries. 
    \item \textbf{Terminal-Parallel Query Quantum algorithms:}
    \begin{itemize} 
        \item Algorithms with a bounded, slowly growing number of adaptive quantum queries before the parallel layer.
        \item Algorithms with polynomially many adaptive classical queries before the parallel layer.
    \end{itemize} 
 
\item \textbf{Constant round parallel query algorithms:} Quantum algorithms with a constant number of adaptive parallel-query layers.
\end{enumerate}
\end{theorem}

Aaronson and Ambainis proposed a route towards resolving the simulation conjecture through the polynomial method, by relating it to a fundamental open question in Fourier analysis. By the result of Beals, Buhrman, Cleve, Mosca and de Wolf.~\cite{beals2001quantum}, the acceptance probability of a $T$-query quantum algorithm is a bounded polynomial of degree at most $2T$. This turns the simulation conjecture into a question in the analysis of bounded low-degree polynomials on the Boolean cube. The Aaronson--Ambainis (AA) conjecture posits that every bounded degree-$d$ polynomial $p:\{0,1\}^N\to[0,1]$ must have an influential variable, quantitatively, $\mathrm{Inf}_i[p]\geq \mathrm{poly}(\mathrm{Var}[p],1/d)$ for some coordinate $i$. The point of this influence lower bound is algorithmic: once one can identify a coordinate that controls a polynomial by a noticeable amount, one can query and restrict that coordinate, iterate, and drive the residual variance down until the remaining polynomial is well-approximated by its mean. In this way, the AA conjecture provides a mechanism for converting global boundedness and low degree into a shallow classical decision tree that works on most inputs. 

Substantial progress has been made in several structured regimes, although the general conjecture remains open. 
Montanaro~\cite{Montanaro12} proved the conjecture for a flat multilinear regime, namely multilinear forms whose Fourier coefficients all have the same magnitude. Defant, Masty\l o and Perez~\cite{DMP18} extended this flat-spectrum viewpoint to bounded polynomials with equal-magnitude Fourier coefficients, obtaining a subexponential $\exp(O(\sqrt{d\log d}))$ dependence rather than the polynomial dependence required by AA. Again, the result is formulated for a Fourier-analytic class, not for a class of quantum algorithms. O'Donnell and Zhao~\cite{OZ16} then showed that, up to polynomial losses, it suffices to prove AA for one-block decoupled polynomials, so their contribution is best viewed as a structural reduction of the general conjecture rather than a new algorithmic special case. Most recently, Bhattacharya~\cite{Bhattacharya24} proved that AA holds for a non-negligible fraction of random restrictions of any bounded degree-$d$ polynomial of sufficiently large variance; this applies in principle to every quantum acceptance polynomial after restriction, but it is not tied to a separately characterized family of quantum query algorithms.

A more explicitly algorithmic breakthrough was obtained by Bansal, Sinha and de Wolf~\cite{BSdW22}, who proved AA for completely bounded degree-$d$ block-multilinear forms; using the completely bounded formalism of Arunachalam, Briet and Palazuelos~\cite{ABP19}, this yields almost-everywhere classical simulation for the associated class of quantum algorithms that query $d$ times on disjoint inputs; in other words, for two different queries, they will act on disjoint input sets. 
Building on this, Gutierrez~\cite{Gutierrez24} introduced Fourier completely bounded polynomials, which characterize arbitrary $d$-query quantum algorithms directly on the Boolean cube, and proved AA for the homogeneous Fourier completely bounded case; equivalently, whenever a $T$-query quantum algorithm has output polynomial homogeneous of degree $2T$.
\paragraph{Our approach: Dense Indistinguishability}
Prior work has largely attacked the simulation conjecture through the broader Fourier-analytic AA conjecture. Our approach instead works directly with quantum query algorithms through dense indistinguishability. This direct approach is motivated by a gap between quantum query algorithms and bounded low-degree polynomials. 
While every $T$-query quantum algorithm gives rise to a degree-$2T$ polynomial describing its acceptance probability, the converse fails in general: not every bounded low-degree polynomial arises from a quantum query algorithm. Thus, approaching the simulation conjecture through the AA conjecture may require overcoming structural difficulties that are artifacts of the broader polynomial setting rather than of quantum algorithms themselves. Motivated by this gap, in this paper we study the simulation conjecture directly, by exploring its link to a conjecture which, roughly speaking, posits that quantum algorithms cannot, by making few queries, distinguish between
the uniform distribution over oracles versus oracles drawn from so-called “dense”
distributions. 
We call this the Dense Indistinguishability Conjecture.
\paragraph{Dense Indistinguishability Conjecture:}
Before stating the conjecture formally, we  explain the notion of a \emph{dense distribution.}  We say that a random variable $\rv{F}$ over $\{0,1\}^N$ is \emph{($1-\delta$)-dense} if for all subsets $S \subseteq [N]$ of coordinates, the marginal distribution of $F \sim \rv{F}$ restricted to the coordinates in $S$ has min-entropy at least $(1 - \delta)|S|$. In other words, every subset of coordinates has near maximal min-entropy. 
\begin{conjecture}[Dense Indistinguishability Conjecture]\label{conj:dense_intro}
Let $O_X$ be a $(1-\delta)$-dense distribution over $O:[N] \to \{0,1\}$ and $O_U$ be uniform.
There exist absolute constants $C, a \geq 0, b > 0$ such that for all $T$ query quantum algorithms $\As$, 
\[
    \Big|\Pr_{O \sim O_X}[\mathcal{A}^O \rightarrow 1]-\Pr_{O \sim O_U}[\mathcal{A}^O \rightarrow 1]\Big|\leq C \cdot T^a \cdot \delta^b.
\]
\end{conjecture}

A similar conjecture was first formulated by Guo, Li, Liu, and Zhang~\cite{guo2021unifying}. It asserts that, for any quantum algorithm, the effect of arbitrary classical advice can be simulated by ordinary advice that records only input-output pairs of the random oracle. Here, we consider the dense indistinguishability conjecture, which implies and is seemingly stronger than the conjecture in~\cite{guo2021unifying}.

\begin{lemma}\cite{guo2021unifying}(Conjecture implication, informal.)\label{lem:informalimplication}
    Conjecture \ref{conj:dense_intro} implies Conjecture \ref{conj:sim_intro}.
\end{lemma}
Conjecture \ref{conj:dense_intro}  is known to be true for all classical query algorithms, as first proven by 
Coretti, Dodis, Guo, and Steinberger \cite{coretti2018random}.
Furthermore, this conjecture was also identified to have connections to a classical oracle separation between QMA and QCMA~\cite{liu2025qma}.

\subsection{Our Results}

In this paper, we consider 
\begin{enumerate}
    \item Quantum algorithms that may perform limited preprocessing, but whose main quantum access to the input occurs through one massively parallel query layer.
    \item Quantum algorithms with a constant number of adaptive parallel query layers. 
\end{enumerate}
We ask whether such algorithms  satisfy the almost-everywhere classical simulation principle predicted by the simulation conjecture. More specifically, parallel quantum-query algorithms consist of the following  steps:
\begin{itemize}
    \item Prepare a potentially entangled register that is independent of the input $x$:
    \begin{align*}
        \sum_{i_1, b_1, i_2, b_2, \ldots, i_T, b_T, z} \alpha_{i_1, b_1, i_2, b_2, \ldots, i_T, b_T, z}  \ket{i_1, b_1, i_2, b_2, \ldots, i_T, b_T, z},
    \end{align*}
    where $z$ is the ancilla, and for all $j \in [T]$, $i_j,b_j$ are the $j$-th query registers where $i_j$ denotes the input and $b_j$ denotes the control register. 
    \item Apply $T$ parallel queries to obtain:
    \begin{align*}
        \sum_{i_1, b_1, i_2, b_2, \ldots, i_T, b_T, z} (-1)^{x_{i_1} b_1 + x_{i_2} b_2 + \cdots + x_{i_T} b_T} \alpha_{i_1, b_1, i_2, b_2, \ldots, i_T, b_T, z}  \ket{i_1, b_1, i_2, b_2, \ldots, i_T, b_T, z}.
    \end{align*}
    \item Finally, an arbitrary basis measurement is applied to the post-query state.
\end{itemize}
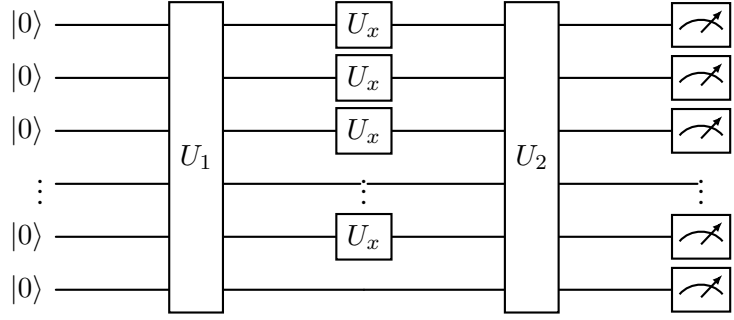
\begin{figure}[t]
    \centering
\[
\begin{quantikz}[row sep={0.7cm,between origins}, column sep=1.5cm]
\lstick{$\ket{0}$} & \gate[wires=6]{U_1} & \gate{U_x} & \gate[wires=6]{U_2} & \meter{} \\
\lstick{$\ket{0}$} &                      & \gate{U_x} &                      & \meter{} \\
\lstick{$\ket{0}$} &                      & \gate{U_x} &                      & \meter{} \\
\lstick{$\vdots$} &                       & \vdots &                       & \vdots \\
\lstick{$\ket{0}$} &                     & \gate{U_x}   &                     & \meter{}   \\
\lstick{$\ket{0}$} &                     &     &                     & \meter{}
\end{quantikz}
\]
\label{fig:parallel_algo}
\caption{Parallel query algorithms considered in this work. The last wire indicates the ancilla.}
\end{figure}

\paragraph{One base case, two extension routes.} Our central technical result proves Conjecture \ref{conj:dense_intro} for parallel query quantum algorithms. We construct a coupling that transforms any dense oracle
distribution into the uniform distribution while changing each coordinate with small probability. In
one parallel layer, the query weight placed on each coordinate is fixed before the oracle is seen. A
parallel version of the BBBV hybrid argument can therefore be averaged over the coupling, yielding
the desired indistinguishability bound and hence the parallel query simulation theorem by invoking Lemma \ref{lem:informalimplication}. This theorem then supports two complementary extensions.
\paragraph{First extension- Terminal parallel query algorithms:}The first extension applies to Conjecture \ref{conj:dense_intro}, and consequently to Conjecture \ref{conj:sim_intro}. It is organized around the
Computationally Hidden Flipped Set (CHFS) conjecture (Conjecture \ref{CHFSgeneral}, Section \ref{sec:CHFS_implies_dense}). The coupling identifies a sparse random set
of coordinates on which a dense and a uniform oracle differ. Adaptivity becomes dangerous because
future query weights can correlate with this flipped set. CHFS asks if the flipped set can remain
computationally hidden even from a quantum algorithm interacting with the oracle. We prove
weaker forms of this hiding statement that suffice for a polynomially long adaptive classical prefixes
and for bounded, slowly growing adaptive quantum prefixes followed by one terminal parallel-query
layer.
\paragraph{Second extension- Constant round parallel query algorithms:}The second extension uses the simulation theorem for parallel query algorithms as a base case of a strong induction based argument on the number of rounds of adaptivity, yielding simulation theorems for quantum algorithms with constant rounds of adaptivity, where each round can be an arbitrarily large parallel query.\\

These two routes cover incomparable regimes. The CHFS route allows the number of adaptive
quantum queries before the $T$ parallel quantum queries to grow as $\frac{\log K}{\log \log K}$, where $K= \max(T, 1/\epsilon, 1/\delta)$, and $\epsilon, \delta$ are as defined in Conjecture \ref{conj:sim_intro}. The fixed-layer induction cannot accommodate such growing
depth: its polynomial degree depends exponentially on the number of layers, so polynomial query
complexity is obtained only when the number of layers is fixed.
\\

\noindent We now state our main results formally.
\begin{theorem}[Parallel Query Theorem, Dense Indistinguishabiliy Conjecture]\label{thm:maindense}
    Let $O_X$ be a $(1-\delta)$-dense distribution over $O:[N] \to \{0,1\}$ and $O_U$ be uniform.
There exist absolute constants $C, a \geq 0, b > 0$ such that for all $T$ query parallel quantum algorithms $\As$, 
\[
    \Big|\Pr_{O \sim O_X}[\mathcal{A}^O \rightarrow 1]-\Pr_{O \sim O_U}[\mathcal{A}^O \rightarrow 1]\Big|\leq C \cdot T^a \cdot \delta^b.
\]
\end{theorem}
This theorem implies Conjecture \ref{conj:sim_intro} for parallel query algorithms.
\begin{theorem}[Parallel Query Theorem, Simulation Conjecture]\label{thm:main_intro}
Let $\mathcal{A}$ be a quantum algorithm making $T$ parallel queries (as defined above) to a boolean input $x=(x_1, \dots x_N)$, and $\epsilon, \delta > 0$. There is a deterministic classical algorithm that makes $\text{poly}(T, 1/\epsilon, 1/\delta)$ queries to the $x_i's$ and approximates $\mathcal{A}'s$ acceptance probability with an additive error $\epsilon$ on a $(1-\delta)$ fraction of inputs. 
\end{theorem}

Note that in our setting, the query registers may contain overlapping inputs and may also involve controlled queries. As a result, the model we consider is not covered by the framework of~\cite{BSdW22}, which applies to algorithms whose quantum queries are restricted to disjoint input blocks. Thus, the result of~\cite{BSdW22} and ours concern two incomparable classes of quantum algorithms. Importantly, although parallel-query quantum algorithms may appear structurally simple at first glance, they already capture several natural and meaningful forms of quantum computation. In particular, this class is powerful enough to exhibit exponential quantum-classical separations for sampling problems, as shown by the Yamakawa--Zhandry problem~\cite{yamakawa2022verifiable}. It also includes the standard quantum algorithm for Forrelation~\cite{bqppolyhierarchy}. Therefore, our theorem applies to a class of quantum algorithms that is both mathematically nontrivial and rich enough to cover canonical examples exhibiting quantum advantage.

\paragraph{Terminal-parallel algorithms and the CHFS route.} The purely parallel model is the base case of terminal-parallel quantum query algorithms, by which we mean algorithms whose main quantum access to the input occurs through a final massively parallel query layer. We next prove Conjecture \ref{conj:dense_intro} and  Conjecture \ref{conj:sim_intro} for quantum algorithms that perform limited preprocessing before the final massively parallel quantum layer, by introducing the Computational Hidden Flipped Set (CHFS) conjecture (Conjecture \ref{CHFSgeneral}), which we briefly discussed above and will also explain in the overview. This captures two natural forms of limited adaptivity, described as follows:   
\begin{enumerate}
 
\item First, the preprocessing may consist of an arbitrary polynomial number of adaptive classical queries.
\item Second, it may consist of a bounded number of quantum queries, after which the algorithm makes one large parallel quantum-query layer. In this case, the quantitative bound deteriorates exponentially in the size $r$ of the quantum prefix, so the result applies not only to a constant $r$, but also for slowly growing $r$.
\end{enumerate}
These results show that our approach is not restricted to completely non-adaptive quantum algorithms, but extends to terminal-parallel algorithms whose adaptivity is confined to a controlled prefix.
 We note that first variant can also be obtained more directly from our parallel-query simulation theorem in a black-box way. We nevertheless derive it through CHFS in order to demonstrate that CHFS is not merely tailored to one special case, but may serve as a more general framework for future progress on the full simulation conjecture.

\begin{theorem}[Simulation for Terminal-Parallel Algorithms with Limited Preprocessing]\label{thm:main_intro2}

Let $\mathcal{A}$ be a quantum algorithm making \underline{ $r$  (adaptive) quantum queries}, followed by a layer of $T_q$ parallel  queries to a boolean input $x=(x_1, \dots x_N)$,  and  $\epsilon, \delta > 0$. Let $K = \max\{T_q,{1}/{\delta},{1}/{\epsilon}\}$ and $r = O(\log K / \log\log K)$. There is a deterministic classical algorithm that makes $\text{poly}(T_q, 1/\epsilon, 1/\delta)$ queries to the $x_i's$ and approximates $\mathcal{A}'s$ acceptance probability with an additive error $\epsilon$ on a $(1-\delta)$ fraction of inputs.

Let $\mathcal{A}$ be a quantum algorithm with \underline{$T_c$ (adaptive) classical queries}, followed by a layer of $T_q$ parallel queries to a boolean input $x=(x_1, \dots x_N)$, and $\epsilon, \delta > 0$. There is a deterministic classical algorithm that makes $\text{poly}(T_c, T_q, 1/\epsilon, 1/\delta)$ queries to the $x_i's$ and approximates $\mathcal{A}'s$ acceptance probability with an additive error $\epsilon$ on a $(1-\delta)$ fraction of inputs.

\end{theorem}
\paragraph{Constant round parallel quantum algorithms:} We then consider a complementary model. A $T$ query, $d$-layer parallel-query quantum algorithm  has the form
$$
V_d U_O^{\otimes T_d}V_{d-1}\cdots V_1U_O^{\otimes T_1}V_0,
$$
followed by an arbitrary basis  measurement, where the $V_r$ are oracle-independent unitaries and layer $r$ contains $T_r$ parallel queries. This result is proven by a strong induction based argument on the number of layers $d$ of the quantum algorithm. Crucially,  Theorem \ref{thm:main_intro} provides the base case for the induction.
\begin{theorem}[Simulation for Constant Round Parallel Quantum Algorithms]\label{thm:fixed_adaptivity_intro}There exist absolute constants $c,C>0$ such that the following holds. Let $\mathcal{A}$ be a $T$ query quantum algorithm with $d$ rounds of adaptivity.  There is a deterministic classical algorithm that makes $O\Big(\frac{CdT}{\epsilon \cdot \sqrt{\delta}}\Big)^{c4^d}$ queries to the $x_i's$ and approximates $\mathcal{A}'s$ acceptance probability with an additive error $\epsilon$ on a $(1-\delta)$ fraction of inputs. 
    
\end{theorem}
We emphasize that Theorem \ref{thm:fixed_adaptivity_intro}  is not obtained by establishing the stronger statement of Conjecture \ref{conj:dense_intro} for constant round parallel query algorithms. Furthermore,
 Theorem \ref{thm:fixed_adaptivity_intro} is \emph{complementary} to Theorem \ref{thm:main_intro2}. The induction can absorb many queries inside each of a fixed number
of layers, but its exponent grows like $4^d$ and therefore does not prove Conjecture \ref{conj:sim_intro} for quantum algorithms with a slowly growing number of adaptive queries. 

\begin{remark}[Dense indistinguishability versus simulation]
Although Theorems~\ref{thm:main_intro2} and~\ref{thm:fixed_adaptivity_intro}
both establish classical simulation, their proofs extend the one-layer
case in different ways. The CHFS route to
Theorem~\ref{thm:main_intro2} extends the dense-indistinguishability
argument to terminal-parallel algorithms with limited preprocessing.
Theorem~\ref{thm:fixed_adaptivity_intro}, however, uses the parallel-query
simulation theorem (Theorem~\ref{thm:main_intro}) as the base case of an
induction, without proving Conjecture~\ref{conj:dense_intro} for several
adaptive parallel-query layers.

Indeed, dense indistinguishability concerns every
dense oracle distribution individually.
In Section~\ref{subsec:boolean-presampling}, we prove that the simulation
conjecture (Conjecture \ref{conj:sim_intro}) is equivalent to the Boolean presampling conjecture (Conjecture 4 in  \cite{guo2021unifying}), which
concerns a collection of distributions. The same proof does not establish
an equivalence between the dense indistinguishability Conjecture~\ref{conj:dense_intro} and the simulation conjecture ~\ref{conj:sim_intro}.
Moreover, the additional structure in Boolean presampling suggests that
dense indistinguishability is stronger than simulation, although we do not establish a strict
separation.
\end{remark}

\paragraph{Connections to Certifiable Randomness.}
Yamakawa and Zhandry~\cite{yamakawa2022verifiable} showed that any quantum algorithm that successfully solves the Yamakawa--Zhandry problem must produce an output with high entropy. As an immediate consequence, our results can be interpreted as yielding certifiable randomness protocols for restricted classes of quantum provers.

\begin{theorem}[Certifiable Randomness for Restricted Provers, Informal]
There exists a certifiable randomness protocol in the quantum random oracle model that is secure against 
\begin{enumerate}
    \item quantum algorithms making either constant rounds of polynomially many parallel queries, or
    \item polynomially many classical queries followed by a single layer of polynomially many parallel quantum queries.
\end{enumerate} 
\end{theorem}

This theorem follows directly from Theorem \ref{thm:main_intro2} and Theorem  \ref{thm:fixed_adaptivity_intro}. Indeed, the construction in~\cite{yamakawa2022verifiable} relies on parallel repetition, which preserves the structure of a single parallel-query stage.

\subsection{Proof Overview and Discussions}

As mentioned above, in this work, we follow an approach initiated by Guo, Li, Liu, and Zhang~\cite{guo2021unifying}, who identified a strong connection between the simulation conjecture and the notion of dense distributions. Let $X$ be a random variable over $\{0,1\}^N$. We say that $X$ is $(1-\delta)$-dense if, for every subset $S \subseteq [N]$, the min-entropy of the marginal $X_S$ satisfies
\[
H_\infty(X_S) \geq (1-\delta)|S|.
\]
Equivalently, for every string $y_S \in \{0,1\}^{|S|}$,
\[
\Pr[X_S = y_S] \leq 2^{-(1-\delta)|S|}.
\]
Guo et al.~showed that the simulation conjecture would follow from a conjectured indistinguishability statement between dense distributions and the uniform distribution. To state this conjecture, it is convenient to view the input as an oracle $O : [N] \to \{0,1\}$, whose truth table can be viewed as a binary string in $\{0,1\}^N$. We adopt this convention throughout the paper.

Recall the dense-indinstinguishability conjecture from Conjecture \ref{conj:dense_intro}.

Our first observation is that the implication from Conjecture \ref{conj:dense_intro} to the simulation conjecture proven in~\cite{guo2021unifying} is algorithm-preserving. Namely, for every class of algorithms $\mathcal{C}$, if Conjecture \ref{conj:dense_intro} holds for all algorithms in $\mathcal{C}$, then the simulation conjecture holds for all algorithms in $\mathcal{C}$ as well. Therefore, it remains to show that dense distributions are indistinguishable from the uniform distribution for parallel-query quantum algorithms. We discuss this observation formally in~\Cref{sec:dense_implies_simu}.

At first glance, the one-layer parallel case may seem almost trivial. Indeed, for a \emph{classical}
algorithm making $T$ queries, the transcript depends only on the $T$ queried coordinates.
Hence the distinguishing advantage is controlled by the statistical distance between the
corresponding $T$-bit marginals under $O_X$ and $O_U$. Since $(1-\delta)$-denseness implies
that every such marginal places probability at most $2^{-(1-\delta)T}$ on each point, one gets
the elementary bound $1 - 2^{-\delta T} \le (\ln 2)\,\delta T$. 
This was shown by Coretti, Dodis, Guo and Steinberger~\cite{coretti2018random}.
In particular, the classical proof only uses the fact that all marginals of size at most $T$ are
close to uniform. However, this line of reasoning is fundamentally insufficient in the quantum
setting. Informally, there are distributions whose
small marginals, and even very large marginals, may look extremely dense, but which are
nevertheless distinguishable by quantum algorithms making very few queries (or just parallel queries). 
One such example is the Forrelated oracle~\cite{bqppolyhierarchy,aaronson2015forrelation}: it is almost uniform up to marginals of size $N^{1/4}$.
Thus, even in the parallel-query setting,
one cannot hope to argue from bounded-size marginals alone. 
A successful proof must exploit the
stronger fact that the distribution is dense \emph{for every} marginal of the oracle.

Our proof uses denseness in this strongest sense, and proceeds by constructing an explicit coupling between a dense oracle distribution and
the uniform distribution. More precisely, our main technical theorem shows that for every
$(1-\delta)$-dense distribution $O_X$, there exists a randomized \emph{modification algorithm}
$M^*$ such that $M^*(O)$ is distributed \emph{exactly} as $O_U$ when $O \leftarrow O_X$.
Moreover, the modification is sparse in a strong sense: 
\begin{theorem}[Good Coupling]\label{thm:good_coupling_intro}
For every coordinate $i \in [N]$, the
probability that $M^*$ flips the $i$-th oracle bit is at most
\begin{align*}
    \sqrt{\frac{\ln 2}{2}\delta} + o(1),
\end{align*}
where the probability is over the randomness of $M^*$ and the distribution of $O_X$, and $o(1)$ can be made arbitrarily small. 
\end{theorem}
Thus, the coupling does not merely say that one can globally transform a dense distribution into
uniform; it says that this transformation can be performed while keeping the probability of
touching each individual coordinate uniformly small. This per-coordinate guarantee is the key
quantitative statement that ultimately enables the indistinguishability argument. 
\paragraph{Application to Parallel Quantum Algorithms:}
Why is such a coupling useful for parallel quantum algorithms? The reason is a parallel-query
analogue of the BBBV hybrid argument~\cite{BBBV97}. 
For an arbitrary $T$-query quantum algorithm, the usual
hybrid method bounds the change in the final state by the total query magnitude on the
modified inputs across all query rounds. In the one-layer parallel setting, this bound becomes
particularly clean: if $F \subseteq [N]$ is the set of coordinates on which two oracles differ, then
the trace distance between the two final states is at most
$
\sqrt{\sum_{i \in F} W_i},
$
where $W_i$ is the probability that one of the $T$ parallel query registers contains $i$ in the
pre-query state. The simplification is that, in a single parallel round, the weights
$W_i$ are fixed by the input-independent pre-query state and therefore do \emph{not} depend on
the oracle itself. This oracle-independence is exactly what allows us to average over the coupling
and combine the BBBV bound with the per-coordinate flipping guarantee. After applying
Jensen's inequality and summing over coordinates, we obtain
\begin{align*}
\left|
\Pr_{O \sim O_X}[\mathcal{A}^O \to 1]
-
\Pr_{O \sim O_U}[\mathcal{A}^O \to 1]
\right|
\le
\sqrt{
T\cdot\left(\sqrt{\frac{\ln 2}{2}\delta}+o(1)\right)
}
=
T^{1/2}\left(\frac{\ln 2}{2}\delta\right)^{1/4}+o(1),
\end{align*}
which is the quantitative form of our main dense indistinguishability theorem for one-layer
parallel quantum algorithms. We formally prove it in~\Cref{sec:good_coupling_implies_dense}.

 We now explain the coupling theorem, which is our main technical contribution. 
A natural first attempt is to process the oracle bit-by-bit. Suppose we expose the oracle $O$ sampled from a $(1-\delta)$ dense distribution $O_X$ in some order, and at each step we
look at the conditional distribution of the next bit given the previously exposed prefix. If that
conditional bit is biased, we flip it with the appropriate probability to make it unbiased. By
construction, the output becomes exactly uniform at the end of the process. An entropy
calculation utilizing the denseness of the oracle distribution $O_X$ then shows that the expected \emph{total} number of flips is small: roughly speaking,
the sum of the conditional biases is controlled by the total entropy deficit, which is at most
$\delta N$. This yields a bound of order $O(\sqrt{\delta } \cdot N)$ on the expected number of modified
coordinates. While already nontrivial, this ``global'' guarantee is too weak for our purpose,
because a quantum algorithm might concentrate all of its query weight on a small set of
coordinates that are flipped disproportionately often. 

\paragraph{Achieving coordinate-wise control: }Another main new ingredient is a refinement of this idea that upgrades global control to
\emph{coordinate-wise} control. The key insight is that the order in which coordinates are exposed
should itself be randomized, and chosen carefully as a function of the distribution $O_X$. We prove,
by induction on $N$, that one can choose a distribution over such orders so that after averaging
over the order and the internal randomness of the modification procedure, every coordinate is
flipped with probability at most
$
\sqrt{\frac{\ln 2}{2}\delta}+o(1).
$
 Conceptually, the induction shows that the entropy deficit
can be distributed evenly enough across coordinates that no individual coordinate needs to be
modified too often, achieving Theorem \ref{thm:good_coupling_intro}. We choose the distribution over the orders in which the coordinates are exposed through a linear system designed to bound every coordinate’s modification probability. Summing the equations forces the weights to sum to one. The main technical difficulty is establishing that this system admits a nonnegative solution, so that these weights define a valid probability distribution. This requires a delicate argument using determinant identities and the matrix-tree theorem, whose details we defer to Section ~\Cref{sec:good_coupling_exits}.

\paragraph{Route I: From parallel algorithms to terminal-parallel algorithms.}

The above argument also makes clear why it does \emph{not} immediately imply the full dense
indistinguishability conjecture for general adaptive quantum algorithms. In the parallel setting,
the query weights $W_i$ are determined before the oracle is seen, so one can safely average over
the modification process. In the adaptive setting, by contrast, the algorithm's later query
distribution may depend on the earlier oracle answers. Consequently, the amount of query weight
placed on coordinate $i$ can be strongly correlated with the event that $i$ lies in the flipped set.
Thus, it is no longer sufficient to know that each coordinate is flipped with small \emph{average}
probability over the random oracle.

This gap is precisely why we introduce the Computational Hidden Flip Set (CHFS) conjecture as an alternate route to proving the dense indistinguishability conjecture, as well as the simulation conjecture. Put differently, what one would need is a stronger form of hiding, saying that even after interacting with the oracle, the algorithm still cannot \emph{find} a coordinate that is likely to have been flipped. 
 
\paragraph{Computational Hidden Flip Set Conjecture (CHFS):}
Now we introduce the CHFS conjecture. Informally, let $\mathcal{M}$ be a
modification algorithm that transforms a uniform oracle into a dense one,\footnote{Such a modification algorithm exists once the existence of a modification algorithm that transforms a dense oracle distribution into a uniform oracle distribution is proven. This is proven in section \ref{sec:good_coupling_exits}.} and let
$\mathrm{Flip}_{O, \mathcal{M}}$ denote the set of coordinates on which $\mathcal{M}$ changes the oracle.
The CHFS conjecture says that, for every efficient quantum algorithm $\mathcal{A}$ that outputs
a coordinate $i \in [N]$, the probability that $i \in \mathrm{Flip}_{O, \mathcal{M}}$ is at most
$C \cdot T^a \cdot \delta^b$ for some absolute constants $C, a \geq 0, b > 0$ (note the constants here are not necessarily the same as those in the dense indistinguishability conjecture). In other words, the flipped set should not only be sparse on average; it should
also be computationally hidden from any efficient quantum observer. \Cref{sec:CHFS_implies_dense} shows that such
a hidden-flip statement would imply dense indistinguishability for fully adaptive quantum query
algorithms. We view CHFS as a more direct formulation of the obstacle posed by adaptivity,
and therefore as a plausible intermediate target on the path toward the full simulation conjecture.  

\paragraph{Limited preprocessing before the terminal parallel layer.} After formulating CHFS, we prove two limited-adaptivity extensions where some version of this
hiding intuition can still be established, thus proving Conjecture \ref{conj:dense_intro} and \ref{conj:sim_intro} for quantum algorithms with limited adaptivity.
\indent \paragraph{Classical Preprocessing:}The first extension allows an arbitrary polynomial number of adaptive \emph{classical} queries
before the final parallel quantum round. The key point is that conditioning on a classical query
transcript does not destroy denseness too severely. Indeed, once a transcript $z$ is fixed, the
remaining unqueried part of the oracle is still dense, with only a mild degradation in the
density parameter on average over $z$. This lets us rerun the coupling
argument conditionally on the classical transcript and conclude indistinguishability for this
hybrid model. Informally, even though the algorithm is adaptive in its first stage, the first stage
only carves out a polynomial-size transcript, and the unobserved portion of the oracle remains
dense enough for the one-layer parallel argument to go through. This is shown in~\Cref{sec:hybrid_quantum}.

\paragraph{Quantum Preprocessing:}The second extension allows a bounded number of quantum queries, followed by the large parallel-query layer. Here the challenge is subtler,
because the quantum queries prior to the parallel-query layer  can create superpositions and correlations that depend on all inputs of the
oracle in a way classical transcripts cannot capture. Our intuition is that a bounded number of
quantum queries can extract only a limited amount of information, so the parallel-layer query
pattern cannot depend too sharply on the oracle. To formalize this, we study the amplitudes
generated after making a bounded number of quantum queries as low-degree functions of the input oracle and use concentration
bounds, via hypercontractivity, to bound their fluctuations. This yields a weaker hidden-flip
statement, sufficient when the number of quantum queries prior to the parallel query stage is bounded. In
this way we extend the simulation result to bounded-depth quantum algorithms, albeit with
a dependence on the number of queries prior to the parallel query layer that is only acceptable when the number of queries is bounded. This is shown in~\Cref{sec:two_layer_quantum}.

\paragraph{Route II: From parallel quantum algorithms to constant round parallel quantum algorithms:} This route bootstraps the one-layer parallel-query simulation theorem to algorithms with any fixed number $d$ of adaptive parallel-query layers. The simulator does not attempt to reconstruct the quantum state. Instead, it identifies the few oracle coordinates on which each query layer may place substantial weight and queries those coordinates classically. 

The proof proceeds by strong induction on $d$. In particular, let $a(O)$ be the acceptance probability of a $T$ query parallel quantum algorithm. Running  the classical simulator guaranteed by Theorem \ref{thm:main_intro} with additive error $\rho/\sqrt{2}$ and failure probability $\rho^2/2$ gives a decision tree $b$ satisfying
\[
\mathbb{E}_{O \sim O_U}[(a(O)-b(O))^2]\leq \rho^2.
\] This supplies the base case.

For the inductive step, consider a $d$ layer algorithm with $T_r$ queries in layer $r$, and let $T=\sum_r T_r$. For all $i\in [N], 1 \leq r\leq d$, let  $W_{r,i}(O)$ denote the query weight of coordinate $i$ in layer $r$. Two ingredients drive the construction.

First, in Lemma \ref{lem:d-layer-variance}, we show that the acceptance probability has small variance whenever the maximum remaining query weight is small in every layer:
\[
\operatorname{Var}[a(O)]\leq 16T^2\sum_{r=1}^d\sqrt{\mathbb{E}[\max_i W_{r,i}(O)]}
\]
The same bound holds on every restricted subcube. Thus, it suffices to construct a classical decision tree for which, in every quantum layer, all coordinates that remain unfixed have small query weight on average.
The weights in layer 1 are oracle-independent, so the simulator directly queries every coordinate whose first-layer weight exceeds a chosen threshold. 
Later layers require the main inductive idea.
\paragraph{From binary-output simulation to regularity of an
output distribution.}

For $r\geq 2$, we look at the query weights of layer $r$ of the quantum algorithm as an output distribution produced by the first $r-1$ layers. This distribution can be denoted as $\boldsymbol{p}^O= \{p^O_i\}_{i \in [N]}$. Crucially, this distribution can itself be sampled using only the first $(r-1)$ quantum layers. Hence every binary postprocessing of this distribution is the acceptance probability of an $r-1$-layer binary-output algorithm, to which the induction hypothesis applies. Now, the induction hypothesis supplies simulators for every binary post processing of this output distribution.

The difficulty is to control the entire
output distribution, which may have $N$ coordinates.
Applying a separate simulator to every coordinate would
introduce an unacceptable dependence on $N$. Lemma \ref{lem"regularize} reduces this task to binary-output
simulation. If $\mathbf{p}^O$ still varies substantially inside a current leaf, then some binary postprocessing of its output has substantial variance. The induction hypothesis supplies a classical simulator for this binary-output algorithm.
The classical simulator estimates this bit's acceptance
probability, and its transcript refines the current leaf. Whenever the average remaining variation
is too large, a round of such refinements produces some information gain. Hypercontractivity bounds the
total available mutual information by $O(Q_{r-1})$,
where $Q_{r-1}$ is the number of queries in the first
$r-1$ layers, independently of the output alphabet (Lemma \ref{lem:output-information}).
This information budget bounds the number of refinement
rounds, giving a depth bound independent of $N$.
Lemma \ref{lem"regularize} then constructs a classical transcript $Y_r$ conditioned on which the output distribution of the query weights in layer $r$ is close in mean square to its conditional mean, with depth independent of the number of possible output coordinates. More concretely, 
the resulting transcript $Y_r$ satisfies 
\[
\mathbb{E}\Big[||\boldsymbol{p}^O- \mathbb{E}[\boldsymbol{p}^O\mid Y_r]||\Big]\leq \eta
\]
for some carefully chosen $\eta$.
Then, at each regularized leaf, the simulator queries every coordinate $i$ such that $(\boldsymbol{p}^O \mid Y)_i$ is above the threshold. There are few such coordinates because the output distribution has total mass $\leq 1$. The remaining coordinates are such that their maximum layer-$r$ query weight is small on average.
Repeating this procedure for every layer produces a final transcript $Y$ such that, for all $r \in [d]$, 
\[
\mathbb{E}\Big[\max_{i \in \text{Free}_Y} W_{r,i}(O)\Big]
\]
is small. Here $\text{Free}_Y$ is the set of coordinates still free/unfixed at the final leaf. Applying the variance bound from Lemma \ref{lem:d-layer-variance} allows us to conclude that $\operatorname{Var}[a(O)|Y]$ is small. The full details of this induction based argument are in Section \ref{sec:fixed-adaptive-layers}. 

\section{AI Disclosure}
We initially tried to prove Conjecture \ref{conj:dense_intro} for adaptive quantum algorithms using an induction based argument. However, we realized that there is evidence of Conjecture \ref{conj:dense_intro} being significantly stronger than Conjecture \ref{conj:sim_intro}. We prompted ChatGPT 5.5 Pro to prove Conjecture \ref{conj:sim_intro} for more rounds of adaptivity using Theorem \ref{thm:main_intro}   and an induction based argument. ChatGPT 5.5 Pro suggested some of the key ideas required for the proof of Theorem \ref{thm:fixed_adaptivity_intro}. We found it crucial to prompt the model specifically to study Theorem \ref{thm:main_intro} and search for a way to perform induction; more open-ended attempts without this direction were not successful.  

\section{Concurrent and Independent Work}
Independent of our work, \cite{BDST26} proved that, if $f$ is the acceptance probability of a $d$-round quantum algorithm that makes $t$ parallel queries in each round, then there is a classical algorithm that makes $2^{O(d^2)}\cdot  (td \cdot \log(1/\delta)/\epsilon)^{O(d)}$ queries and approximates $f$ with an $\epsilon$ additive error on a $(1-\delta)$ fraction of inputs. 

Theorems \ref{thm:maindense}, \ref{thm:main_intro} and Theorem \ref{thm:main_intro2} in this paper were obtained concurrently and independently of \cite{BDST26}. Before either manuscript was publicly available and while we were working towards proving Theorem \ref{thm:fixed-adaptive-layers}, we exchanged informal statements of our theorems with the authors of ~\cite{BDST26}. No proof ideas or technical details were shared.  We subsequently independently obtained Theorem \ref{thm:fixed-adaptive-layers}. 

\section{Preliminaries}
\subsection{Directed Graphs, Laplacian and Matrix Determinants}
We define directed graphs, degree notions, Laplacians, arborescences and related notions. 
\subsubsection{Directed Graphs and Laplacian}
\
Let \( G = (V, E) \) be a directed graph on vertex set \( V = [n] = \{1,\dots,n\} \), where  
\( E \subseteq V \times V \). We assume \(G\) is unweighted.

For \(i,j \in V\), we write \(i \to j \in E\) if there is a directed edge from \(i\) to \(j\).

\begin{definition}[Neighborhoods]
For \(v \in V\), define:
\[
N^+(v) := \{u \in V : v \to u \in E\} \quad \text{(out-neighbors)},
\]
\[
N^-(v) := \{u \in V : u \to v \in E\} \quad \text{(in-neighbors)}.
\]
\end{definition}

\begin{definition}[Out-Degree]
The \emph{out-degree} and \emph{in-degree} of \(v\) are
\[
\deg^+(v) := |N^+(v)|, \qquad \deg^-(v) := |N^-(v)|.
\]
\end{definition}

\begin{definition}[Out-Laplacian]\label{out_laplacian}
Let \(D^+\) be the diagonal matrix with \(D^+_{ii} = \deg^+(i)\).  
The \emph{out-Laplacian} is
,
\[
L^{\mathrm{out}}_{ij} =
\begin{cases}
\deg^+(i) & \text{if } i=j,\\
-1 & \text{if } i \to j \in E,\\
0 & \text{otherwise}.
\end{cases}
\]
\end{definition}
\begin{definition}[In-Arborescence]
Let \(r \in V\). An \emph{in-arborescence rooted at \(r\)} is a directed tree \(T\) such that:
\begin{itemize}
    \item For every \(v \neq r\), there is exactly one outgoing edge from \(v\),
    \item The root \(r\) has out-degree zero,
    \item For every \(v \in V\), there is a directed path from \(v\) to \(r\).
\end{itemize}
Equivalently, all edges are oriented toward the root.
\end{definition}
\subsubsection{Matrix Determinants}
We recall basic definitions and properties of determinants that will be used throughout.

\begin{definition}[Determinant]
Let \(M \in \mathbb{R}^{n \times n}\). The determinant of \(M\), denoted \(\det(M)\) or $
|M|$, is defined by the Leibniz expansion
\[
\det(M)
:=
\sum_{\sigma \in S_n}
\operatorname{sgn}(\sigma)
\prod_{i=1}^n M_{i,\sigma(i)},
\]
where \(S_n\) is the set of permutations on \([n]\), and \(\operatorname{sgn}(\sigma)\in\{+1,-1\}\) is the sign of \(\sigma\).
\end{definition}

\begin{definition}[Minor]\label{minor}
For \(i,j \in [n]\), let \(M^{(i,j)}\) denote the matrix obtained by deleting the \(i\)-th row and \(j\)-th column of \(M\). The \emph{minor} of \(M\) at \((i,j)\) is
\[
\det\big(M^{(i,j)}\big).
\]
\end{definition}

\begin{definition}[Cofactor]\label{cofactor}
The \emph{cofactor} of \(M\) at \((i,j)\) is
\[
C_{i,j} := (-1)^{i+j} \det\big(M^{(i,j)}\big).
\]
\end{definition}

\begin{lemma}[Cofactor Expansion]\label{cofactorexpansion}
For any fixed row \(i\),
\[
\det(M) = \sum_{j=1}^n M_{i,j} \, C_{i,j}.
\]
An analogous formula holds for expansion along any column.
\end{lemma}

\begin{definition}
    [Adjucate]\label{adjucate}
    Let \(M \in \mathbb{R}^{n \times n}\). The \emph{adjugate} of \(M\), denoted \(\operatorname{adj}(M)\), is defined as the transpose of the cofactor matrix:
\[
\operatorname{adj}(M)_{ij} := (-1)^{i+j} \det\big(M^{(j,i)}\big),
\]
where \(M^{(j,i)}\) is the matrix obtained by deleting the \(j\)-th row and \(i\)-th column of \(M\).
\end{definition}
The adjugate satisfies the fundamental identity:
\begin{lemma}
\label{Adjdetidentity}
For any \(M \in \mathbb{R}^{n \times n}\),
\[
M \cdot \operatorname{adj}(M) = \operatorname{adj}(M)\cdot M = \det(M)\, \mathbb{I}_n,
\]
where $\mathbb{I}_n$ is the identity matric
\end{lemma}

\begin{theorem}[Cramer's Rule]\label{cramers}
Let \(M \in \mathbb{R}^{n \times n}\) be invertible, and let \(b \in \mathbb{R}^n\). Consider the linear system
\[
M x = b.
\]
For each \(i \in [n]\), let \(M^{(i)}\) denote the matrix obtained by replacing the \(i\)-th column of \(M\) with the vector \(b\). Then the unique solution \(x \in \mathbb{R}^n\) satisfies
\[
x_i = \frac{\det\big(M^{(i)}\big)}{\det(M)}.
\]
\end{theorem}

 \begin{theorem}[Matrix Tree Theorem]\cite{tutte2001graph}
\label{matrixtreethm}
Let \(G\) be a directed graph with out-Laplacian \(L^{\mathrm{out}}\).  
Then for any \(r \in V\),
\[
\det\big(L^{\mathrm{out}^{(r,r)}}\big)
=
\#\{\text{in-arborescences of } G \text{ rooted at } r\}.
\]
where $L^{\mathrm{out}^{(r,r)}}$ is the submatrix obtained by removing the row and column indexed by $r$.
\end{theorem}
 \subsection{Concentration Inequalities}
 \begin{theorem}[Holder's Inequality]\label{holder's}
Let $X_1,\dots,X_k$  nonnegative random variables , and let $p_1,\dots,p_k \geq 1$  be exponents with $\sum_i 1/p_i = 1$. Then
\[
\mathbb{E}\left[\prod_{i=1}^k X_i\right]
\leq \prod_{i=1}^k \left(\mathbb{E}[X_i^{p_i}]\right)^{1/p_i}.
\]
\end{theorem}

 \begin{theorem}[Scalar Bonami-Beckner hypercontractivity]\cite{Bonami, Beckner}\label{scalar_hypercont}
    Let $g: \{-1,1\}^n\rightarrow \mathbb{R}$ be a real valued function of Fourier degree at most $t$. Then for every $q \geq 2$, 
    \[
    \mathbb{E}[|g(x)|^q]\leq (q-1)^{qt/2}\cdot \mathbb{E}[|g(x)|^2]^{q/2}.
    \]
 \end{theorem}
 \begin{theorem}[Bonami-Beckner hypercontractivity for vector spaces]
     \label{vec:bonami-beckner}
     Let $F: \{-1,1\}^n\rightarrow \mathbb{C}^m$ be a function of Fourier degree at most $t$. Then, for every even $q\geq 2$, 
     \[
     (\mathbb{E}[||F||_2^q])\leq (q-1)^{qt/2}\cdot  (\mathbb{E}[||F||_2^2])^{q/2}.
     \]
 \end{theorem}
\begin{proof}

Let 
$F(x)= (F_1(x), \dots F_m(x))$, where each $F_i(x): \{-1,1\}^n \rightarrow \mathbb{C}$ has degree at most $t$.
We can write each coordinate $F_i(x)= a_i(x)+ i \cdot b_i(x)$ where $a_i, b_i: \{-1,1\}^n\rightarrow \mathbb{R}$. 
Now define 
\[
f(x)= (a_1(x), b_1(x), \dots a_m(x), b_m(x))\in \mathbb{R}^{2m}.
\]
Each coordinate in $f$ has Fourier degree at most $t$, and $||F(x)||_2=||f(x)||_2$.
Thus it suffices to prove the theorem  for $\mathbb R^d$-valued functions with coordinatewise degree at most $t$, where $d=2m$.
\begin{align*}
  ||f(x)||_2^{q} 
  = \Big(\sum_{i=1}^{d}|f_i(x)|^2\Big)^{q/2} 
  = \sum_{\alpha_1+ \dots +\alpha_d=\frac{q}{2} }{\frac{q}{2}\choose \alpha_1, \dots \alpha_d}\prod_{i=1}^{d}|f_{i}(x)|^{2 \alpha_i}.
\end{align*}

Using linearity of expectation,
\[
\mathbb{E}[||f(x)||_2^q]= \sum_{\alpha_1+ \dots +\alpha_d=\frac{q}{2} }{\frac{q}{2}\choose \alpha_1, \dots \alpha_d}\mathbb{E}\Big[\prod_{i=1}^{d}|f_{i}(x)|^{2 \alpha_i}\Big].
\]
Now fix $\alpha_1, \dots \alpha_d$ such that $\alpha_1, +\dots +\alpha_d=\frac{q}{2}$.
Let $\frac{1}{p_i}= \frac{2\alpha_i}{q}$ and observe that $\sum_{i}\frac{1}{p_i}=1$.
Holder's inequality (Theorem \ref{holder's}) implies that \[\mathbb{E}\left[\prod_{i=1}^{d}|f_i(x)|^{2 \alpha_i}\right]\leq \prod_{i=1}^{d}\left(\mathbb{E}\left[|f_i(x)|^{2 \alpha_i\cdot p_i}\right]\right)^{1/p_i}=\prod_{i=1}^{d}\left(\mathbb{E}\left[|f_i(x)|^{q}\right]\right)^{2\alpha_i/q}.\]
Theorem \ref{scalar_hypercont} implies that for all $i \in [d]$, \[\mathbb{E}\left[|f_i(x)|^{q}\right]\leq (q-1)^{tq/2}\cdot  \left(\mathbb{E}[|f_i(x)|^2]\right)^{q/2}.\] 
Thus, 
\[
\left(\mathbb{E}\left[|f_i(x)|^{q}\right]\right)^{2\alpha_i/q}\leq (q-1)^{t\alpha_i}\cdot  \left(\mathbb{E}[|f_i(x)|^2]\right)^{\alpha_i}.
\]
Which implies that (since $\sum_i\alpha_i=\frac{q}{2}$),
 \[\mathbb{E}\left[\prod_{i=1}^{d}|f_i(x)|^{2 \alpha_i}\right]\leq (q-1)^{t \sum_i\alpha_i}\cdot \prod_{i=1}^{d}\left(\mathbb{E}[|f_i(x)|^2]\right)^{\alpha_i}=(q-1)^{qt/2}\cdot \prod_{i=1}^{d}\left(\mathbb{E}[|f_i(x)|^2]\right)^{\alpha_i}.
\]
Thus, 
\begin{align*}
  \mathbb{E}[||f(x)||_2^q] &= \sum_{\alpha_1+ \dots +\alpha_d=\frac{q}{2} }{\frac{q}{2}\choose \alpha_1, \dots \alpha_d} (q-1)^{qt/2}\cdot\prod_{i=1}^{d}\left(\mathbb{E}[|f_i(x)|^2]\right)^{\alpha_i} \\
  &=  (q-1)^{qt/2}\cdot \sum_{\alpha_1+ \dots +\alpha_d=\frac{q}{2} }{\frac{q}{2}\choose \alpha_1, \dots \alpha_d}\prod_{i=1}^{d}\left(\mathbb{E}[|f_i(x)|^2]\right)^{\alpha_i}\\
  &= (q-1)^{t q/2}\cdot \Big(\sum_{i=1}^{d}\mathbb{E}[|f_i(x)|^2]\Big)^{\frac{q}{2}}.
\end{align*}
Observe that $\sum_{i=1}^{d}\mathbb{E}[|f_i(x)|^2]= \mathbb{E}\sum_{i=1}^{d}|f_i(x)|^2= \mathbb{E}[||f(x)||_2^2]$, we conclude the theorem.

\end{proof}
\subsection{Quantum Query Lower Bounds}
\begin{theorem} 
\label{thm:bbbv97}
Let $A$ denote a $T$-query quantum algorithm. Let $O: [N] \to [M]$ be an oracle and let $\ket{\phi_j}$ denote the intermediate state of $A$ right before the $j$'th query to $O$.  Let $F \subseteq [N]$ be a set of query inputs, and let $\widetilde{O}$ denote an oracle that is identical to $O$ on all inputs $[N]\setminus F$. Then
\[
    \Big \| \ket{\phi_{T+1}} - \ket{\widetilde{\phi}_{T+1}} \Big \| \leq \sqrt{T \sum_{j=1}^T\sum_{i \in F}  W_i(\ket{\phi_j^O})},
\]
where $\ket{\phi_T},\ket{\widetilde{\phi}_T}$ are the final states of the algorithm $A$ when querying oracles $O$ and $\widetilde{O}$, respectively.  
\end{theorem}
\begin{theorem}[Polynomial representation of quantum query algorithms]
\label{thm:adaptive-polynomial-representation}
Let $x=(x_1,\ldots,x_N)\in\{-1,1\}^N$ be an oracle written in phase-query
form. Let $\mathcal A$ be an arbitrary quantum algorithm making $r$ queries to the oracle.  Let $|\psi^x\rangle$ denote the final pure state of $\mathcal A$
before measurement. Then there exist oracle-independent vectors
$\{|v_S\rangle : S\subseteq [N], |S|\le r\}$ such that
\[
    |\psi^x\rangle
    =
    \sum_{\substack{S\subseteq [N]\\ |S|\le r}}
    \chi_S(x)\,|v_S\rangle,
\]
where
\[
    \chi_S(x):=\prod_{i\in S} x_i .
\]
In particular, for every oracle-independent projector $\Pi_e$, the
vector-valued function
\[
    F_e(x):=\Pi_e|\psi^x\rangle
\]
has Fourier degree at most $r$, and the corresponding event probability
\[
    p_e(x):=\|\Pi_e|\psi^x\rangle\|_2^2
\]
is a real multilinear polynomial of degree at most $2r$.
\end{theorem}
\subsection*{Parallel query algorithms}
A parallel query quantum algorithm with oracle access to $O$ is modeled by two unitary $U_1, U_2$: it initializes a pre-query quantum state $U_1 \ket {0}$, makes the $T$ parallel queries $U_O^{\otimes T}$ on the first $T$ registers (where $U_O : \ket{i, b} \to (-1)^{O(i) \cdot b} \ket{i, b}$) and finally applies $U_2$ and a standard basis measurement.

\subsection*{Simulation conjecture}

We first state the simulation conjecture formally, as first formalized in~\cite{AA08}:
\begin{conjecture}[Simulation conjecture]\label{thm:simulationconj}
    Let $\mathcal{A}$ be a quantum algorithm making $T$ queries to a boolean input $x=(x_1, \dots x_N)$, $\epsilon, \delta > 0$. There is a deterministic classical algorithm that makes $\text{poly}(T, 1/\epsilon, 1/\delta)$ queries to the $x_i's$ and approximates $\mathcal{A}'s$ acceptance probability with an additive error $\epsilon$ on a $(1-\delta)$ fraction of inputs. 
\end{conjecture}

\subsection*{Dense distributions and related conjectures}

\begin{definition}[$(1-\delta)$-dense oracle distributions]
Let $O_X$ be a distribution over Boolean functions $O:[N]\to\{0,1\}$.
We say that $O_X$ is \emph{$(1-\delta)$-dense} if for every $k\in\{0,1,\dots,N\}$, every choice of distinct inputs $x_1<\cdots<x_k\in[N]$, and every assignment $r_1,\dots,r_k\in\{0,1\}$,
\[
\Pr_{O\leftarrow O_X}\bigl[O(x_j)=r_j\ \text{for all }j\in[k]\bigr]\ \le\ 2^{-(1-\delta)k}.
\]
Equivalently, for every subset $S\subseteq[N]$, the restriction $O_X|_S$ (interpreted as a binary string) has min-entropy at least $(1-\delta)|S|$, i.e.,
\[
H_\infty\!\bigl((O|_S)_{O\leftarrow O_X}\bigr)\ \ge\ (1-\delta)|S|.
\]
\end{definition}

Next, we give the main conjecture related to the indistinguishability between dense and uniform distributions.
\begin{conjecture}[Dense indistinguishability conjecture]
    \label{thm:main}
    Let $O_X$ be a $(1-\delta)$-dense distribution over $O:[N] \to \{0,1\}$ and $O_U$ be uniform.
    There exist absolute constants $C, a \geq 0, b > 0$ such that for all $T$ query quantum algorithms $\As$, 
    \[
    \Big|\Pr_{O \sim O_X}[\mathcal{A}^O \rightarrow 1]-\Pr_{O \sim O_U}[\mathcal{A}^O \rightarrow 1]\Big|\leq C \cdot T^a \cdot \delta^b.
    \]
\end{conjecture}

\subsection*{Boolean functions}
\begin{definition}[$(t, \delta)$-smooth]
    Let $t > 0$ be an integer and $\delta > 0$.
    We say a boolean function $f:\{-1,1\}^n \to [0,1]$ is $(t,  \delta)$-smooth if for any $|J| \leq t$ and $x_J \in \{-1, 1\}^J$, it has that 
    \begin{align*}
        \left| \mathbb{E}_{x_{\bar{J}}}[f(x_J, x_{\bar{J}})]  - \mathbb{E}_{x}[f(x)]\right| \leq \delta.
    \end{align*}
\end{definition}

\section{Simulation Conjecture and Dense Indistinguishability Conjecture}
\label{sec:dense_implies_simu}

In this section, we recall the implication from dense indistinguishability
to the simulation conjecture established in~\cite{guo2021unifying}, and
explain how to apply this implication to the algorithm classes considered
in this paper. In \Cref{subsec:boolean-presampling}, we prove the converse
of the Boolean presampling-to-simulation implication from that work,
thereby establishing the equivalence of Boolean presampling and
simulation. The arguments below establish dense indistinguishability as
a sufficient condition for these equivalent conjectures.

We first state a conjecture which is equivalent to Conjecture~\ref{thm:simulationconj}.

\begin{conjecture}
\label{thm:varA}
Let $\mathcal{A}$ denote a $T$ query quantum algorithm which makes queries to a uniform random oracle $O \sim O_U$, and let $\text{Var}(\As)$ denote 
    \[
    \text{Var}(\As):=\mathbb{E}_{O \sim O_U}\Big(\Pr[\mathcal{A}^{O}=1]\Big)^2-\Big(\mathbb{E}_{O \sim O_U}\Pr[\mathcal{A}^{O}=1]\Big)^2,
    \]
    where the probabilities are over randomness of the algorithm only. Then, for any quantum algorithm with $T$ queries, there exists a subset $I \subseteq [N]$ of size at most $\text{poly}(T/\text{Var}(\mathcal{A}))$ and a string $W \in \{0,1\}^I$ such that,  
    \[
    \Big|\Pr_{O \sim O^{W}_{U}}[\mathcal{A}^O\rightarrow 1]-\Pr_{O \sim O_{U}}[\mathcal{A}^O\rightarrow 1]\Big|\geq \Omega(\text{Var}(\mathcal{A})).
    \]
    Here $O_U^W$ is fixed on $I$ ($O_I = W$) and uniform everywhere else. 
\end{conjecture}

Now we will state a conjecture which is implied by Conjecture \ref{thm:varA} (indeed, they claim that both are equivalent in~\cite{guo2021unifying}). 
\begin{conjecture}\label{restatementvarconj}
 There is a constant $C>0$. For the acceptance probability $f: \{-1,1\}^n \rightarrow [0,1]$ of any quantum algorithm $\As$ that makes $T$ queries to a Boolean string $x=(x_1, \dots x_N)$ and any $\delta>0$, if $f$ is $((T/\delta)^C, (\delta/T)^C)$ smooth function, then $\text{Var}(f) = \text{Var}(\As) \leq \delta$.
\end{conjecture}
\begin{lemma}
Conjecture \ref{thm:varA} implies Conjecture \ref{restatementvarconj}. Furthermore, if the former holds only for all quantum algorithms $\As \in \mathbb{A}$, then so does the latter.
\end{lemma}
The proof is missing in~\cite{guo2021unifying}, thus we give it here.
\begin{proof}
    For any quantum algorithm $\As$ with variance $\sigma$, let $I, W$ be defined in Conjecture~\ref{thm:varA}. Furthermore, assume $W$ fixes $p(T / \sigma)$ bits and the distinguishing advantage is $q(\sigma)$. Both $p, q$ are monotonic functions. 

    Let $C$ be a constant that to be chosen later. Now suppose $f$ is $((T/\delta)^C, (\delta/T)^C)$ smooth but $\text{Var}(\As) = \sigma > \delta$. By the monotonicity of $p, q$, we have $|I| = p(T/\sigma) < p(T/\delta)$ and $q(\sigma) > q(\delta)$. By appropriately choosing the constant $C$ (which depends on $p, q$), we have
    \begin{align*}
        |I| & < p(T/\delta) \leq (T/\delta)^C \\
        q(\sigma) & > q(\delta) \geq (\delta/T)^C.
    \end{align*}
    Thus, it contradicts with $f$ being $((T/\delta)^C, (\delta/T)^C)$ smooth.
\end{proof}

Now, we state a conjecture which is equivalent to Conjecture \ref{thm:simulationconj}. Indeed, we only need the fact that Conjecture \ref{restatesimconj} implies Conjecture \ref{thm:simulationconj}.
\begin{conjecture}
    \label{restatesimconj}
        There is a constant $C'>0$. For the acceptance probability $f: \{-1,1\}^n\rightarrow [0,1]$ of any quantum algorithm that makes $T$ queries to a boolean string $x=(x_1, \dots x_N)$ and any $\epsilon>0$, there is a decision tree $g: \{-1,1\}^N\rightarrow [0,1]$ of depth $(T/\epsilon)^C$ such that $||g-f||_2\leq \epsilon$.
\end{conjecture}

We have the following theorem, which is proved in~\cite{guo2021unifying}. Crucially, we observe that the reductions used to establish these implications do not alter the underlying algorithm.
\begin{theorem}[Dense indistinguishability conjecture implies simulation conjecture]  \label{thm:main_implication}
Conjecture \ref{thm:main} implies Conjecture \ref{thm:varA}, which implies Conjecture \ref{restatementvarconj}, which implies Conjecture \ref{restatesimconj}, which implies Conjecture \ref{thm:simulationconj}. 

Furthermore, if Conjecture \ref{thm:main} holds for all algorithms in $\mathbb{A}$, so does Conjecture \ref{thm:simulationconj}.
\end{theorem}

\subsection{Boolean presampling and its relation to simulation}
\label{subsec:boolean-presampling}

We show that the Boolean specialization of the presampling conjecture
(Conjecture~4 of~\cite{guo2021unifying}) follows from the simulation
conjecture (Conjecture~\ref{conj:sim_intro}). Together with the reverse
implication proved in Theorem~4 of~\cite{guo2021unifying}, this establishes
the equivalence of Boolean presampling and simulation. Dense
indistinguishability (Conjecture~\ref{conj:dense_intro}) implies both by
\Cref{thm:main_implication}.

For completeness, we recall the presampling conjecture of
Guo, Li, Liu, and Zhang~\cite{guo2021unifying}, specialized to Boolean
oracles. A distribution over Boolean oracles
$O:[N]\to\{0,1\}$ is called $P$-bit-fixing if there is a set
$I\subseteq[N]$ with $|I|\le P$ and an assignment
$W\in\{0,1\}^{I}$ such that the oracle is fixed to $W$ on $I$
and is uniform on all remaining coordinates. We allow convex
combinations of such distributions.

\begin{conjecture}[Boolean presampling conjecture; Conjecture~4 of~\cite{guo2021unifying}]
\label{conj:boolean-presampling}
There exist a universal constant $C>0$ and a function
$h:\mathbb N\to\mathbb R_{+}$ such that the following holds.
For every $N,S\in\mathbb N$, let $O_U$ be the uniform distribution
over Boolean oracles $O:[N]\to\{0,1\}$, and let
$
F:\{0,1\}^{N}\to\{0,1\}^{S}
$
be an arbitrary function. For every integer $P\ge1$, there exists a family
$
\{O_{Y_z}\}_{z\in\{0,1\}^{S}}
$
such that each $O_{Y_z}$ is a convex combination of
$P$-bit-fixing oracle distributions and, simultaneously for every
integer $T\ge1$ and every quantum
algorithm $\mathcal A$ that takes an $S$-bit input and
makes at most $T$ quantum queries to its oracle,
\[
\Bigl|
\Pr_{O\sim O_U}
   \bigl[\mathcal A^{O}(F(O))\to1\bigr] -
\Pr_{\substack{Z\sim F(O_U)\\
                \widetilde O\sim O_{Y_Z}}}
   \bigl[\mathcal A^{\widetilde O}(Z)\to1\bigr]
\Bigr|
\le
h(S)\frac{T}{P^{C}}.
\]
\end{conjecture}

The function $h$ is independent of $N,F,P,T$, and $\mathcal A$;
its growth in $S$ is unrestricted.
Here the family $\{O_{Y_z}\}_{z\in\{0,1\}^{S}}$ may depend on
$F$ and $P$, but not on $T$ or the distinguisher $\mathcal A$.
Thus one family of bit-fixing mixtures must simultaneously fool all
quantum distinguishers, with the stated bound for each query budget $T$.

We now show that this conjecture follows from the simulation
conjecture (Conjecture \ref{conj:sim_intro}).

\begin{theorem}[Simulation implies Boolean presampling]
\label{thm:simulation-implies-bit-fixing}
If Conjecture~\ref{conj:sim_intro} holds, then
Conjecture~\ref{conj:boolean-presampling} holds with $h(S)=K2^{S/2}$
for some absolute constant $K\ge1$.
\end{theorem}
The dependence on the advice length $S$ is exponential. The theorem
therefore establishes the presampling conjecture with unrestricted
growth of $h$, without providing a polynomial dependence on $S$.
\begin{proof}
We first use the simulation conjecture to obtain a mean-square
approximation. For $0<\rho<1$, apply Conjecture~\ref{conj:sim_intro}
with additive error $\epsilon=\rho/\sqrt{2}$ and failure probability
$\delta=\rho^2/2$, and clip the simulator's output to $[0,1]$.
This does not increase the query depth or the approximation error.
Then there exist
absolute constants $K\ge1$ and $k\ge1$ such that the acceptance
probability
$
f(O):=\Pr[\mathcal A^O\to1]
$
of every $t$-query quantum algorithm has a deterministic classical
decision-tree approximation $\hat{g}$, with values in $[0,1]$, of depth at most
$
\left(\frac{Kt}{\rho}\right)^k
$
such that
$$
\begin{aligned}
\mathbb E_{O\sim O_U}[(f(O)-\hat{g}(O))^2]
&\le
\frac{\rho^2}{2}
+
\Pr_{O\sim O_U}
\left[
|f(O)-\hat{g}(O)|>\frac{\rho}{\sqrt{2}}
\right]\\
&\le
\rho^2.
\end{aligned}
$$

Set
$
c:=\frac{1}{2k},
t:=\lfloor P^c\rfloor,
\rho:=KP^{-c}.
$
If $\rho\ge1$, choose $O_{Y_z}=O_U$ for every $z$. The desired bound
is then at least $2^{S/2}T\rho\ge1$, so it holds automatically.
Henceforth assume $\rho<1$. Then
$
\left(\frac{Kt}{\rho}\right)^k
=
(tP^c)^k
\le
P.
$

For each $z\in\{0,1\}^S$, let
$
E_z:=\{O:F(O)=z\},
p_z:=\Pr_{O\sim O_U}[O\in E_z].
$
Fix $z$ with $p_z>0$. We first construct $O_{Y_z}$.

Let $\mathcal A$ be any quantum algorithm making at most $t$
queries, with classical input $z$, and let
$
f(O):=\Pr[\mathcal A^O(z)\to1].
$
Let $\mathcal T$ denote a decision tree of depth at most $P$
computing $\hat g$, which approximates $f$ with mean-square error
at most $\rho^2$. We remove repeated-coordinate queries and
unreachable leaves, without increasing its depth.
Let $\mathcal L$ denote the set of leaves of $\mathcal{T}$. For each leaf
$\tau\in\mathcal L$, let $I_\tau\subseteq[N]$ be the set of
coordinates queried along the path to $\tau$, and let
$W_\tau\in\{0,1\}^{I_\tau}$ be the corresponding assignment of
answers. Since $\mathcal T$ has depth at most $P$,
$
|I_\tau|\le P.
$
Moreover, the event that a uniform oracle reaches the leaf $\tau$
is exactly
$$
L_\tau
:=
\{O:O|_{I_\tau}=W_\tau\}.
$$

Define $g(O)=g_\tau$ for $O\in L_\tau$, where
\[
g_\tau:=\mathbb E_{O\sim O_U}[f(O)\mid L_\tau].
\]
Relabeling each leaf by this conditional mean does not increase the
mean-square error. Indeed, $\hat g$ is constant on each leaf and
$\mathbb E[f-g\mid L_\tau]=0$, so the conditional-expectation
projection identity gives
\[
\mathbb E_{O\sim O_U}[(f-\hat g)^2]
=\mathbb E_{O\sim O_U}[(f-g)^2]
+\mathbb E_{O\sim O_U}[(g-\hat g)^2].
\]
Consequently, $\mathbb E_{O\sim O_U}[(f-g)^2]\le\rho^2$.

Finally, because $O_U$ is uniform and
$L_\tau=\{O:O|_{I_\tau}=W_\tau\}$, the conditional distribution
$O_U\mid L_\tau$ is precisely $O_U^{W_\tau}$. Hence
$$
g_\tau
=
\mathbb E_{O\sim O_U^{W_\tau}}
[f(O)].
$$
Since $|I_\tau|\le P$, $O_U^{W_\tau}$ is a
$P$-bit-fixing distribution. Thus every leaf value $g_\tau$ is
exactly the acceptance probability of $\mathcal A$ under a
$P$-bit-fixing distribution.

By Cauchy--Schwarz,
$$
\begin{aligned}
\mathbb E_{O\sim O_U\mid E_z}[f(O)]
&\le
\mathbb E_{O\sim O_U\mid E_z}[g(O)]
+
\mathbb E_{O\sim O_U\mid E_z}[|f(O)-g(O)|]\\
&=\sum_{\tau\in \mathcal{L}}\Pr[L_{\tau}|E_z]\cdot g_{\tau}
+
\mathbb E_{O\sim O_U\mid E_z}[|f(O)-g(O)|]\\
&\le
\max_\tau g_\tau
+
\frac{
\mathbb E_{O\sim O_U}
\left[
\mathbf 1_{E_z}(O)|f(O)-g(O)|
\right]
}{p_z}\\
&\le
\max_\tau g_\tau
+
\frac{\rho}{\sqrt{p_z}}.
\end{aligned}
$$
Consequently, for every at-most-$t$-query distinguisher there exists
a $P$-bit-fixing distribution $D$ such that
$$
\mathbb E_{O\sim O_U\mid E_z}[f(O)]
-
\mathbb E_{O\sim D}[f(O)]
\le
\frac{\rho}{\sqrt{p_z}}.
$$

The distribution $D$ obtained above may depend on the distinguisher.
We use minimax to obtain one mixture that works simultaneously for
all distinguishers.
Fix any finite collection of at-most-$t$-query
algorithms with acceptance probabilities $f_1,\ldots,f_m$.
For any probability vector $\lambda=(\lambda_1,\ldots,\lambda_m)$,
the function
$$
f_\lambda(O)
=
\sum_{j=1}^m\lambda_j f_j(O)
$$
is itself the acceptance probability of an at-most-$t$-query
quantum algorithm: choose $j$ according to $\lambda$ and run the
$j$th algorithm. This random choice adds no oracle queries.
Therefore, for every $\lambda$, there exists a
$P$-bit-fixing distribution $D$ such that
$$
\sum_{j=1}^m\lambda_j
\left(
\mathbb E_{O\sim O_U\mid E_z}[f_j(O)]
-
\mathbb E_{O\sim D}[f_j(O)]
\right)
\le
\frac{\rho}{\sqrt{p_z}}:= \varepsilon_z.
$$
For all $j \in [m]$,
let $L(f_j,D):= E_{O\sim O_U\mid E_z}[f_j(O)]
-
\mathbb E_{O\sim D}[f_j(O)]$. 
The minimax theorem says
\[
\underbrace{\max_\lambda\min_D
\sum_j\lambda_j L(f_j,D)}_{\le\varepsilon_z}
=
\min_\theta\max_j
\sum_D\theta_D L(f_j,D),
\]where \(\theta\) is a probability distribution over the available bit-fixing distributions. Thus there is one choice of \(\theta\) such that
\[
\sum_D\theta_D L(f_j,D)\le\varepsilon_z
\qquad\text{for every }j.
\]
Define the mixture \(O_{Y_z}:=\sum_D\theta_DD\). Since expectations are linear,
\[
L(f_j,O_{Y_z})=\sum_D\theta_D L(f_j,D).
\]
Thus, there is a convex combination
$O_{Y_z}$ of $P$-bit-fixing distributions such that
$$
\mathbb E_{O\sim O_U\mid E_z}[f_j(O)]
-
\mathbb E_{O\sim O_{Y_z}}[f_j(O)]
\le
\frac{\rho}{\sqrt{p_z}}
$$
for every $j$.

For fixed $N$ and $P$, there are only finitely many
$P$-bit-fixing distributions. Their convex hull is therefore compact.
Each distinguisher's inequality defines a
closed subset of this set. By the finite minimax argument, every
finite collection of these subsets has nonempty intersection.
The finite intersection property of compact sets therefore gives
a single $O_{Y_z}$ satisfying the inequality for every
at-most-$t$-query distinguisher.

Applying the same inequality to the complementary distinguisher
gives
$$
\left|
\mathbb E_{O\sim O_U\mid E_z}[f(O)]
-
\mathbb E_{O\sim O_{Y_z}}[f(O)]
\right|
\le
\frac{\rho}{\sqrt{p_z}}
$$
simultaneously for every at-most-$t$-query distinguisher. If
$p_z=0$, set $O_{Y_z}:=O_U$.

We now compare the two experiments in
Conjecture~\ref{conj:boolean-presampling}.
Fix any algorithm $\mathcal A$ making at most $t$ queries, and let
$$
f_z(O):=\Pr[\mathcal A^O(z)\to1].
$$
Then
$$
\begin{aligned}
&\left|
\Pr_{O\sim O_U}
\left[
\mathcal A^O(F(O))\to1
\right]
-
\Pr_{\substack{Z\sim F(O_U)\\
\widetilde O\sim O_{Y_Z}}}
\left[
\mathcal A^{\widetilde O}(Z)\to1
\right]
\right|\\
&\qquad\le
\sum_{z:p_z>0}
p_z
\left|
\mathbb E_{O\sim O_U\mid E_z}[f_z(O)]
-
\mathbb E_{O\sim O_{Y_z}}[f_z(O)]
\right|\\
&\qquad\le
\rho\sum_{z:p_z>0}\sqrt{p_z}\\
&\qquad\le
\rho\sqrt{2^S\sum_zp_z}\\
&\qquad=
K2^{S/2}P^{-c}.
\end{aligned}
$$

Before the last Cauchy--Schwarz bound, the advantage is at most
$\rho\sum_{z:p_z>0}\sqrt{p_z}$. This sharper, distribution-dependent
bound explains the factor $2^{S/2}$: the sum is at most the square
root of the number of advice values with positive probability.

Therefore, whenever $1\le T\le t$,
$$
K2^{S/2}P^{-c}
\le
K2^{S/2}\frac{T}{P^c}.
$$
If $T>t=\lfloor P^c\rfloor$, then $T>P^c$, and hence
$$
K2^{S/2}\frac{T}{P^c}\ge1.
$$
The desired inequality is then automatic, since a distinguishing
advantage is always at most $1$.

Thus the same family $\{O_{Y_z}\}_z$ works simultaneously for every
$T$. Taking $C=c$ establishes Conjecture~\ref{conj:boolean-presampling}
with $h(S)=K2^{S/2}$, proving the theorem.
\end{proof}

Theorem~4 of Guo, Li, Liu, and Zhang~\cite{guo2021unifying} proves the
reverse implication: Boolean presampling implies the simulation
conjecture. Their argument uses only one bit of advice ($S=1$),
so unrestricted growth of $h(S)$ does not obstruct the equivalence.
Therefore,
\[
\text{Simulation}
\quad\Longleftrightarrow\quad
\text{Boolean presampling}.
\]
Combining this with Theorem \ref{thm:main_implication} gives
\[
\text{Dense indistinguishability}
\quad\Longrightarrow\quad
\text{Simulation}
\quad\Longleftrightarrow\quad
\text{Boolean presampling}.
\]

The two conjectures impose different distributional requirements.
Dense indistinguishability requires every dense oracle distribution to
be indistinguishable from uniform. Boolean presampling instead concerns
an ensemble of conditional distributions induced by classical advice
about a uniform oracle. Define
\[
p_z:=\Pr_{O\sim O_U}[F(O)=z],
\qquad
O_{X_z}:=O_U\mid[F(O)=z]
\quad\text{for }p_z>0.
\]
One way to prove the Boolean presampling conjecture is to decompose $O_{X_z}$ into a convex combination of bit-fixing dense distribution, and then use the dense indistinguishability conjecture to replace it with bit-fixing uniform distribution. However, since the stronger constraints on $O_{X_z}$ (like the condition that $\sum_z p_z O_{x_Z}$ is a uniform distribution), it does not necessarily go through the dense indistinguishability. 

\section{Simulation Theorem for Parallel Quantum Algorithms}

By \Cref{thm:main_implication}, we only need to prove the following theorem:
\begin{theorem}
\label{thm:main_parallel}
    Let $O_X$ be a $(1-\delta)$-dense distribution over $O:[N] \to \{0,1\}$ and $O_U$ be uniform.
    There exist absolute constants $C, a \geq 0, b > 0$ such that for all $T$ query parallel quantum algorithms $\As$, 
    \[
    \Big|\Pr_{O \sim O_X}[\mathcal{A}^O \rightarrow 1]-\Pr_{O \sim O_U}[\mathcal{A}^O \rightarrow 1]\Big|\leq C \cdot T^a \cdot \delta^b.
    \]
    More specifically, we show $ C \cdot T^a \cdot \delta^b = T^{1/2} \left(\frac{\ln 2}{2} \delta\right)^{1/4} + o(1)$, where $o(1)$ can be made arbitrarily small.
\end{theorem}

\subsection{BBBV}
We give a theorem that resembles~\cite{BBBV97} but give a better bound for parallel algorithms.
\begin{theorem}
\label{thm:parallelbbbv}
Let $A$ denote a $T$ parallel query quantum algorithm. Let $O: [N] \to [M]$ be an oracle and let $\ket{\phi}$ denote the intermediate state of $A$ right before making the parallel query to $O$.  Let $F \subseteq [N]$ be a set of query inputs, and let $\widetilde{O}$ denote an oracle that is identical to $O$ on all inputs $[N]\setminus F$. Then
\[
    \Big \| \ket{\psi} - \ket{\widetilde{\psi}} \Big \|_{\sf{Tr}} \leq \sqrt{ \sum_{i \in F} W_i},
\]
where $\ket{\psi},\ket{\widetilde{\psi}}$ are the final states of the algorithm $A$, respectively; and $W_i$ is defined as the probability that when measuring all $t$ query registers, one of the outcomes is equal to $i$.
\end{theorem}

\begin{proof}
    The pre-query state $\ket \phi$ is
    \begin{align*}
        \ket \phi = \sum_{i', z} \alpha_{i', z} \ket{i', z},
    \end{align*}
    where $i'$ consists of $t$ inputs and $z$ consists of everything else. 
    The query magnitude on $i$ is 
    \begin{align*}
        W_i = \sum_{\substack{i' = i'_1, \ldots, i'_t, z \\ \exists k,\, i'_k = i}} |\alpha_{i', z}|^2.
    \end{align*}

    It is easy to see that the post-query states (with oracle access to $O$ or $\widetilde{O}$) only differ on the basis states whose $i'$ consists of any input in $F$. Thus, the $\ell^2$ distance between the final states $\ket \psi$ and $\ket {\widetilde{\psi}}$ is at most $\sqrt{\sum_{i \in F} W_i}$; so is their trace distance. The proof naturally extends to queries with a control qubit.
\end{proof}

\subsection{\texorpdfstring{Good Coupling Implies \Cref{thm:main_parallel}}{Good Coupling Implies the Parallel Simulation Conjecture}}
\label{sec:good_coupling_implies_dense}

We first give the following coupling theorem that shows one can turn any $(1-\delta)$-dense distribution into a uniform distribution, and only flip each coordinate with a small probability. 
\begin{theorem}[A good coupling between dense and uniform distributionS] \label{thm:strongchangeoracle}
Let $O_X$ be any $(1-\delta)$-dense oracle distribution over functions
$O:[N]\to\{0,1\}$. Then there exists a (possibly randomized) algorithm $M^*$
such that, the output oracle distribution $M^*(O)$ (when $O \gets O_X$) is distributed
identically to the uniform oracle distribution over $\{0,1\}^{[N]}$.
Moreover, for every $i\in\{1,\dots,N\}$, $M^*$ flips the $i$-th coordinate with a small probability:
\[
\Pr\!\big[\,M^*(O)(i)\neq O(i)\,\big]
\;\le\;
\sqrt{\frac{\ln 2}{2}\,\delta}\;+\;o(1),
\]
where the probability is taken over the randomness of $O\leftarrow O_X$ and
of $M^*$, and where $o(1)$ denotes a term that can be made arbitrarily small.
\end{theorem}

\begin{lemma}
    \Cref{thm:strongchangeoracle} implies \Cref{thm:main_parallel}.
\end{lemma}
\begin{proof}
Let $O_X$ be a $(1-\delta)$-dense distribution over $O:[N] \to \{0,1\}$ and $O_U$ be uniform. Since $M^*$ changes $O_X$ to $O_U$, for each $O \in O_X$, let $\text{Flip}_{O, M^*}$ be the set of inputs where $M^*$ flips their outputs. 

By \Cref{thm:parallelbbbv} we have, 
\begin{align*}
   & \Big|\Pr_{O \sim O_X}[\mathcal{A}^O \rightarrow 1]-\Pr_{O \sim O_U}[\mathcal{A}^O \rightarrow 1]\Big|\\
   &=\Big|\Pr_{O \sim O_X}[\mathcal{A}^O \rightarrow 1]-\Pr_{O \sim O_X,M^*}[\mathcal{A}^{M^*(O)} \rightarrow 1]\Big|
\\& \leq \sum_{\text{Flip}_{O, M^*}} \Pr_{O \sim O_X}[\text{Flip}_{O, M^*}] \cdot \sqrt{\sum_{i \in \text{Flip}_{O, M^*}} W_i} \\&\leq \sqrt{\sum_{\text{Flip}_{O, M^*}} \Pr[\text{Flip}_{O, M^*}] \sum_{i \in \text{Flip}_{O, M^*}} W_i}.
\end{align*}
Here since we only consider algorithm $\As$ who makes parallel queries, $W_i$ is always oracle independent (and the following argument will fail if it depends on $O$); the second inequality comes from Jensen's inequality.

Finally,
\begin{align*}
    \sum_{\text{Flip}_{O, M^*}} \Pr[\text{Flip}_{O, M^*}] \sum_{i \in \text{Flip}_{O, M^*}} W_i & = \sum_{\text{Flip}_{O, M^*}} \Pr[\text{Flip}_{O, M^*}] \sum_{i} W_i \cdot \delta_{i \in \text{Flip}_{O, M^*}} \\
    & = \sum_{i} W_i \cdot \left( \sum_{\text{Flip}_{O, M^*}} \Pr[\text{Flip}_{O, M^*}] \cdot \delta_{i \in \text{Flip}_{O, M^*}} \right)\\
    & \leq \sum_i W_i \cdot \left( \sqrt{\frac{\ln 2}{2} \delta} + o(1) \right) \\
    & \leq T \cdot \left( \sqrt{\frac{\ln 2}{2} \delta} + o(1) \right).
\end{align*}
The second last inequality is due to \Cref{thm:strongchangeoracle} and the last one is because number of parallel queries is $T$. Thus, we conclude the proof. 
\end{proof}

\subsection{The Existence of a Good Coupling}

\label{sec:good_coupling_exits}

Recall that  $O_X$  is a $(1-\delta)$-dense distribution over oracles $O: [N]\rightarrow \{0,1\}$ and $O_U$ is the uniform distribution. 
We will show that it is possible to change a small fraction of the oracle \footnote{By a small fraction of the oracle, we mean a small fraction of the coordinates $i \in [N]$.} drawn from $O_X$ to make the distribution identical to the uniform distribution over oracles $O_U$. Here 'change' means flipping the output of the oracle at some inputs. 

We start by giving the following weaker theorem \Cref{changeoracle} comparing to \Cref{thm:strongchangeoracle}; based on which, we will prove \Cref{thm:strongchangeoracle}. The theorem gives an expected total modification, rather than per-coordinate. 

\begin{theorem} \label{changeoracle}
Let $O_X$ be any $(1-\delta)$-dense oracle distribution and $O_U$ be a uniform oracle distribution. There is a random algorithm $M$ that takes an oracle $O$ from $O_X$, 
chooses a set $\text{Flip}_{O, M} \subseteq [N]$, and flips the answer of $O$ on $\text{Flip}_{O, M}$ such that the output oracle distribution $M(O)$ (when $O \gets O_X$) is distributed
identically to the uniform oracle distribution $O_U$, where the randomness is taken over $O \sim O_X$ and the randomness of $M$. The expected size of $\text{Flip}_{O, M}$ is 
\begin{equation*}
	\mathbb{E}\left[|\text{Flip}_{O, M}|\right] \leq \sqrt{\frac{\ln 2}{2} \delta} \cdot N
\end{equation*}
where the randomness is taken over $O \sim O_X$ and the randomness of $M$. 
\end{theorem}

\begin{proof}
We view each of the oracles $O \sim O_X$ as a binary string of length $N$. Therefore, we can view the $(1-\delta)$ dense distribution over oracles, $O_X$, as $N$ joint binary random variables
$X_1, X_2, \cdots, X_N$. 

We will introduce some notation for simplicity: for any $0 \leq n \leq N$ and any $x \in \{0, 1\}^n$, let $p_x = \Pr[X_1 = x_1, \cdots, X_n = x_n]$. It is easy to see
that $p_{\emptyset} = 1$ and for any $x \in \{0,1\}^N$, $p_x \leq 2^{-(1-\delta) N}$ because the distribution $O_X$ is $(1-\delta)$ dense. Moreover, as an immediate consequence of its denseness, we have that for every $S \subseteq [N]$, the min-entropy $H_\infty(X_S) \geq (1-\delta) |S|$. 

Consider the following procedure which  generates any joint binary distribution: 
\begin{enumerate}
\item Let $p_{\emptyset} = 1$. 
\item For $n = 0, \cdots, N-1$, 
	\begin{itemize}
		\item  For any $x \in \{0, 1\}^{n}$,   choose $0 \leq q_{x, 0} \leq 1$ and $q_{x, 1} = 1 - q_{x, 0}$; 
		\item Let $p_{x||0} = p_x q_{x, 0}$ and  $p_{x||1} = p_x q_{x, 1}$; 
	\end{itemize}
\end{enumerate}
It is easy to see that $\{p_x\}$ for $x \in \{0, 1\}^n$ is a marginal distribution of $X_1, X_2, \cdots, X_n$. 

\begin{ob}
The sum of all the marginal distributions is $N$: 
\begin{equation*}
	\sum_{n=0}^{N-1} \sum_{x \in \{0,1\}^n } p_x =\sum_{n=0}^{N-1} 1 = N
\end{equation*}
\end{ob}

\begin{ob}
$X_1, X_2, \cdots, X_N$ has entropy at least $(1 - \delta) N$. 
\begin{equation*}
	H(X_1, X_2, \cdots, X_N) = - \sum_{x \in \{0,1\}^N } p_x \log_2(p_x) \geq  (1 - \delta) N  \cdot \sum_{x \in \{0,1\}^N } p_x =  (1 - \delta) N 
\end{equation*}
Besides, for any subset $S \subseteq [N]$, $H(X_S) \geq (1 - \delta) |S|$. 
\end{ob}

Given the above procedure to generate the distribution $X_1, X_2, \cdots, X_N$, the following randomized algorithm $M$  changes the random variables and 
the resulting distribution is a uniform distribution over $\{0,1\}^N$.
\begin{algorithm}\caption{Modification Algorithm $M$}\label{alg:ModificationalgM}
\begin{algorithmic}
\State \textbf{Input:} $O = (x_1,\dots,x_N) \sim O_X \subseteq \{0,1\}^N$. 
\For{$n = 0, \dots N-1$}
    \State Let $x^{(n)} := (x_1,\dots,x_n)$ be the length $n$  prefix of $O$. Note that $x^{(0)}= \varepsilon$, the empty string.
    \State Let $p_{x^{(n)}} := \Pr[(X_1,\dots X_n) = x^{(n)}]$.
    \State Let $b := x_{n+1}$
    \State Define
        \[q_{x^{(n)},b} := \Pr[X_{n+1} = b \mid (X_1,\dots X_n) = x^{(n)}]
    \]
    \State (Equivalently, $p_{x^{(n)}\| b} = p_{x^{(n)}} \cdot q_{x^{(n)},b}$)
    \If{$q_{x^{(n)},b} > \tfrac{1}{2}$} 
        \State Set  $\tilde{X}_{n+1}$ to $1-b$ with probability $\frac{q_{x^{(n)},b} - 1/2}{q_{x^{(n)},b}}$, otherwise set $\tilde{X}_{n+1}$ to $b$.
        \Else 
    \State Set $\tilde X_{n+1} \gets b$
 \EndIf
\EndFor
\State \textbf{Output:}  $\tilde{O} = (\tilde{X}_1, \dots \tilde{X}_N)$
\end{algorithmic}
\end{algorithm}

\begin{lemma}
Algorithm $M$ outputs the  uniform distribution $O_U$ over $\{0,1\}^N$. 
\end{lemma}
\begin{proof}
We prove that for all $n\in \{0, \dots N-1\}$, after round $n$, the modified random variable
$\tilde{X}_{n+1}$ is uniform and independent of the prefix
$X^{(n)} := (X_1,\dots,X_n)$. Fix any prefix $x \in \{0,1\}^n$. Recall that
\[
q_{x,b} := \Pr[X_{n+1} = b \mid X^{(n)} = x].
\]
Let $\tilde{X}_{n+1}$ denote the output of the $n^{th}$ iteration of Algorithm \ref{alg:ModificationalgM}.

Recall that $x_{n+1}=b$, and without loss of generality, we can assume that
For every $x$ and $b$ such that $q_{x, b} \geq 1/2$, the algorithm flips $b$ to $1-b$ with probability $\frac{q_{x,b}-1/2}{q_{x,b}}$. Hence, 
\begin{align*}
     \Pr[X^{(n)}=x \wedge \tilde{X}_{n+1}=b] = \Pr[X^{(n)}=x \wedge X_{n+1}=b]\cdot \frac{1/2}{q_{x,b}}
     = p_x \cdot \frac{1}{2}.
\end{align*}

This implies that $\Pr[X^{(n)}=x \wedge \tilde{X}_{n+1}=1-b]= p_x \cdot \frac{1}{2}$ as well. 
\noindent
In both cases,
\[
\Pr[\tilde{X}_{n+1} = b \wedge X^{(n)} = x] = \frac{1}{2} \cdot \Pr[X^{(n)}=x].
\]
Thus $\tilde{X}_{n+1}$ is uniform and independent of $X^{(n)}$.
Since each $\tilde{X}_i$ depends only on $X^{(i)}$, it follows inductively that
$(\tilde{X}_1,\dots,\tilde{X}_N)$ is uniform over $\{0,1\}^N$.
\end{proof}
\noindent Next, we will prove that \begin{equation*}
	\mathbb{E}\left[|\text{Flip}_{O, M}|\right] \leq \sqrt{\ln 2 / 2 \cdot \delta} \cdot N
\end{equation*}
where  the randomness is taken over $O \sim O_X$ and the randomness of $M$. We begin by making the following observations:

\begin{ob}
For all $n \in \{0, \dots N-1\}$, and for all prefixes $x \in \{0,1\}^n$ of length $n$, define $e_x = |q_{x, 0} - 1/2| = |q_{x, 1} - 1/2|$.  The expected number of modifications at step $n$ of Algorithm \ref{alg:ModificationalgM} is the following:
\begin{eqnarray*}
\sum_{x:  q_{x, b} > 1/2} \Pr\left[ X_{[n]} = x, X_{n+1} = b \right] \cdot \left( \frac{q_{x, b} - 1/2}{q_{x, b}} \right) = \sum_{x \in \{0,1\}^n } p_x e_x
\end{eqnarray*}
So the expected total number of modifications is $\sum_{n=0}^{N-1}  \sum_{x \in \{0,1\}^n } p_x e_x$. 
\end{ob}

\begin{lemma}
\begin{equation*}
	\sum_{n=0}^{N-1}  \sum_{x \in \{0,1\}^n }  p_x e^2_x \leq \frac{\ln 2}{2} \delta \cdot N
\end{equation*}
\end{lemma}
\begin{proof}
	Consider the difference between $H(X_1, X_2, \cdots, X_{n+1})$ and $H(X_1, \cdots, X_n)$, 
\begin{eqnarray*}
	H(X^{(n+1)}) - H(X^{(n)}) &=& \sum_{x} p_x \log_2(p_x) -  \sum_{x, b} p_x q_{x, b} \log_2(p_x q_{x, b})  \\
			&=& - \sum_{x, b} p_x \left( q_{x, 0} \log_2(q_{x, 0}) + q_{x, 1} \log_2(q_{x, 1})  \right)  = \sum_x p_x H_b(q_{x, 0})
\end{eqnarray*}
where $H_b$ is the binary entropy function $H_b(x) = 1 - \frac{1}{2 \ln 2} \sum_{n \geq 1} \frac{(1 - 2 x)^{2 n}}{n (2 n - 1)} \leq 1 - \frac{2}{\ln 2} |x - 1/2|^2$. 
It implies $ \sum_x p_x H_b(q_{x, 0}) \leq 1 - \sum_x p_x \cdot \frac{2}{\ln 2} e_x^2$. 
Summing over $n$, the LHS evaluates to  
\[
\sum_{i=0}^{N}H(X^{(i+1)})-H(X^{(i)})= H(X^{(N)})
\]
Thus, we have that 
\begin{eqnarray*}
	(1-\delta) N \leq H(X^{(N)}) = \sum_{n=0}^{N-1} \sum_{x \in \{0,1\}^n } p_x H_b(q_{x, 0}) \leq N - \sum_{n=0}^{N-1} \sum_{x \in \{0,1\}^n } p_x \cdot \frac{2}{\ln 2} e_x^2
\end{eqnarray*}
Our lemma follows. 
\end{proof}

By Cauchy-Schwarz inequality, we have 
\begin{eqnarray*}
	\left(\sum_{n=0}^{N-1}  \sum_{x \in \{0,1\}^n } p_x e_x\right)^2 &\leq& \left( \sum_n \sum_x (\sqrt{p_x})^2 \right) \cdot \left( \sum_n \sum_x (\sqrt{p_x} e_x)^2 \right)  \\
		&\leq&  \frac{\ln 2}{2} \delta N^2
\end{eqnarray*} 
which implies that 
\[
\sum_{n=0}^{N-1}  \sum_{x \in \{0,1\}^n } p_x e_x\leq N \sqrt{\frac{\ln 2}{2}\cdot \delta}
\]
Our theorem follows. 
\end{proof}

\begin{corollary} \label{changeoracleanyorder}
Let $O \sim O_X$, where $O_X$ is a $(1-\delta)$-dense oracle distribution.
Let $M_{\pi}$ be the same as Algorithm \ref{alg:ModificationalgM}, but it executes according to the ordering specified by the permutation $\pi$: $\pi(0), \pi(1), \cdots, \pi(N-1)$. 
Let $\text{Flip}_{O, M_{\pi}, j} \subseteq \{\pi(0), \pi(1), \cdots, \pi(j-1)\}$ be the set that $M_{\pi}$ modifies 
in the first $j$ rounds on input $O \sim O_X.$ We have for every permutation $\pi$,  for every $0 \leq j < N$, 
\begin{equation*}
	\mathbb{E}\left[ |\text{Flip}_{O, M_{\pi}, j} | \right] \leq \sqrt{\frac{\ln 2}{2} \delta} \cdot j
\end{equation*}
The randomness is taken over $O \sim O_X$ and the random coins of $M_{\pi}$. 

When $\pi$ is an identity permutation and $j = N - 1$, it is equivalent to {Theorem \ref{changeoracle}}. 
\end{corollary}
\begin{proof}
   Since $O_X$ is a $(1-\delta)$ dense distribution, by definition, for all subsets $S$ of size up to $N$, 
 the min entropy of $O_X$ restricted to coordinates in $S$ is lower bounded by $(1-\delta)\cdot |S|$.  Furthermore, since this is the only guarantee required for the proof of \Cref{changeoracle}, the same  proof works.
\end{proof}

Now we will prove a stronger version of Theorem \ref{changeoracle}, which says that, informally, there exists a modification algorithm which takes as input $O \sim O_X$, and outputs an oracle distribution according to $O_U$, with the stronger guarantee that \emph{all} coordinates $i \in [N]$ are flipped with small probability over the choice of $O \sim O_X$, and the randomness of the modification algorithm. 
\begin{theorem} [A stronger version of Theorem \ref{changeoracle}]  \label{strongchangeoracle}
For any $(1-\delta)$-dense oracle distribution $O_X$, there exists a randomized algorithm $M^*$ where $M^*$ such that if $O \sim O_X$, then $M^*(O)$ is distributed as $O_U$,  and for every $i \in \{1, \cdots, N\}$, 
$M^*$ modifies $O(i)$ with probability at most $\sqrt{\frac{\ln 2}{2} \cdot \delta} + o(1)$ where $o(1)$ can be any arbitrarily small number.  
\end{theorem}

Before proving the theorem, we present its symmetric counterpart, which transforms $O_U$ into $O_X$.
\begin{corollary}\label{uniformtodenseM}
      For all $(1-\delta)$ dense distributions $O_X$ over oracles $O:[N]\rightarrow \{0,1\}$,  there exists an algorithm $M_U^*$  such that if $O \sim O_U$ which is the uniform distribution over oracles $O:[N]\rightarrow \{0,1\}$, then $M_U^*(O)$ is distributed according to  $O_X$, and for every $i \in \{1 \dots N\}$, $M_U^*(O)$ modifies $O(i)$ with probability at most $\sqrt{\frac{\ln 2}{2}\cdot \delta}+o(1)$ where $o(1)$ can be any arbitrarily small number.   
\end{corollary}
\begin{proof}
    Let $c_{O, O'}$ be the probability that $O$ is sampled from $O_X$ and it is coupled/modified into $O'$ as in \Cref{strongchangeoracle}. Since it is symmetric, we can also define $M^*_U$ by $c_{O, O'}$; $c_{O, O'}$ defines the probability that $O'$ is sampled from $O_U$ and it is coupled/modified into $O$.

    Clearly, for every $i \in [N]$, 
    \begin{align*}
        \sqrt{\frac{\ln 2}{2} \delta} + o(1) \geq \Pr_{O\sim O_X}[M^* \text{ flips } O(i)] = \sum_{\substack{O, O'\\ O(i) \ne O'(i)}} c_{O, O'} = \Pr_{O'\sim O_U}[M_U^* \text{ flips } O'(i)].
    \end{align*}
\end{proof}

Finally, we prove \Cref{strongchangeoracle}.

\begin{proof}[Proof for \Cref{strongchangeoracle}]
Observe that for any distribution $D_{\Pi}$ over  all permutations of $\pi: [N]\rightarrow [N]$, if the modification algorithm $M$ (Algorithm \ref{alg:ModificationalgM}) picks each permutation $\pi$
according to the $D_\Pi$ and runs $M_\pi$ from Corollary \ref{changeoracleanyorder}, $M_{\pi: \pi \leftarrow D_{\Pi}}(O): O \sim O_X$ will always be the uniform oracle distribution $O_U$. Therefore, it is sufficient to find a ``good'' distribution 
over permutations, which we will denote by $D_{\Pi^*}$, such that the expected modification for each coordinate $i \in [N]$ is ``small''. Thus $M^*$ is defined as:
\begin{enumerate}
    \item Sample $\pi \sim D_{\Pi^*}$. 
    \item Run $M_{\pi}$ from Corollary \ref{changeoracleanyorder}.
\end{enumerate}
We will prove the existence of $D_{\Pi^*}$ by induction on the number of coordinates the oracle $O$ is supported on. 
\paragraph{Base Case:}
When $N = 1$, there is only one permutation defined on $1$ random variable, which is the identity permutation. Therefore, $M^*= M_{\mathbb{I}}$. Theorem \ref{changeoracle} then implies that, 
\[
\mathbb{E}_{O \sim O_X}\Big[\Pr_{M^*}[O(x_1)\in \text{Flip}_{O, M^*}]\Big]\leq \sqrt{\frac{\ln 2}{2}\cdot \delta}\cdot 1=\sqrt{\frac{\ln 2}{2}\cdot \delta}. 
\]
\paragraph{Induction Hypothesis: } 
Assume the claim holds for $N-1$ coordinates. . Therefore, for any $(1-\delta)$ dense distribution $O_X$   over oracles $O: [N-1]\rightarrow \{0,1\}$, there exists a randomized algorithm $M^*_{(N-1)}$ such that $M^*_{(N-1)}(O)$ is equivalent to the uniform oracle distribution and for every $i \in \{1, \dots N-1\}, $ $M^*_{(N-1)}$ modifies $O(i)$ with probability at most $c= \sqrt{\frac{\ln 2}{2} \delta} + (i-1) \varepsilon$.

\paragraph{Induction Step:} The following notation will be used in our analysis:
\begin{enumerate}
    \item For each $i \in [N]$, let $p_i$ denote the probability  that $i$ is the last coordinate that is iterated over. More concretely, 
    \[
    p_i = \Pr_{\pi\sim D_{\Pi}}[\pi(N)=i].
    \]
    \item For each $i\in[N]$, let $M^{(i)}$ be the algorithm obtained by placing coordinate $i$ last and using the induction hypothesis  to choose a randomized order on the remaining coordinates $[N]\setminus\{i\}$. Equivalently, $M^{(i)}$ is defined as the algorithm which places coordinate $i$ last, and iterates over the remaining coordinates $[N]\setminus\{i\}$ according to the ordering specified by $M^*_{N-1}$ in the induction hypothesis. More concretely,  $M^{(i)}$ is defined as follows:
    \begin{itemize}
        \item By the induction hypothesis applied to the marginal distribution
\(O_X|_{[N]\setminus\{i\}}\), there exists a distribution
\(D_{\Pi}^{(-i)}\) over permutations of \([N]\setminus\{i\}\) with the
desired coordinate-wise flip guarantees for this marginal.
\item Sample
\(\sigma \sim D_{\Pi}^{(-i)}\), form the full permutation
\[
\pi = (\sigma(1),\ldots,\sigma(N-1),i),
\]
and run the full fixed-order modification algorithm \(M_\pi\) on the
original \(N\)-coordinate oracle \(O \sim O_X\).
    \end{itemize}
    For $i,j\in[N]$, define
\[
q_{i,j}
:= \Pr_{O\sim O_X,\;M^{(i)}}[O(j)\in \mathrm{Flip}_{M^{(i)},O}].
\]

Then the induction hypothesis implies that for every $i\neq j$,
\[
q_{i,j}\le c.
\]
Also, by Corollary~\ref{changeoracleanyorder}, the modification algorithm according to any ordering will have the expected total modification being $\sqrt{\frac{\ln 2}{2} \delta} N$. Therefore,  
\[
\sum_{j=1}^N q_{i,j}\le cN
\qquad\text{for every }i\in[N].
\]
Now, define $M^*$ to be the following algorithm:
\begin{enumerate}
    \item Choose $i$ according to $\{p_i\}_{i \in [N]}$. 
    \item Run $M^{(i)}$.
\end{enumerate}
    \end{enumerate}
Let $c' = c + \varepsilon= \sqrt{\frac{\ln 2}{2} \delta} + i \varepsilon$. 
Then, the expected modification of $O(i)$, for any $i \in [N]$ is the following:
		\begin{eqnarray*}
			\sum_{j=1}^N p_j q_{j, i} &=& \sum_{j=1, j \ne i}^N p_j q_{j, i} + p_i q_{i, i} \\
						&=&  \sum_{j=1, j \ne i}^N p_j q_{j, i}  +  p_i \left( \sum_{j=1}^N q_{i, j} - \sum_{j=1, j \ne i}^N q_{i, j} \right) \\
						&\leq&  \sum_{j=1, j \ne i}^N p_j q_{j, i}  +  p_i \left(   c N  - \sum_{j=1, j \ne i}^N q_{i, j} \right) \\
						&\leq&  \sum_{j=1, j \ne i}^N p_j q_{j, i}  +  p_i \left(   c' N  - \sum_{j=1, j \ne i}^N q_{i, j} \right).
		\end{eqnarray*}

Let us consider the system of these linear equations: 
\begin{equation*}
\forall i \in \{ 1, \cdots, N\},  \sum_{j=1, j \ne i}^N p_j q_{j, i}  +  p_i \left(   c' N - \sum_{j=1, j \ne i}^N q_{i, j} \right) = c'.
\end{equation*}
 
 Now from Lemma \ref{lem:invertibility}, we can conclude that there exists a set of solutions $\{p_i\}_{i \in [N]}$ to the system such that, $\forall i \in [N] , p_i \geq 0 $, and $\sum_{i=1}^{N}p_i=1$. 
 For this choice of $\{p_i\}_{i \in [N]}$, the modification probability of every coordinate $i$ is at most $c'$. Since $\varepsilon>0$ was arbitrary, this proves the theorem.

\end{proof}

\begin{lemma}
 \label{lem:invertibility}
     There exists a set $\{p_i\}_{i \in [N]}$ such that  $ \forall i\in [N], p_i\geq 0$ and   $\sum_{i \in [N]}p_i=1$, and it satisfies the following system of linear equations:
     \[\forall i \in [N], 
     \sum_{j=1, j \neq i}^{N}p_j q_{j,i}+p_i \left(c'\cdot N-\sum_{j=1, j \neq i}^{N} q_{i,j} \right)=c'.\]
\end{lemma}
 
 \begin{proof}
 For ease of notation, for all $i \in [N]$, we denote $\sum_{j=1,j\neq i}^{N}$ by $\sum_{j\neq i}$.
 Let us assume that $\{p_i\}_{i \in [N]}$ is a satisfying assignment to the system of equations specified above. 
 
 Summing up all the equations in the system, we get:
 \begin{align*}
     \sum_{i=1}^{N}\left(\sum_{ j \neq i}p_j q_{j,i}+p_i(c'\cdot N-\sum_{ j \neq i} q_{i,j})\right)&= \sum_{i=1}^{N}\sum_{ j \neq i}p_j q_{j,i}+c'\cdot N\sum_{i=1}^{N}p_i-\sum_{i=1}^{N}\sum_{ j \neq i} p_iq_{i,j}\\
     &= c'\cdot N\sum_{i=1}^{N}p_i.
 \end{align*}
 Thus, $c'\cdot N\sum_{i=1}^{N}p_i= c'\cdot N$, which is the RHS after summing the system. So we can conclude that for any satisfying assignment, $\{p_i\}_{i \in [N]}, \sum_{i=1}^{N}p_i=1$. Therefore we only need to find a satisfying assignment non-negative $\{p_i\}_{i=1}^N$.

Let $A$ be a  $N \times N$ matrix such that $A_i$, that is, the $i^{th}$ column of $A$, denotes the vector of coefficients of $p_i$. That is, 
\[
\forall i \in[N], A_i= \left(q_{i,1}, q_{i,2},\dots,q_{i,i-1},c'\cdot N-\sum_{ j \neq i} q_{i,j}, q_{i,i+1}, \dots ,q_{i,N} \right)^T.
\]

Writing the system of equations as $[A_1, \dots A_{N}]\cdot [p_1, \dots p_{N}]^T=[c, \dots c]^T$, it is sufficient to show that $A$ is invertible in order to establish the existence of $\{p_i\}_{i \in [N]}$. In order to show that $p_i>0$ for all, 
$i \in [N]$, 
our proof will proceed in 3 steps:
\begin{enumerate}
    \item \label{item 1} Prove that $\forall j \in [N], |A_{c'}^{(j)}|>0$, where $A_{c'}^{(j)}$ is the matrix obtained by replacing $A_j$, the $j^{th}$ column of $A$, with $[c', \dots c']^T.$ [Claim \ref{clm:submatrixdet}]
    \item Use item \ref{item 1} to show that $|A|>0$. [Claim \ref{clm:Aposdet}]
    \item Use Cramer's Rule (Theorem \ref{cramers}) to establish that $\forall i \in [N], p_i>0$.
\end{enumerate}
\begin{claim} \label{clm:submatrixdet}
   Let  $A$ be the matrix described above and let the matrix obtained by replacing $A_j$ with $[c', \dots ,c']^T$ be denoted by $A_{c'}^{(j)}$. Then, $\forall j \in [N]$, $|A_{c'}^{(j)}|>0$.   
\end{claim}
\begin{proof}

 We will formally analyze the case when $j=N$, the other cases follow identically.  
\begin{align*}
 |A_{c'}^{(N)}|&=\left| \begin{bmatrix}
        c'\cdot N-\sum_{j \neq 1} q_{1,j}& q_{2,1}  &\dots & q_{N-1,1}& c'\\
        q_{1,2} & c'\cdot N-\sum_{j \neq 2} q_{2,j}&\dots & q_{N-1,2}& c'\\
        \vdots & \vdots & \vdots & \vdots & \vdots\\
        q_{1,N} & q_{2,N} & \dots & q_{N-1,N}  &   c'
    \end{bmatrix} \right|\\
    &= \left| \begin{bmatrix}
        c'\cdot N-\sum_{j \neq 1} q_{1,j}& q_{2,1} &\dots & q_{N-1,1}& c'\\
        q_{1,2} & c'\cdot N-\sum_{j \neq 2} q_{2,j} &\dots & q_{N-1,2}& c'\\
        \vdots & \vdots & \vdots & \vdots & \vdots\\
        c'\cdot N & c'\cdot N & \dots & \dots  &   c'\cdot N
    \end{bmatrix} \right|\\
    &= N \cdot \left| \begin{bmatrix}
        c'\cdot N-\sum_{j \neq 1} q_{1,j}& q_{2,1}  &\dots & q_{N-1,1}& c'\\
        q_{1,2} & c'\cdot N-\sum_{j \neq 2} q_{2,j}& \dots  & q_{N-1,2}& c'\\
        \vdots & \vdots & \vdots & \vdots & \vdots\\
        c' & c' & \dots & \dots  &   c'
    \end{bmatrix} \right|\\
    &= c'N \cdot \left| \begin{bmatrix}
        c'\cdot (N-1)-\sum_{j \neq 1} q_{1,j}& q_{2,1}-c' &\dots & q_{N-1,1}-c'& 0\\
        q_{1,2}-c' & c'\cdot (N-1)-\sum_{j \neq 2} q_{2,j} &\dots & q_{N-1,2}-c'& 0\\
        \vdots & \vdots & \vdots & \vdots & \vdots\\
        1 & 1 & \dots & \dots  &   1
    \end{bmatrix} \right|.
\end{align*}
We make the following observations for the above sequence of equations:
\begin{itemize}
    \item The second equality is established by adding the sum of rows 1 through $N-1$ to the $N^{th}$ row, and using the fact that  adding multiples of rows together does not change the determinant. 
    \item The third equality is established by factoring out $N$ from the $N$th row. 
    \item The fourth equality is established by subtracting the $N$th row from each of the rest of the rows (rows 1 through $N-1$), and then factoring out $c'$ from the last row.   
\end{itemize}

Now for all $i,j\in [N]$, let $\Delta_{i,j}=c'-q_{i,j}$. Observe that $\forall i\neq j: \Delta_{i,j}>0$. 
Then, 
\begin{align*}
   &  c'N \cdot \left| \begin{bmatrix}
        c'\cdot (N-1)-\sum_{j \neq 1} q_{1,j}& q_{2,1}-c' &\dots & q_{N-1,1}-c'& 0\\
        q_{0,1}-c' & c'\cdot (N-1)-\sum_{j \neq 2}q_{2,j} &\dots & q_{N-1,2}-c'& 0\\
        \vdots & \vdots & \vdots & \vdots & \vdots\\
        1 & 1 & \dots & \dots  &   1
    \end{bmatrix} \right|\\
    &=  c'N \cdot \left| \begin{bmatrix}
        \sum_{j \neq 1} \Delta_{1,j}& -\Delta_{2,1} & -\Delta_{3,1} &\dots & -\Delta_{N-1,1}& 0\\
        -\Delta_{1,2} & \sum_{j \neq 2} \Delta_{2,j} & \dots &\dots & -\Delta_{N-1,2}& 0\\
        \vdots & \vdots & \vdots & \vdots & \vdots\\
        1 & 1 & \dots & \dots & \dots &   1
    \end{bmatrix} \right|.
\end{align*}
where the equality follows from the definition of $\Delta_{i,j}$. 
Now observe that, from Laplace expansion along the $N^{th}$ row, Lemma \ref{cofactorexpansion} implies that the determinant of this matrix is
\begin{align*}
    \sum_{j=1}^{N}(-1)^{N+j}b_{N,j}m_{N,j},
\end{align*}
where $b_{N,j}$ is the entry of the $N^{th}$ row and $j^{th}$ column of the matrix, and $m_{N,j}$ is the determinant of the sub-matrix obtained by removing the $N^{th}$ row and $j^{th}$ column. For all $j \neq N$, $m_{N,j}=0$, since the resulting submatrix contains an all $0$ column. 
Thus, 
\begin{align*}
    \sum_{j=1}^{N}(-1)^{N+j}b_{N,j}m_{N,j}&= (-1)^{2N}b_{N,N}\cdot m_{N,N}\\
    &= 1 \cdot m_{N,N}.
\end{align*}
Where $m_{N,N}$ is the determinant of the following matrix, which we will denote by $\mathsf{M}$:
   \[  \begin{bmatrix}
        \sum_{j \neq 1} \Delta_{1,j}& -\Delta_{2,1} & -\Delta_{3,1} &\dots & -\Delta_{N-1,1}\\
        -\Delta_{1,2} & \sum_{j \neq 2} \Delta_{2,j} & \dots &\dots & -\Delta_{N-1,2}\\
        \vdots & \vdots & \vdots & \vdots & \vdots\\
        -\Delta_{1,N-1}&  -\Delta_{2,N-1} & \dots & \dots  & \sum_{j \neq N-1}  \Delta_{N-1,j}.
    \end{bmatrix} 
\]

We will now prove that $|\mathsf{M}|$
    is positive. Since $c'N>0$, this will imply that the determinant of $A^{(N)}_{c'}$ is positive as well.
    First, observe that, for all $i \in [N-1]$ we can write $\mathsf{M}_i$, the $i^{th}$ column of $\mathsf{M}$, as the following:
    \[
    \mathsf{M}_i= \sum_{j\neq i}\Delta_{i,j}\cdot v_{i,j},
    \]
    where $v_{i,j}\in \mathbb{R}^{N-1}$ is the following vector:
    \[
    v_{i,j}= \begin{cases}
        e_i-e_j & \text{if } j \in [N-1], j \neq i\\
        e_i & \text{if } j=N
    \end{cases},
    \]
    where $e_i$ is the $i^{th}$ standard basis vector. Therefore, \[\mathsf{M}= \left[\sum_{j\neq 1}\Delta_{1,j}\cdot v_{1,j}, \sum_{j\neq 2}\Delta_{2,j}\cdot v_{2,j}, \dots, \sum_{j\neq N-1}\Delta_{N-1,j}\cdot v_{N-1,j} \right].\]
    Using the linearity of the determinant, 
   \begin{align*}
    \text{det}(\mathsf{M}) =  \sum_{j_1 \neq 1}\dots \sum_{j_{N-1} \neq N-1} \Big(\Pi_{i=1}^{N-1}\Delta_{i,j_i}\Big)\cdot \text{det}[v_{1,j_1}, v_{2,j_2}, \dots v_{N-1,j_{N-1}}].
    \end{align*}
    Since $\Delta_{i,j}>0$ for all $i,j\in [N-1]$, it is sufficient to show that $\text{det}[v_{1,j_1}, v_{2,j_2}, \dots v_{N-1,j_{N-1}}]>0$ for all $j_1, \dots j_{N-1}$.
    By construction, for every choice of $j_1, \dots j_{N-1}$, the matrix $[v_{1,j_1}, \dots v_{N-1, j_{N-1}}]$ has the following structure:
    \begin{enumerate}
        \item All diagonal entries are $1$. 
        \item Each column has at most one $-1$. 
        \item All other entries are 0.
    \end{enumerate}
    
    For ease of notation, let $J$ denote $(j_1, \dots j_N)$, and let $\beta^{J}$ denote $[v_{1,j_1}, \dots v_{N-1, j_{N-1}}]$.
Next, we will show that $\beta^J$ is a  submatrix of the out-laplacian matrix, $\tilde{L}^J$, corresponding to   the directed unweighted graph $\tilde{G}^J= (V,E)$ over $N$ vertices constructed as follows:
\begin{enumerate}
    \item Let $S= \{1, \dots N-1\}$. Then the set of vertices is $V=S \cup \{r\}$, where $r \notin S$.  \footnote{Note that the addition of this dummy vertex is why our graph is over $N$ vertices.} 
    \item The set of edges $E$ is described as follows: For all $k \in S$
    \begin{itemize}
        \item  If there exists $\ell\in S $ such that $\beta^{J}_{k, \ell}=-1$,  then,  $\langle k, \ell\rangle \in E$.
        \item Otherwise, $\langle k, r\rangle \in E$.
    \end{itemize}
\end{enumerate}
By construction, $\tilde{G}^J$ is such that for every vertex $k \in S$, $\exists v \in V: \langle k, v \rangle\in E$.\footnote{In fact, there is exactly one such  vertex $v$.} Let $\tilde{L}^J$ be the out-laplacian matrix (Definition \ref{out_laplacian}) of $\tilde{G}$. Thus,  $\tilde{L}^J$ has the following structure:
\begin{equation*}
   \tilde{L}^J_{u,v}= \begin{cases}
       -1 & \text{if }  \langle u, v\rangle \in E\\
       \deg^+(u)  & \text{if } u=v\\
       0 & \text{otherwise}.
    \end{cases}
\end{equation*}
    We can conclude that, $\forall u,v \in S$, $\tilde{L}^J_{u,v}= \beta^J_{u,v}$ by construction.

   Thus, $\beta^J$ is  exactly  $\tilde{L}^{J^{(r,r)}}$, that is, the matrix obtained by deleting the row and column indexed by $r$ from $\tilde{L}^J.$
From the Matrix Tree theorem (Theorem  \ref{matrixtreethm}),  
\[
|\tilde{L}^{J^{(r,r)}}|= \#\{\text{in-arborescences of } \tilde{G}^J \text{ rooted at } r\}
\]
Therefore, we can  conclude that  $|\tilde{L}^{J^{(r,r)}}|= |\beta^J|\geq 0$. 
Recall that 
\[
\text{det}(\mathsf{M})= \sum_{J= j_1, \dots j_{N-1}}\left(\prod_{i=1}^{N-1}\Delta_{i,j_i}\right)\cdot \text{det}(\beta^J).
\]
Therefore, we can infer that $\text{det}(\mathsf{M})\geq 0$. To show that $\text{det}(\mathsf{M})$ is indeed positive, let $J= (N \dots N)$, i.e., $\forall i \in [N-1], j_i= N $. Then for all $i \in [N-1]$, $v_{i,N}= e_i$, which implies that $\beta^J= \mathbb{I}$, the identity matrix. Therefore, $\det(\beta^J)= \det[v_{1,N}. \dots v_{N-1,N}]=1$. \footnote{Another way to see this is the following: when $J= (N \dots N)$, then $\beta^{J}$ is the identity matrix, which implies that every vertex $\ell \in S$ has an edge pointing towards $r$ by construction. So there is 1 in-arborescence of $\tilde{G}^J$ rooted at $r$. }This implies that $\text{det}(\mathsf{M})>0$, which in turn implies that $\text{det}(A^{(N)}_{c'})>0$.
\end{proof}
\begin{claim}\label{clm:Aposdet}
    $|A|>0$.
\end{claim}
\begin{proof}

Let $ d_i= \det(A^{(i)}_{c'}), \forall i \in [N]$. Then $\vec{d}=[d_1, \dots d_N]^T \in (\mathbb{Z}^+)^N$.
From Lemma \ref{Adjdetidentity}, $\text{det}(A) \cdot \mathbb{I}_{N}= A \cdot \text{adj}(A)$, so we can conclude that $\text{det}(A) \cdot \mathbb{I}_{N}\cdot \vec{c}= A \cdot \text{adj}(A)\cdot \vec{c}= A \cdot \vec{d}$.  We know that $\vec{d}$ is a strictly positive vector and at least one entry of $A$ is strictly positive:  namely, for all $i \in [N], $ $A_{i,i}= c'\cdot N -\sum_{j\neq i}q_{i,j}>0$. Thus $A \cdot \vec{d}>0$. Since $c'>0$, this implies that $\text{det}(A)>0$.

\end{proof}
Claim \ref{clm:Aposdet} establishes that $\det(A)>0$, and Claim \ref{clm:submatrixdet} establishes that for all $i \in [N]$, $\det(A^{(i)})>0$. Since $ p_i = \frac{\det(A^{(i)})}{\det(A)}>0$ from Cramer's rule (Theorem \ref{cramers}), this implies that for all $i \in [N], p_i>0$. 
\end{proof}

\section{Computationally Hidden Flipped Set Conjecture}

We introduce the Computationally Hidden Flipped Set (CHFS) Conjecture. Before explaining the conjecture, recall that $O_X$ is a $(1-\delta)$ dense distribution over oracles $O:[N]\rightarrow \{0,1\}$, and $O_U$ is the uniform distribution over oracles $O:[N]\rightarrow \{0,1\}$. Now, let $\mathcal{M}$ be a modification algorithm such that $\mathcal{M}(O_U)\equiv O_X$ (or $\mathcal{M}(O_X) \equiv O_U$). The existence of such an algorithm is proven in section \Cref{sec:good_coupling_exits}. Informally, the CHFS conjecture says that, given the ability to make  few queries to $O\sim O_U$ (or $O\sim O_X$ respectively), no quantum algorithm can output a coordinate $i \in [N]$ such that $O(i)$ was flipped by  $\mathcal{M}$ with high probability. We present the CHFS conjecture formally below. 
\begin{conjecture}[Computationally Hidden Flipped Set Conjecture]\label{CHFSgeneral}
    Let $O_X$ be a $(1-\delta)$-dense distribution and $O_U$ be a uniform distribution. There exists a modification algorithm $\mathcal{M}$ such that the following holds:  $\mathcal{M}(O_U) \equiv O_X$.  Let $\text{Flip}_{O,\mathcal{M}}$ denote the set of all coordinates that $\mathcal{M}$ flips on input $O$.  Then, there exist absolute constants $C,a \geq 0, b > 0$ such that, for every $T$ query quantum algorithm $\mathcal{A}$ outputting a coordinate $i \in [N]$, we have 
    \begin{align*}
        \Pr_{O \sim O_U,\mathcal{A}, {M}}\left[\mathcal{A}^O \to i \, \wedge \,  i \in \text{Flip}_{O,\mathcal{M}}\right] \leq C \cdot T^a\cdot \delta^b,
    \end{align*}

    A similar conjecture can be defined by swapping the roles of $O_X$ and $O_U$.
\end{conjecture}

\subsection{Connection to Dense Indistinguishability Conjecture }
\label{sec:CHFS_implies_dense}
\begin{lemma}\label{CHFS_imply}
    Conjecture \ref{CHFSgeneral} implies Conjecture \ref{thm:main} (Dense indistinguishability conjecture). 
\end{lemma}
\begin{proof}

    Let $\ket {\phi_t^O}$ be the quantum state of the algorithm $\mathcal{A}$ right before the $t$-th query under oracle $O \sim O_U$.
    Let $W_i(\ket {\phi_t^O})$ be the query weight on input $i$, in the $t$-th stage under oracle $O$.
    By Theorem \ref{thm:bbbv97}, the distinguishing advantage is bounded by 
    \begin{align*}
       \mathbb{E}_{O\sim O_U, 
       \mathcal{M}}\left[ \sqrt{T \cdot \sum_{t=1}^T \sum_{i \in \text{Flip}_{O,\mathcal{M}}} W_i\left( \ket{\phi_t^O} \right) } \right].
    \end{align*}
  
    By the concavity of $\sqrt{x}$, we have it is at most
    \begin{align*}
        \sqrt{T \cdot \sum_{t=1}^T \mathbb{E}_{O\sim O_U,  \mathcal{M}} \left[ \sum_{i\in F} W_i\left( \ket{\phi_t^O} \right)\right] }.
    \end{align*}

    As long as the term $\mathbb{E}_{O\sim O_U,  \mathcal{M}} \left[ \sum_{i\in \text{Flip}_{O, \mathcal{M}}} W_i\left( \ket{\phi_t^O} \right)\right]$ is small for all $t$, we can conclude the proof. 
    Assume for some $t$, the term $\mathbb{E}_{O\sim O_U, \mathcal{M}} \left[ \sum_{i\in \text{Flip}_{O, \mathcal{M}}} W_i\left( \ket{\phi_t^O} \right)\right]$ is at least $g(t, \delta)$. We can construct an  adversary $\mathcal{A'}$ making $t-1$ queries for \ref{CHFSgeneral} as follows:
    \begin{enumerate}
        \item Run the algorithm $\mathcal{A}$ for the first $(t-1)$ queries and prepare the state before the $t$-th query.
        \item Measure the query registers to obtain and output $i$.
    \end{enumerate}

    Clearly, the probability that an input is in the flipped set is equal to $g((t-1), \delta)$. Thus, assuming Conjecture \ref{CHFSgeneral} is true, we will have $g((t-1), \delta) <    C\cdot T^a\cdot \delta^b$.

    Therefore, the distinguishing advantage is at most 
    \begin{align*}
        \sqrt{T \cdot \sum_{t=1}^T   (C \cdot t^a\cdot \delta^b)} \leq T \cdot \sqrt{ C \cdot T^a\cdot \delta^b}.
    \end{align*}
\end{proof}

\section{Simulation Theorems for Other Variants} 
 In this section, we will prove weaker variants of Conjecture \ref{CHFSgeneral} and show that these variants imply Conjecture \ref{thm:main} for additional variants of parallel quantum algorithms, including paralel quantum algorithms with limited adaptivity. 
 
\subsection{Simulation Theorem for Hybrid Quantum Algorithms}\label{sec:hybrid_quantum}

\begin{theorem}[A weaker variant of Conjecture \ref{CHFSgeneral}]\label{CHFS_classical}
  Let $O_X$ be a $(1-\delta)$-dense distribution and $O_U$ be a uniform distribution. Let $\mathcal{M}_C$ denote  the modification algorithm in Algorithm \ref{modalg:classical}. 
  Let $\text{Flip}_{O,\mathcal{M}_C}$ denote the set of all coordinates that $\mathcal{M}_C$ flips on input $O$. 
    Then there exist absolute constants $C,a \geq 0, b > 0$ such that, for every $T$ query classical algorithm $\mathcal{A}$ outputting a coordinate $i$, we have 
    \begin{align*}
        \Pr_{O \sim O_X,\mathcal{A}, \mathcal{M}_C}\left[\mathcal{A}^O \to i \, \wedge \,  i \in \text{Flip}_{O,\mathcal{M}_C}\right] \leq C \cdot T^a\cdot \delta^b.
    \end{align*}
\end{theorem}

\begin{theorem}[A weaker variant of Conjecture \ref{thm:main}]
\label{thm:main_classcial+parallel}
    Let $O_X$ be a $(1-\delta)$-dense distribution over oracles $O:[N] \to \{0,1\}$ and  $O_U$ be uniform.
    There exist absolute constants $C,a \geq 0, b > 0$  such that for all quantum algorithms $\As$ making $r$ (adaptive)  classical queries followed by $T$ parallel quantum queries 
    \[
    \Big|\Pr_{O \sim O_X}[\mathcal{A}^O \rightarrow 1]-\Pr_{O \sim O_U}[\mathcal{A}^O \rightarrow 1]\Big|\leq  C\cdot  (T + r)^a  \cdot \delta^b
    \]
\end{theorem}
\noindent Note that \Cref{thm:main_classcial+parallel} reproves the classical dense indistinguishability theorem when we set $T = 0$, with worse parameters compared to Claim 3~\cite{coretti2018random}.

\begin{lemma}\label{hybrid_CHFS_imply}
    Theorem  \ref{CHFS_classical} implies Theorem \ref{thm:main_classcial+parallel}.
\end{lemma}
\begin{proof}

   Let $\mathcal{A}$ be the quantum algorithm making $r$ classical queries followed by $T$ parallel queries to an oracle $O$ drawn from either $O_X$ or $O_U$. For $t \in [r+1]$, denote $\ket {\phi_t^O}$ to be the quantum state of the algorithm $\mathcal{A}$ right before the $t$-th query under oracle $O$; if $t = r+1$, it denotes the last layer of parallel quantum queries.
    
  By Theorem \ref{thm:bbbv97} and Theorem \ref{thm:parallelbbbv}, the distinguishing advantage is bounded by 
    \begin{align*}
       \mathbb{E}_{O\sim O_X, \mathcal{M}_C} \left[ \sqrt{ (r+1) \cdot \left( \sum_{t=1}^{r}\sum_{i \in \text{Flip}_{ O, \mathcal{M}_C}} W_i\left( \ket{\phi_t^O} \right) + \sum_{i \in \text{Flip}_{ O, \mathcal{M}_C}} W_i\left( \ket{\phi_{r+1}^O} \right)  \right) }\right],
    \end{align*}
    where $W_i$ in the second summation denotes the probability that when measuring all $T$ query registers, (at least) one of the outcomes is equal to $i$. 
    By the concavity of $\sqrt{x}$ and linearity of expectation, we have it is at most
    \begin{align*}
        \sqrt{r+1} \cdot \sqrt{\sum_{t=1}^r\mathbb{E}_{O\sim O_X, \mathcal{M}_C} \left[ \sum_{i \in \text{Flip}_{ O, \mathcal{M}_C}} W_i\left( \ket{\phi_t^O} \right)  \right] +  \mathbb{E}_{O\sim O_X, \mathcal{M}_C} \left[\sum_{i \in \text{Flip}_{ O, \mathcal{M}_C}}W_i\left( \ket{\phi_{r+1}^O} \right)  \right]}.
    \end{align*}
   As long as the term $\mathbb{E}_{O\sim O_X, \mathcal{M}_C} \left[ \sum_{i \in \text{Flip}_{ O, \mathcal{M}_C}} W_i\left( \ket{\phi_{t}^O} \right)  \right]$ is small for all $t \in [r+1]$, we can conclude the proof. 
    By \Cref{CHFS_classical}, the distinguishing advantage is at most  $\sqrt{C  \cdot (r+1) (T+r) \cdot r^a \cdot \delta^b}$.
\end{proof}

The rest of this section is dedicated to proving Theorem \ref{CHFS_classical}.
\paragraph{Proof of Theorem \ref{CHFS_classical}}
\begin{proof}
Let $\mathcal{A}$ be any $t$ query classical algorithm, and 
let $\mathbf{z}=(\mathbf{i_1},\mathbf{y_1}), \dots (\mathbf{i_t},\mathbf{y_t})$ be  the random transcript obtained by making $t$ queries to $O \sim O_X$. For all $z \sim \bf{z}$, for all $j \in [t]$, let $z_j:= (i_1,y_1), \dots (i_j,y_j)$, and let $O_X^{z_j}$ denote the distribution $O_X$ conditioned on $\mathbf{z_j}=z_j$. Let $M^*$ denote the modifying algorithm in Theorem \ref{strongchangeoracle}. Then, define the modifying algorithm $\mathcal{M}_C$ to be the following:
\begin{algorithm}
\caption{Modification Algorithm $\mathcal{M}_{C}$}\label{modalg:classical}
\begin{algorithmic}
\State \textbf{Input:} $O=(x_1,\dots,x_N)\sim O_X$
\State Run $\mathcal A^O$  and let $z=((i_1,y_1),\dots,(i_t,y_t))$
be the resulting transcript. 
\For{$j=1,\dots,t$}
    \State Let $z_{j-1}:=((i_1,y_1),\dots,(i_{j-1},y_{j-1}))$
    \State Let
    \[
    q_{i_j,b}:=\Pr_{O'\sim O_X^{z_{j-1}}}[\,O'(i_j)=b\,]
    \qquad\text{for each } b\in\{0,1\}
    \]
    \State Let $b:=y_j$
    \If{$q_{i_j,b} > \tfrac12$}
        \State Set $\tilde x_{i_j}\gets 1-b$ with probability $\frac{q_{i_j,b}-1/2}{q_{i_j,b}}$,
        and otherwise set $\tilde x_{i_j}\gets b$.
    \Else
        \State Set $\tilde x_{i_j}\gets b$.
    \EndIf
\EndFor
\State Sample the remaining coordinates
\[
(\tilde x_i)_{i\in [N]\setminus \{i_1,\dots,i_t\}}
\]
according to  Algorithm  $M^*$ from Theorem \ref{thm:strongchangeoracle} applied to the conditional distribution $O_X^z$.

\State \textbf{Output:} $\tilde O=(\tilde x_1,\dots,\tilde x_N)$
\end{algorithmic}
\end{algorithm}

 \begin{claim}
     Algorithm $\mathcal{M}_C$ outputs the uniform distribution.
 \end{claim}
\begin{proof}
  For $j \leq t$, $\mathcal{M}_C$, produces a uniformly random coordinate for each queried coordinate $i_j$, by the same analysis as Algorithm \ref{alg:ModificationalgM}. 
  Conditioned on the transcript $z$, the remaining coordinates are distributed
according to $O_X^z$. Applying $M^*$ to these coordinates yields a uniform
distribution over $[N]\setminus \{i_1,\dots,i_t\}$.
\end{proof}

\noindent Now we will establish an upper bound on the winning probability of a $t$ query classical adversary $\mathcal{A}$. 
For a transcript $z= (i_1,y_1), \dots (i_t,y_t)$, define $\mathrm{qry}(z)= \{i_1, \dots i_t\}$. Define $\mathrm{qry}(z)_j$ be the $j$-th input $i_j$.

\begin{align*}
     &\Pr_{O \sim O_X, \mathcal{A}, \mathcal{M}_C} [\mathcal{A}^O \rightarrow i \wedge i \in \text{Flip}_{O,\mathcal{M}_C}]\\
     =& \Pr_{O \sim O_X,\mathcal{A}, \mathcal{M}_C} [\mathcal{A}^O \rightarrow i \wedge i \in \text{Flip}_{O,\mathcal{M}_C} \wedge i \in \mathrm{qry}(z)] + \Pr_{O \sim O_X,\mathcal{A}, \mathcal{M}_C} [\mathcal{A}^O \rightarrow i \wedge i \in \text{Flip}_{O,\mathcal{M}_C} \wedge i \not\in \mathrm{qry}(z)]\\
     \leq& \Pr_{O \sim O_X,\mathcal{A}, \mathcal{M}_C} [\exists i \in \mathrm{qry}(z), i \in \text{Flip}_{O, \mathcal{M}_C}] + \Pr_{O \sim O_X,\mathcal{A}, \mathcal{M}_C} [\mathcal{A}^O \rightarrow i \wedge i \in \text{Flip}_{O,\mathcal{M}_C} \wedge i \not\in \mathrm{qry}(z)] \\
     \leq& \sum_{j=1}^t \Pr_{O \sim O_X,\mathcal{A}, \mathcal{M}_C} [\mathrm{qry}(z)_j \in \text{Flip}_{O, \mathcal{M}_C}] + \Pr_{O \sim O_X,\mathcal{A}, \mathcal{M}_C} [\mathcal{A}^O \rightarrow i \wedge i \in \text{Flip}_{O,\mathcal{M}_C} \wedge i \not\in \mathrm{qry}(z)].
\end{align*}

We will complete the proof using Lemma \ref{lem:queried_dense} and Lemma \ref{lemma:avgdensesmall}, which are proved below. 

\begin{lemma}\label{lem:queried_dense}
Let $\bf{z}=(i_1,y_1), \dots, (i_t,y_t)$ denote the random transcript obtained by making $t$ queries to $O \sim O_X$, and for all $z \sim \bf{z}$, let $\mathrm{qry}(z):= \{i_1, \dots, i_t\}$. Then
    \[
    \sum_{j=1}^t\Pr_{\substack{O\sim O_X \\ A, \mathcal{M}_C}}[\mathrm{qry}(z)_j \in \text{Flip}_{O,\mathcal{M}_C}]\leq \sqrt{\frac{\delta \cdot \ln 2}{2}} \cdot t.
    \]
 \end{lemma}
\begin{proof}
The proof is similar to that of Lemma \ref{changeoracle}, except the order is adaptive. We will prove this theorem for deterministic classical algorithms; this will easily extends to  randomized algorithms (which is a convex combination of deterministic algorithms). 

The modification algorithm flips the distribution according to the order $i_1, i_2, \ldots, i_t$, where each $i_j$ itself is a random variable. Similar to Lemma \ref{changeoracle}, for all $j \in  \{0, 1, \ldots, t-1\}$ and $x \in \{0,1\}^j$, $p_x$ is defined as the probability that the outcome for the first $j$ queries is $x$.  Since the algorithm is deterministic, once $x$ is fixed, we know $i_1, i_2, \ldots, i_t$. Let $q_{x, b}$ be the conditional probability that the next bit is $b$ conditioned on the transcript $x$ so far. Define $e_x = |q_{x, 0} - 1/2| = |q_{x, 1} - 1/2|$. The quantity we want to bound can be represented by $\sum_{j=0}^{t-1} p_x e_x$. 

Since $\{p_x\}_{x \in \{0,1\}^j}$ is induced by a classical deterministic decision tree, we have that $\{p_x\}_x$ is a distribution, it is $(1-\delta)$-dense and thus has Shannon entropy at least $(1-\delta) j$. 
Therefore, following the identical entropy argument and Cauchy-Schwarz in Lemma \ref{changeoracle}, we have
\begin{align*}
    \sum_{j=0}^{t-1} p_x e_x \leq \sqrt{\frac{\ln 2}{2} \delta} \cdot t.
\end{align*}
\end{proof}

Now we bound the second term. 
\begin{lemma}
\label{lemma:avgdensesmall}
     Let $O_X$ be a $(1-\delta)$ dense distribution, and let $\mathcal{A}$ be a classical algorithm that makes $t$ queries to $O \sim O_X$. Let $\bf{z}= (i_1,y_1), \dots (i_t,y_t)$ be the random variable over $([N]\times \{0,1\})^t$ of the transcript obtained by $\mathcal{A}$, and for all $z= (i_1,y_1), \dots (i_t,y_t)$ in the support of $\bf{z}$, let $\mathrm{qry(z)}:= \{i_1, \dots i_t\}$ denote the set of queried points.  Let $O^{z}_{X}$ denote the distribution of the restriction $O_{|N\setminus\mathrm{qry(z)}}$, where $O \sim O_X$ conditioned on $\mathbf{z}=z$. Let $\delta_z$ be the smallest real number such that $O^z_X$ is $(1-\delta_z)$ dense. Then
    \[
    \mathbb{E}_{z \sim \bf{z}}[\delta_z]\leq \delta \cdot (t+1), 
    \]
    where the randomness of the expectation is taken over $O \sim O_X$ and the classical algorithm $A$.
\end{lemma}

Given the above lemma, we can bound the second term.
\begin{align*}
    & \Pr_{O \sim O_X} [A^O \rightarrow i \wedge i \in \text{Flip}_{O,\mathcal{M}_C} \wedge i \not\in \mathrm{qry}(z)] \\
    &= \sum_{z  \sim \mathbf{z}} \Pr_{O \sim O_X,\mathcal{A}}[z]\cdot  \Pr_{O \sim O_X,\mathcal{A},\mathcal{M}_C} [A^O \rightarrow i \wedge i \in \text{Flip}_{O,\mathcal{M}_C} \wedge i \not\in \mathrm{qry}(z)|z]\\
    &= \sum_{z  \sim \mathbf{z}} \Pr_{O \sim O_X,\mathcal{A}}[z]\sum_{i} \Pr[\mathcal{A}^O\rightarrow i]\cdot  \Pr_{O \sim O_X,\mathcal{M}_C} [ i \in \text{Flip}_{O,\mathcal{M}_C} \wedge i \not\in \mathrm{qry}(z)|z,\mathcal{A}^O\rightarrow i]\\
    &\leq \sum_{z  \sim \mathbf{z}} \Pr_{O \sim O_X}[z] \cdot \max_{i^*\not\in \mathrm{qry}(z)} \Pr_{O \sim O^z_X}[i^* \in \text{Flip}_{O,\mathcal{M}_C}] \\
     &= \sum_{z  \sim \mathbf{z}} \Pr_{O \sim O_X}[z] \cdot \sqrt{\frac{\ln 2}{2} \cdot \delta_z} \\
     &\leq \sqrt{\frac{\ln 2}{2} \cdot \mathbb{E}_{z  \sim \mathbf{z}} [\delta_z]}.\\
    &\leq \sqrt{\frac{\ln 2}{2} \cdot \delta \cdot (t+1)}
\end{align*}

Finally, we prove Lemma \ref{lemma:avgdensesmall} .
\begin{proof}[Proof for Lemma \ref{lemma:avgdensesmall}]
We will prove this theorem for deterministic classical algorithms; this will easily extends to  randomized algorithms (which is a convex combination of deterministic algorithms). 

Fix a realization $z = (i_1,y_1), \dots (i_t,y_t)$ of $\bf{z}$. Let $S \subseteq [N]\setminus \{i_1, \dots i_t\}$.  
     \begin{align*}
         &\Pr_{O \sim O_X^z}[O(S)=y_S]\\
         &= \frac{ \Pr_{O \sim O_X}[O(S)=y_S\wedge O(i_1)=y_1, \dots O(i_t)= y_t]}{\Pr_{O \sim O_X}[O(i_1)=y_1, \dots O(i_t)= y_t]}\\
         &\leq \frac{2^{-(1-\delta)\cdot (|S|+t)}}{\Pr_{O \sim O_X}[O(i_1)=y_1, \dots O(i_t)= y_t]}\\
         &= \frac{2^{-(1-\delta)\cdot (|S|+t)}}{\Pr[z]}.
     \end{align*}
    Since the distribution $O_X^z$ is $(1-\delta_z)$ dense, this implies that
    \begin{align*}
         \frac{2^{-(1-\delta)\cdot (|S|+t)}}{\Pr[z]} \leq 2^{-(1-\delta_z)\cdot |S|}. 
    \end{align*}
   Therefore, 
   \[
   \delta_z\geq \delta + (-t+\delta \cdot t- \log (\Pr[z])) / |S|.
   \]
   Observe that  $ -t+\delta \cdot t- \log (\Pr[z]) \geq -(1-\delta) t+(1-\delta) t \geq 0$ due to denseness.

  Since for all sets $S$ such that $|S| \leq N-t$, we want $\delta_z\geq \delta \cdot |S|-t+\delta \cdot t- \log (\Pr[z])$, we will set $|S|=1$, since $\delta_z$ is maximized when $|S|=1$.

  Thus 
  \[\delta_z:= \delta(t+1)-t-\log (\Pr[z]).\]

   Now, 
   \begin{align*}
       \mathbb{E}_{z \sim \bf{z}}[\delta_z] & = \sum_{z}\Pr[z]\cdot \delta_z\\
       &= \sum_z \Pr[z]\cdot \Big(\delta(t+1)-t-\log (\Pr[z])\Big)\\
       &= \delta \cdot (t+1)-t+\sum_{z}\Pr[z]\log (1/\Pr[z])\\
       &= \delta \cdot (t+1)-t+ H(\bf{z})\\
       &\leq \delta \cdot (t+1).
    \end{align*}
    The last inequality is due to $H({\bf{z}})\leq t$, since $z$ is a transcript produced by a deterministic classical decision tree making $t$ queries. 
 \end{proof}

 Combining both lemmas, we conclude the proof.

\end{proof}

\subsection{Simulation Theorem with Bounded Quantum Preprocessing}
\label{sec:two_layer_quantum}

\begin{theorem}[A weaker variant of Conjecture \ref{CHFSgeneral}]\label{CHFS_constant+parallel}
  Let $O_X$ be a $(1-\delta)$-dense distribution and $O_U$ be a uniform distribution. Let $M_U^*$ be the modification algorithm in Theorem \ref{uniformtodenseM}. Let $\text{Flip}_{O,M^*_U}$ denote the set of all coordinates that $M^*_U$ flips on input $O$ sampled from $O_U$. 


    There exist  absolute constants $C, b > 0$ such that, for all quantum algorithms making $t$  queries to $O$:
    
    \begin{align*}
        \Pr_{O \sim O_U,\mathcal{A}, M_U^*}\left[\mathcal{A}^O \to i \, \wedge \,  i \in \text{Flip}_{O,M^*_U}\right] \leq  (C \log (e/\delta))^t\cdot  \delta^b.
    \end{align*}
\end{theorem}

\begin{theorem}[A weaker variant of Conjecture \ref{thm:main}]
\label{thm:main_constant+parallel}
    Let $O_X$ be a $(1-\delta)$-dense distribution over oracles $O:[N] \to \{0,1\}$ and  $O_U$ be uniform. There exist absolute constants $C > 0,a \geq 0, b > 0$ such that,   for 
 quantum algorithms $\As$ making $r$   quantum queries and then $T$ parallel quantum queries, 
    \[
    \Big|\Pr_{O \sim O_X}[\mathcal{A}^O \rightarrow 1]-\Pr_{O \sim O_U}[\mathcal{A}^O \rightarrow 1]\Big|\leq   (C \cdot \log (e/\delta))^r \cdot T^a \cdot \delta^b.
    \]
\end{theorem}
Our second limited-adaptivity extension allows a bounded initial quantum stage followed by a large parallel-query stage. The quantitative bound deteriorates exponentially in the number $r$ of initial quantum queries, as $((C\log(e/\delta))^r)$. Consequently, the theorem applies not only for constant $r$, but also for slowly growing $r$. In particular,
we can set $r= \frac{  \log (1/\delta)}{ 
\log \log (1/\delta)}$. The reduction to the simulation conjecture requires $\delta = K^{-O(1)}$ for $K = \max\{T', 1/\delta', 1/\epsilon')$, where $T', \delta', \epsilon'$ are those defined in the Simulation Conjecture; thus, we can set $r = O(\log K / \log\log K)$, which gives the first half of \Cref{thm:main_intro2}.

\begin{lemma}
    Theorem  \ref{CHFS_constant+parallel} implies Theorem \ref{thm:main_constant+parallel}.
\end{lemma}
The proof is similar to those for Lemma~\ref{CHFS_imply} and Lemma~\ref{hybrid_CHFS_imply}, thus we omit it here.
The rest of this section is dedicated to the proof of Theorem \ref{CHFS_constant+parallel}. We first prove a moment bound when the number of queries is bounded. 

\medskip

\begin{lemma}[Moment bound for a fixed outcome after $t$ queries]
\label{lem:fixed-outcome-concentration}
Let $\mathcal A$ be an arbitrary $r$-query quantum algorithm, and let $e$ be
any measurement event. For $x\sim\{-1,1\}^N$ uniform, define
\[
    p_e(x):=\Pr[\mathcal A^x \in e],
    \qquad
    \overline p_e:=\mathbb E_x[p_e(x)].
\]
Then for every integer $s\ge 1$,
\[
    \mathbb E_x[p_e(x)^s]
    \le
    (2s-1)^{rs}\,\overline p_e^{\,s}.
\]
Equivalently,
\[
    \left(\mathbb E_x[p_e(x)^s]\right)^{1/s}
    \le
    (2s-1)^r\,\overline p_e .
\]

\end{lemma}

\begin{proof}
Let $\Pi_e$ be the projector corresponding to event $e$, and set
\[
    F_e(x):=\Pi_e|\psi_r^x\rangle .
\]
By Theorem~\ref{thm:adaptive-polynomial-representation}, $F_e$ is a
vector-valued multilinear polynomial of degree at most $r$. Also
\[
    p_e(x)=\|F_e(x)\|_2^2.
\]
By vector-valued hypercontractivity (Theorem \ref{vec:bonami-beckner}) applied with moment $2s$,
\[
    \mathbb E_x\left[\|F_e(x)\|_2^{2s}\right]
    \le
    (2s-1)^{rs}
    \left(\mathbb E_x\left[\|F_e(x)\|_2^2\right]\right)^s.
\]
Substituting $\|F_e(x)\|_2^2=p_e(x)$ gives
\[
    \mathbb E_x[p_e(x)^s]
    \le
    (2s-1)^{rs}\,\overline p_e^{\,s}.
\]
Taking $s$-th roots gives the equivalent formulation.
\end{proof}

Using the moment bound, we prove Theorem \ref{CHFS_constant+parallel}.
\begin{proof}[Proof of Theorem \ref{CHFS_constant+parallel}]
For all $O \sim O_U$ and all $i \in [N]$, let $p_i^O$ denote $\Pr[A^O \rightarrow i]$, and let $q_i^O$ denote $ \Pr_{\mathcal{M}_U^*}[i \in \text{Flip}_{O,\mathcal{M}_U^*}]$. Now we will analyze $\sum_{i}\mathbb{E}_{O\sim O_U}[p_i^O\cdot q_i^O]$, since 
\begin{align*}
    &\Pr_{O \sim O_U, \mathcal{A},M^*_U,i}[\mathcal{A}^O\rightarrow i\wedge i \in \text{Flip}_{O,M_U^*}]\\
        &=\sum_{i}\Pr_{O, \mathcal{A},M^*_U}[\mathcal{A}^O \rightarrow i \wedge i \in \text{Flip}_{O,M^*_U}]\\
        &= \sum_{i}\mathbb{E}_{O}[p_i^O\cdot q_i^O].
\end{align*}
Holder's Inequality  (Theorem \ref{holder's}) implies that, for all $i \in [N]$ 
\begin{align*}
   \mathbb{E}_{O}[p_i^O\cdot q_i^O]&\leq \mathbb{E}_{O}[(p_i^O)^s]^{1/s}\cdot\mathbb{E}_{O}[(q_i^O)^{s'}]^{1/{s'}} \end{align*}
   such that $\frac{1}{s}+\frac{1}{s'}=1$.
   Let $\overline{p}_i= \mathbb{E}_{O}[p_i^O]$. Then Lemma \ref{lem:fixed-outcome-concentration} implies that
   \[
   \mathbb{E}_{O}[(p_i^O)^s]\leq (2s-1)^{st}\overline{p}_i^s
   \]
   
   Furthermore, since $ 0 \leq q_i^O\leq 1$, Corollary \ref{uniformtodenseM} implies that 
   \begin{align*}
       \mathbb{E}_{O}[(q_i^O)^{s'}] 
       \leq \mathbb{E}_{O}[(q_i^O)]
       \leq \sqrt{\frac{\ln 2}{2}\delta}+ o(1).
   \end{align*}
   Let $\eta= \sqrt{\frac{\ln 2}{2}\delta}+ o(1)$.
  Then for all $i \in [N]$, 
   \[
   \mathbb{E}_O[p_i^O \cdot q_i^O]\leq (2s-1)^{t}\overline{p}_i\cdot \eta^{1/s'}
   \]

   Summing over all $i \in [N]$, we get that
   \begin{align*}
       \sum_{i}\mathbb{E}_O[p_i^O \cdot q_i^O]
       &\leq  \sum_{i} (2s-1)^{t}\overline{p}_i\cdot \eta^{1/s'}\\
       &= (2s-1)^{t} \cdot \eta^{1/s'} \sum_{i} \overline{p}_i\\
       &= (2s-1)^{t} \cdot \eta^{1/s'}\\
       &= (2s-1)^{t} \cdot \eta^{1-1/s}
   \end{align*}

   Now, set $L= \log({1/\eta})$. Then 
   \begin{align*}
        (2s-1)^{t} \cdot \eta^{1-1/s}
        =  (2s-1)^{t} \cdot \eta \cdot e^{L/s} 
   \end{align*}
  Let $s= \big\lceil \max{\{2, \frac{L}{t}\}}\big \rceil$ for $t \geq 1$. Then,
  \[e^{L/s}\leq e^t\]
   and 
   \[
   2s-1\leq O(1+\frac{L}{t})
   \]
    Thus,
    \begin{align*}
        (2s-1)^{t} \cdot \eta^{1-1/s} &\leq \eta \cdot  O(1+\frac{L}{t})^t\cdot e^t\\
        &\leq  \eta \cdot  O(1+\frac{\log (1/\eta)}{t})^t
    \end{align*}
    Since $\eta = \Theta(\sqrt{\delta}),$ and $L= \Theta(\log (e/\delta))$, we can conclude that 
    \[(2s-1)^{t} \cdot \eta^{1-1/s}\leq \sqrt{\delta}\cdot  O(1+\frac{\log (e/\delta)}{t})^t\leq \sqrt{\delta}\cdot O(\log(e/\delta))^t\]Putting everything together,
    \[\Pr_{O \sim O_U, \mathcal{A},M^*_U,i}[\mathcal{A}^O\rightarrow i\wedge i \in \text{Flip}_{O,M_U^*}]\leq \sqrt{\delta}\cdot O(\log(e/\delta))^t\]
\end{proof}

\section{Simulation with a Constant Number of Adaptive Parallel-Query Layers}
\label{sec:fixed-adaptive-layers}

Section  \ref{sec:two_layer_quantum} considers a slowly growing quantum query prefix followed by one final parallel query layer. We now consider a complementary form of limited adaptivity. The total query budget may be divided among several parallel-query layers, with arbitrary oracle-independent quantum operations between consecutive layers. Thus every early layer may itself contain many queries, while only the number of parallel-query layers is fixed.

A $T$ query, $d$-layer parallel-query algorithm  has the form
$$
V_d U_O^{\otimes T_d}V_{d-1}\cdots V_1U_O^{\otimes T_1}V_0,
$$
followed by a measurement, where the $V_r$ are oracle-independent unitaries and layer $r$ contains $T_r$ parallel queries.  We write
$$
Q_r
:=
\sum_{s=1}^r T_s,
\qquad
T
:=
Q_d.
$$
The case $d=1$ is covered by Theorem \ref{thm:main_intro}.

 The induction proves a simulation theorem directly; it does not prove dense indistinguishability for several adaptive layers. Thus here we will assume $d \geq 2$.
\begin{theorem}[Simulation for a fixed number of adaptive parallel-query layers]
\label{thm:fixed-adaptive-layers}
There are absolute constants $C,c>0$ such that the following holds. Let $\mathcal{A}$ be a quantum algorithm with at most $d$ adaptive parallel-query layers and at most $T$ total queries. Define
$$
f(O)
:=
\Pr[\mathcal{A}^O\to1].
$$
For every $0<\epsilon,\delta<1$, there is a deterministic classical query algorithm with output $g(O)\in[0,1]$ that makes at most
$$
\left(
\frac{CdT}{\epsilon\sqrt{\delta}}
\right)^{c\cdot 4^d}
$$
queries and satisfies
$$
\Pr_{O\sim O_U}
\left[
|f(O)-g(O)|>\epsilon
\right]
\leq
\delta.
$$
Consequently, the simulation conjecture holds for every fixed number of adaptive parallel-query layers.
\end{theorem}

\subsection{The one-layer theorem in restriction-stable mean-square form}
\label{subsec:fixed-layers-base}

Choose a polynomial upper bound
$$
P_{\mathrm{par}}(T_1,\alpha,\beta)
=
\operatorname{poly}
\left(
T_1,
\frac{1}{\alpha},
\frac{1}{\beta}
\right)
$$
for the number of classical queries guaranteed by Theorem \ref{thm:main_intro}. Thus, if $a(O)$ is the acceptance probability of a $T_1$ query parallel quantum algorithm, then there is a deterministic classical decision tree of depth at most $P_{\mathrm{par}}(T_1,\alpha,\beta)$ whose output $b(O)$ satisfies
$$
\Pr_{O\sim O_U}
\left[
|a(O)-b(O)|>\alpha
\right]
\leq
\beta.
$$

\begin{lemma}[Restriction-stable one-layer mean-square simulation]
\label{lem:parallel-l2-base}
For every $0<\rho<1$, the acceptance probability $a(O)$ of a parallel quantum algorithm making $T_1$ queries has a deterministic decision-tree approximation $b(O)$ satisfying
$$
\mathbb{E}_{O\sim O_U}
\left[
(a(O)-b(O))^2
\right]
\leq
\rho^2,
$$
with depth at most
$$
S_{\mathrm{par}}(T_1,\rho)
:=
P_{\mathrm{par}}
\left(
T_1,
\frac{\rho}{\sqrt{2}},
\frac{\rho^2}{2}
\right)
=
\operatorname{poly}
\left(
T_1,
\frac{1}{\rho}
\right).
$$
The same statement holds after fixing an arbitrary set of oracle values and for every binary postprocessing of an arbitrary finite-output one-layer algorithm. 
\end{lemma}

\begin{proof}
Apply Theorem \ref{thm:main_intro} with
$$
\alpha
:=
\frac{\rho}{\sqrt{2}},
\qquad
\beta
:=
\frac{\rho^2}{2}.
$$
Let the set of all $O \sim O_U$ such that $|a(O)-b(O)|\leq \alpha$ be $\mathsf{Good}$ and let the remaining oracles be $\mathsf{Bad}$. Observe that $\Pr_{O \sim O_U}[O \in \mathsf{Bad}]\leq \beta= \frac{\rho^2}{2}$. For all $O \in \mathsf{Good}$, $((a(O)-b(O))^2 \leq \rho^2/2$.  For all $O \in \mathsf{Bad}$, $((a(O)-b(O))^2 \leq 1$. Therefore
$$
\mathbb{E}_{O\sim O_U}
\left[
(a(O)-b(O))^2
\right]
\leq
\frac{\rho^2}{2}
+
\frac{\rho^2}{2}
=
\rho^2.
$$

If some oracle values are fixed, the corresponding controlled phases are known and can be absorbed into the oracle-independent operations. The restricted computation is therefore a one-layer parallel-query algorithm on the remaining uniform subcube. The same conclusion applies to every binary postprocessing as well. 
\end{proof}

\subsection{A variance bound from the layerwise query weights}
\label{subsec:fixed-layers-variance}

For an oracle $O:[N]\to\{0,1\}$ and $i\in[N]$, let $O^{\oplus i}$ be obtained by flipping $O(i)$. For a real-valued function $h$ on Boolean oracles, define
$$
\partial_i h(O)
:=
\frac{h(O)-h(O^{\oplus i})}{2},
\qquad
M_h(O)
:=
\max_{i\in[N]}|\partial_i h(O)|.
$$
For Fourier analysis only, write $x_i=(-1)^{O(i)}$ and
$$
\chi_S(O)
:=
\prod_{i\in S}x_i
=
(-1)^{\sum_{i\in S}O(i)}.
$$

\begin{lemma}[Pointwise low-degree gradient bound]
\label{lem:pointwise-markov-gradient}
Let $h:\{0,1\}^N\to[-1,1]$ have Fourier degree at most $D$. Then, for every oracle $O: [N]\rightarrow \{0,1\}$,
\begin{equation}\label{eq:first_implication}
    \sum_{i=1}^N
|\partial_i h(O)|
\leq
D^2.
\end{equation}

Consequently,
\begin{equation}\label{eq:second_implication}
\operatorname{Var}_{O\sim O_U}[h(O)]
\leq
D^2
\mathbb{E}_{O\sim O_U}[M_h(O)].
\end{equation}
\end{lemma}

\begin{proof}

Equation \ref{eq:first_implication} follows directly from the pointwise $L_1$ influence bound of Filmus, Hatami, Keller, and Lifshitz [Theorem 3.3, \cite{FHKL16}].
For Equation \ref{eq:second_implication}, the Boolean Poincaré inequality gives
$$
\operatorname{Var}(h)
\leq
\sum_{i=1}^N
\mathbb{E}
\left[
(\partial_i h(O))^2
\right].
$$
Observe that
$$
\sum_{i=1}^N
(\partial_i h(O))^2
\leq
M_h(O)
\sum_{i=1}^N
|\partial_i h(O)|
\leq
D^2M_h(O).
$$
Taking expectations completes the proof.

\end{proof}

Let $|\varphi_r^O\rangle$ be the state immediately before layer $r:$

\begin{align*}
        \ket{\phi_r^O} = \sum_{i', z} \alpha_{i', z} \ket{i', z},
    \end{align*}
    where $i'$ consists of $T_r$ inputs and $z$ consists of everything else. 
    The query magnitude on $i$ is 
    \begin{align*}
        W_{r,i}(O) = \sum_{\substack{i' = i'_1, \ldots, i'_{T_r}, z \\ \exists k,\, i'_k = i}} |\alpha_{i', z}|^2.
    \end{align*}

\begin{lemma}[$d$-layer variance bound]
\label{lem:d-layer-variance}
Let $f(O)$ be the acceptance probability of a quantum algorithm with $d$ parallel-query layers and $T$ total queries. Then
$$
\operatorname{Var}_{O\sim O_U}[f(O)]
\leq
16T^2
\sum_{r=1}^d
\sqrt{
\mathbb{E}_{O\sim O_U}
\left[
\max_{i\in[N]}W_{r,i}(O)
\right]
}.
$$
The same statement holds after fixing any set of oracle values (i.e., restriction). On a restriction, the maximum is over the unfixed coordinates. 
\end{lemma}

\begin{proof}
Fix $i\in[N]$. Then Theorem \ref{thm:parallelbbbv} and the triangle inequality  implies that 
$$
\left\|
|\psi^O\rangle
-
|\psi^{O^{\oplus i}}\rangle
\right\|_2
\leq
2\sum_{r=1}^d
\sqrt{W_{r,i}(O)}.
$$
Thus
$$
|\partial_i f(O)|
\leq
2\sum_{r=1}^d
\sqrt{W_{r,i}(O)}.
$$

By Theorem \ref{thm:adaptive-polynomial-representation}, the acceptance probability $f$ has Fourier degree at most $2T$. Applying Lemma~\ref{lem:pointwise-markov-gradient} with $D=2T$ gives
$$
\begin{aligned}
\operatorname{Var}(f)
&\leq
8T^2\mathbb{E}[M_f(O)]
\\
&\leq
16T^2
\sum_{r=1}^d
\mathbb{E}
\left[
\sqrt{
\max_i W_{r,i}(O)
}
\right]
\\
&\leq
16T^2
\sum_{r=1}^d
\sqrt{
\mathbb{E}
\left[
\max_i W_{r,i}(O)
\right]
},
\end{aligned}
$$
where the last inequality is Jensen's inequality.

After fixing some oracle values, the corresponding phases are known and can be absorbed into the oracle-independent operations. The same hybrid and degree arguments therefore apply on the remaining uniform subcube.
\end{proof}

\subsection{Bonami--Beckner regularity for output distributions}
\label{subsec:fixed-layers-regularity}

In this section, we use hypercontractivity arguments to bound the information contained in the entire output of a bounded-query algorithm. Importantly, the resulting
bound is independent of the size of the output alphabet.

 Following the notation of
Section \ref{sec:two_layer_quantum}, define
$$
p_j^O
:=
\Pr[\mathcal{A}^O\to j],
\qquad
p_j
:=
\mathbb{E}_{O\sim O_U}[p_j^O],
\qquad
\mathbf{p}^O
:=
(p_j^O)_{j\in\mathcal{J}}.
$$
All logarithms in this subsection are natural.

\begin{lemma}[Information in a bounded-query output]
\label{lem:output-information} Let $\mathcal{A}$ be an arbitrary $r$ query quantum algorithm. For $\mathbf{O} \sim O_U$ let $\mathbf{J}$ be
the output of $\mathcal{A}^O$.
Then
$$
I(O;\mathbf{J})
\leq
2r\log 3.
$$
The same conclusion holds after fixing an arbitrary set of oracle
values.
\end{lemma}

\begin{proof}
Here $I(O;\mathbf{J})$ denotes the mutual information in the joint
experiment in which $\mathbf{O}\sim O_U$ is sampled first and $\mathbf{J}$ is
then sampled according to the output distribution of
$\mathcal{A}^O$. Equivalently,
$$
I(O;\mathbf{J})
=
\sum_{j:p_j>0}
p_j
D_{\mathrm{KL}}
\left(
O_{U|\mathbf{J}=j}
\middle\|
O_U
\right).
$$
where $O_{U|\mathbf{J}=j}$ is the uniform distribution over oracles conditioned on $\mathbf{J}=j$.
Fix an output $j$ with $p_j>0$. Apply
Lemma \ref{lem:fixed-outcome-concentration}  to the measurement event
$$
e_j:=\{\mathbf{J}=j\},
$$
and with $s=2$. Since the algorithm makes $r$ queries, the lemma
gives
$$
\mathbb{E}_{O\sim O_U}
\left[
(p_j^O)^2
\right]
\leq
3^{2r}p_j^2.
$$

For every oracle $O\in\{0,1\}^N$, Bayes’ rule gives

\begin{align*}
\Pr_{\mathbf{O} \sim O_U}[\mathbf{O}=O\mid J=j]
&=
\frac{\Pr[J=j\mid \mathbf{O}=O]\Pr_{\mathbf{O} \sim O_U}[\mathbf{O}=O]}
{\Pr[J=j]}\\
&=
\frac{p_j^O}{p_j}\cdot 2^{-N},
\end{align*}

  Therefore, Jensen's inequality gives
$$
\begin{aligned}
D_{\mathrm{KL}}
\left(
O_{U\mid\mathbf{J}=j}
\middle\|
O_U
\right)
&=
\mathbb{E}
\left[
\left.
\log\left(\frac{p_j^O}{p_j}\right)
\right|
\mathbf{J}=j
\right]
\\
&\leq
\log
\mathbb{E}
\left[
\left.
\frac{p_j^O}{p_j}
\right|
\mathbf{J}=j
\right]
\\
&=
\log
\left(
\frac{
\mathbb{E}_{O\sim O_U}[(p_j^O)^2]
}{
p_j^2
}
\right)
\\
&\leq
2r\log 3.
\end{aligned}
$$
Averaging this inequality over the output gives
$$
I(O;\mathbf{J})
\leq
2r\log 3
\sum_{j\in\mathcal{J}}p_j
=
2r\log 3.
$$

After some oracle values are fixed, their queries can be hardwired
into the oracle-independent operations. The remaining oracle is
uniform on the corresponding subcube, and the restricted algorithm
still makes at most $r$ queries. Lemma \ref{lem:fixed-outcome-concentration}
therefore applies without change.
\end{proof}

\begin{lemma}[Finite-output regularization from binary-output simulation]\label{lem"regularize}
Let $k\ge 1$, let
\[
    \mathbf T=(T_1,\ldots,T_k),
    \qquad
    Q_k:=\sum_{r=1}^k T_r,
\]
and suppose there is a function $R_k$ with the following
property- for every $0<\rho<1$, after fixing an arbitrary subset of oracle
values, the acceptance probability
\[
    a(O):=\Pr[\mathcal B^O\to 1]
\]
of every binary-output quantum algorithm $\mathcal B$ with $k$ adaptive
parallel-query layers, whose $r$-th layer contains at most $T_r$ queries for
each $r\in[k]$, has a deterministic classical decision-tree approximation
$b(O)$ of depth at most $R_k(\mathbf T,\rho)$ such that
\[
    \mathbb E\bigl[(a(O)-b(O))^2\bigr]\le \rho^2.
\]
Let $\mathcal A$ be any quantum algorithm satisfying the same per-layer query
bounds and having finite output alphabet $\mathcal J$. For each oracle $O$,
define
\[
    p_j^O:=\Pr[\mathcal A^O\to j],
    \qquad
    \mathbf p^O:=(p_j^O)_{j\in\mathcal J}.
\]
Then, for every $0<\eta<1$, there is a deterministic classical decision tree
of depth at most
\[
    V_k(\mathbf T,\eta)
    :=
    \left\lceil \frac{12Q_k\log 3}{\eta}\right\rceil
    R_k\!\left(\mathbf T,\frac{\sqrt{\eta}}{32}\right)
\]
such that, if $Y$ denotes the resulting transcript and
\[
    \boldsymbol\mu_Y
    :=
    \mathbb E_{O\sim O_U}[\mathbf p^O\mid Y],
\]
then
\[
    \mathbb E_{O\sim O_U}
    \left[
        \left\|\mathbf p^O-\boldsymbol\mu_Y\right\|_2^2
    \right]
    <\eta.
\]
The same conclusion holds after fixing an arbitrary subset of oracle values,
with $O$ uniform on the corresponding subcube. In particular, the depth bound
is independent of $|\mathcal J|$.
\end{lemma}

\begin{proof}
It suffices to prove the unrestricted statement; the same argument applies on every restricted
subcube. Let $\mathbf J$ denote the output of $\mathcal A^O$ in the joint experiment in which
$O\sim O_U$ is sampled first and $\mathbf J$ is then sampled according to $\mathbf p^O$.

We construct the classical decision tree by refining it in rounds. We start with the trivial transcript $Y_0$. Given the current transcript $Y_t$, for every leaf $y$ occuring with positive probability, define
\[
    \boldsymbol\mu_y
    :=\mathbb E[\mathbf p^O\mid Y_t= y],
    \qquad
    v_y
    :=\mathbb E\!\left[
        \left\|\mathbf p^O-\boldsymbol\mu_y\right\|_2^2
        \middle|y
    \right].
\]
Stop if $\mathbb{E}[v_{Y_t}]= \sum_{y}\Pr[Y_t=y]v_y<\eta$.
Otherwise, continue whenever $\mathbb{E}[v_{Y_t}]\geq\eta$

Fix such a leaf $y$. 
Sample independent signs $\sigma_j\in\{-1,1\}$ for $j\in\mathcal J$, independently of $O$, and write $\sigma=(\sigma_j)_{j\in\mathcal J}$. For a fixed sign vector $\sigma$, define \[ g_\sigma(O) := \sum_{j\in\mathcal J}\sigma_jp_j^O = \sigma\cdot\mathbf p^O, \]Then we have, for every fixed $\sigma$, $\mathbb E_{O\mid y}[g_\sigma(O)] = \sigma\cdot\boldsymbol{\mu}_y. $Therefore, by the definition of conditional variance, \begin{align*} \operatorname{Var}_{O\mid y}(g_\sigma(O)) &= \mathbb E_{O\mid y}\!\left[ \left( g_\sigma(O) - \mathbb E_{O\mid y}[g_\sigma(O)] \right)^2 \right] \\ &= \mathbb E_{O\mid y}\!\left[ \left( \sigma\cdot \bigl(\mathbf p^O-\boldsymbol{\mu}_y\bigr) \right)^2 \right]. \end{align*}Since $\mathbb{E}_{\sigma}[\sigma_j \cdot \sigma_{j'}]=0$ for all $j \neq j'$, we get \begin{align*} \mathbb E_\sigma\!\left[ \operatorname{Var}_{O\mid y}(g_\sigma(O)) \right] &= \mathbb E_{O\mid y} \mathbb E_\sigma\!\left[ \left( \sigma\cdot \bigl(\mathbf p^O-\boldsymbol{\mu}_y\bigr) \right)^2 \right] \\ &= \mathbb E_{O\mid y}\!\left[ \left\| \mathbf p^O-\boldsymbol{\mu}_y \right\|_2^2 \right] \\ &= v_y. \end{align*} Thus, there exists a fixed sign vector $\sigma\in\{-1,1\}^{\mathcal J}$ such that \[ \operatorname{Var}_{O\mid y}(g_\sigma(O)) \ge v_y. \] For this sign vector, define \[ a_\sigma(O) := \sum_{j:\,\sigma_j=1}p_j^O. \] This is the acceptance probability of the binary-output algorithm that runs the original algorithm and accepts precisely when its output label $j$ satisfies $\sigma_j=1$. Since $ \sum_{j\in\mathcal J}p_j^O=1$, $g_\sigma(O) = 2a_\sigma(O)-1$. It follows that \[ \operatorname{Var}_{O\mid y}(a_\sigma(O)) = \frac14 \operatorname{Var}_{O\mid y}(g_\sigma(O)) \ge \frac{v_y}{4}. \] 

Apply the assumed binary-output simulator on the subcube corresponding to $y$ with
\[
    \rho:=\frac{\sqrt\eta}{32}.
\]
Let $Z_y$ be the resulting classical transcript, and let $b_y$ denote
the output of the simulator. Then
\[
    \mathbb E\!\left[(a_\sigma(O)-b_y(O))^2\mid Y_t=y\right]
    \le \frac{\eta}{1024}
    \le \frac{v_y}{512}.
\]
where the last inequality follows from the fact that $y$ is refined only
when $v_y\ge\eta/2$.
Because conditional expectation is the best mean-square approximation measurable with respect
to $Z$, the law of total variance implies that
\begin{align*}
    \operatorname{Var}\!\left(
        \mathbb E[a_\sigma(O)\mid Z_y,Y_t=y]
        \middle|Y_t=y
    \right)&= \operatorname{Var}(a_\sigma(O)\mid Y_t=y)- \mathbb{E}[\operatorname{Var}(a_{\sigma}(O)|Z_y,Y_t=y)|Y_t=y]\\
    &\ge \operatorname{Var}(a_\sigma(O)\mid Y_t=y)-\mathbb{E}\!\left[(a_\sigma(O)-b_y(O))^2\mid Y_t=y\right]\\&
    \ge
    \operatorname{Var}(a_\sigma(O)\mid Y_t=y)-\frac{v_y}{512}\\&
    \ge \frac{127v_y}{512}.
\end{align*}
Since $g_\sigma=2a_\sigma-1$,
\[
    \operatorname{Var}\!\left(
        \mathbb E[g_\sigma(O)\mid Z_y,Y_t=y]
        \middle|Y_t=y
    \right)
    \ge \frac{127v_y}{128}.
\]

For each value $z$ of $Z_y$, define
$
    \boldsymbol\mu_{y,z}
    :=\mathbb E[\mathbf p^O\mid Y_t=y,Z=z].
$
The conditional distributions of $\mathbf J$ given $Y_t=y$ and given $(Y_t=y,Z_y=z)$ are
$\boldsymbol\mu_y$ and $\boldsymbol\mu_{y,z}$ respectively. Therefore Pinsker's
inequality and $|\sigma\cdot u|\le \|u\|_1$ give
\begin{align*}
    I(\mathbf J;Z_y\mid Y_t=y)
    &=\mathbb E_{z\mid y}\!
      \left[D_{\mathrm{KL}}(\boldsymbol\mu_{y,z}\|\boldsymbol\mu_y)\right]\\
    &\ge \frac12\,
      \mathbb E_{z\mid y}
      \left[\|\boldsymbol\mu_{y,z}-\boldsymbol\mu_y\|_1^2\right]\\
    &\ge \frac12\,
      \mathbb E_{z\mid y}
      \left[(\sigma\cdot(\boldsymbol\mu_{y,z}-\boldsymbol\mu_y))^2\right]\\
    &=\frac12\,
      \operatorname{Var}\!\left(
        \mathbb E[g_\sigma(O)\mid Z_y,Y_t=y]
        \middle|y
      \right)\\
    &>\frac{v_y}{3}.
\end{align*}
On leaves with $v_y <\eta/2$, let $Z_y$ be a constant, and define $Y_{t+1}= (Y_t, Z_{Y_t})$
If $\mathbb{E}[v_{Y_t}]>\eta$, then  leaves with $v_y>\eta/2$ contribute atleast $\eta/2$. Then by the
chain rule,
\[
    I(\mathbf J;Y_{t+1})-I(\mathbf J;Y_t)
    =I(\mathbf J;Z_{Y_t}\mid Y_t)
    >\frac13\sum_{y:\,v_y\ge\eta/2}
      \Pr[Y_t=y]v_y
    \ge\frac{\eta}{6}.
\]

The accumulated transcript $Y_t$ is a deterministic function of $O$. Hence 
Lemma \ref{lem:output-information}, applied with total query count $Q_k$, gives
\[
    I(\mathbf J;Y_t)\le I(\mathbf J;O)\le 2Q_k\log 3.
\]
Consequently, fewer than $12Q_k\log 3/\eta$ refinement rounds are possible. Each round adds
depth at most
\[
    R_k\!\left(\mathbf T,\frac{\sqrt\eta}{32}\right).
\]
When the process stops,
\[
    \mathbb E\!\left[
        \left\|\mathbf p^O-\boldsymbol\mu_Y\right\|_2^2
    \right]
    =\mathbb E[v_Y]
    <\eta,
\]
which proves the stated depth and regularity bounds.
\end{proof}

\subsection{Induction on the number of layers}
\label{subsec:fixed-layers-induction}
We will now consider a quantum algorithm with $d$ layers of parallel queries. We prove the following theorem:

\begin{theorem}[Restriction-stable mean-square simulation]
\label{thm:fixed-layers-l2}Let $d\ge 2$, let
\[
    \mathbf T=(T_1,\ldots,T_d),
    \qquad
    Q_r:=\sum_{s=1}^r T_s: \forall r \in [d], \text{ and } \mathbf T_{<r}=(T_1,\ldots,T_{r-1})
\]
Let $\mathcal{A}$ be a quantum algorithm with $d$ adaptive parallel-query layers, whose $r^{th}$ layer contains at most $T_r$ queries for all $r \in [d]$,  and let $f(O)$ be its acceptance probability. For every $0<\rho<1$, set $\tau
:=
\frac{\rho^4}{256d^2 Q_d^4}$ and let $R_1(T_1, \rho):= S_{\text{par}}(T_1, \rho)$ be the one layer bound from Lemma \ref{lem:parallel-l2-base}. One may take

\begin{equation}\label{eq:depthbound}
R_d(\mathbf{T},\rho)
\leq{}
\left\lceil
\frac{T_1}{\tau}
\right\rceil
\\
+
\sum_{r=2}^d
\left[
\left\lceil
\frac{48Q_{r-1}T_r^2\ln 3}{\tau^2}
\right\rceil
R_{r-1}
\left(
\mathbf{T}_{<r},
\frac{\tau}{64T_r}
\right)
+
\left\lceil
\frac{2T_r}{\tau}
\right\rceil
\right],
\end{equation}
Then, there is a deterministic classical decision tree of depth at most $R_d(\mathbf{T},\rho)$ with output $g(O)$ such that
$$
\mathbb{E}_{O\sim O_U}
\left[
(f(O)-g(O))^2
\right]
\leq
\rho^2.
$$
The same holds even after fixing any arbitrary subset of the input bits.

\end{theorem}

\begin{proof}
We argue by strong induction on $d$. The case when $d=1$ is Lemma~\ref{lem:parallel-l2-base}. Fix $k\geq 1$ and assume that the theorem holds, on every restriction,for every number of layers $j \in \{1, \dots ,k\}$.
Now, let $\mathcal{A}$ have $k+1$ layers. We construct the classical decision tree in $k+1$ stages such that  for every quantum layer $1\leq r\leq k+1$, the unfixed/free coordinates $\text{Free}_r$ satisfy
\begin{equation}\label{eq_var}
\mathbb{E}_O\left[\max_{i\in \text{Free}_r} W_{r,i}(O)\right]\leq \tau.
\end{equation}
Then, we will use Lemma \ref{lem:d-layer-variance} to convert Equation \ref{eq_var} into the final mean-square error.
\paragraph{Stage ${r=1}$: control the first query layer}
Query every coordinate $i$ satisfying
$$
W_{1,i}>\tau.
$$
The first-layer weights are oracle-independent and sum to at most $T_1$. Thus at most $T_1/\tau$ coordinates are queried. For all decision tree leaves $y$, let $\text{Free}_y$ be the set of coordinates that remain free/unfixed at leaf $y$. Thus every $i \in \text{Free}_y$ has weight at most $\tau$.
\paragraph{Stage $r$, for $2\le r\le k+1$.}
Fix $r\ge2$. Layers containing no queries may be omitted, so assume
$T_r\ge1$.
Define a distribution on $[N]\cup\{\perp\}$ by
\[
p_i^O(r):=\frac{W_{r,i}(O)}{T_r}
\quad\text{for every }i\in[N],
\qquad
p_\perp^O(r):=1-\sum_{i=1}^N p_i^O(r),
\]
and write
\[
\mathbf p^O(r):=(p_i^O(r))_{i\in[N]\cup\{\perp\}}.
\]
Since $\sum_i W_{r,i}(O)\le T_r$, this is a probability distribution.

This distribution can be produced using only the first $r-1$ query
layers. Run the algorithm up to the state immediately before layer $r$
and measure its query registers, obtaining a tuple
$(i_1,\ldots,i_s)$ with $s\le T_r$. Independently choose $K$ uniformly
from $[T_r]$. If $K\le s$ and $i_K$ does not occur among
$i_1,\ldots,i_{K-1}$, output $i_K$; otherwise output $\perp$.
Conditioned on the measured tuple, each distinct query appearing
in it is output with probability exactly $1/T_r$. Consequently,
\[
\Pr[\text{output }i]
=\frac{\Pr[i\text{ appears in the measured tuple}]}{T_r}
=\frac{W_{r,i}(O)}{T_r}.
\]

\paragraph{Stabilize the distribution and query its heavy coordinates.}
Let $X_{r-1}$ be the cumulative classical transcript before stage $r$,
and set
\[
\eta_r:=\frac{\tau^2}{4T_r^2}.
\]
For each value $x$ of $X_{r-1}$ having positive probability, apply
Lemma \ref{lem"regularize} separately on the restricted subcube corresponding to
$X_{r-1}=x$. On this subcube, the selector above is an algorithm with
at most $r-1$ query layers and per-layer query bounds
$\mathbf T_{<r}$. The restriction-stable induction hypothesis therefore
supplies the binary-output simulators required by Lemma \ref{lem"regularize}.

Let $Z_{r,x}$ denote the transcript of this application, and set
\[
Z_r:=Z_{r,X_{r-1}},
\qquad
Y_r:=(X_{r-1},Z_r).
\]
The additional depth at every existing leaf is at most
\[
\begin{aligned}
V_{r-1}(\mathbf T_{<r},\eta_r)
&=
\left\lceil\frac{12Q_{r-1}\log 3}{\eta_r}\right\rceil
R_{r-1}\left(\mathbf T_{<r},\frac{\sqrt{\eta_r}}{32}\right)\\
&=
\left\lceil\frac{48Q_{r-1}T_r^2\log 3}{\tau^2}\right\rceil
R_{r-1}\left(\mathbf T_{<r},\frac{\tau}{64T_r}\right).
\end{aligned}
\]
Attaching these trees at the existing leaves increases the total
depth by at most this same bound.

For every value $y$ of $Y_r$ having positive probability, let
$\mathrm{Free}_y$ be the coordinates still unfixed at that leaf, and
define
\[
\boldsymbol\mu_y^{(r)}
:=\mathbb E[\mathbf p^O(r)\mid Y_r=y],
\qquad
\mu_{y,i}^{(r)}:=(\boldsymbol\mu_y^{(r)})_i.
\]
Lemma \ref{lem"regularize} gives, for every initial leaf $x$,
\[
\mathbb E\left[
  \left\|\mathbf p^O(r)-\boldsymbol\mu_{Y_r}^{(r)}\right\|_2^2
  \,\middle|\,X_{r-1}=x
\right]<\eta_r.
\]
Averaging over $X_{r-1}$ therefore yields
\[
\mathbb E\left[
  \left\|\mathbf p^O(r)-\boldsymbol\mu_{Y_r}^{(r)}\right\|_2^2
\right]\le\eta_r.
\]
Query all $i\in \text{Free}_y$, satisfying
$$
\mu_{y,i}^{(r)}>\sqrt{\eta_r}.
$$ 
Let
\[
R_y =\{i \in \text{Free}_y: \mu_{y,i}^{(r)}\leq \sqrt{\eta_r}\}
\]be the coordinates that remain unfixed after these queries. Because the non-dummy coordinates of $\boldsymbol{\mu}_{y}^{(r)}$ have total mass at most $1$, fewer than
$$
\frac{1}{\sqrt{\eta_r}}
=
\frac{2T_r}{\tau}
$$
additional coordinates are queried at any leaf.
Now fix an oracle $O$ that reaches leaf $y$, and let $i \in R_y$. Since $i$ was not queried, 
\[
\mu_{y,i}^{(r)}\leq \sqrt{\eta_r}
\]
Therefore, 
\begin{align*}
    p_i^{O}(r)
&\leq
\mu_{y,i}^r+ | p_i^{O}(r)-\mu_{y,i}^r|\\
&\leq 
\sqrt{\eta_r}
+
\left\|
\boldsymbol{p}^{(O)}(r)-\boldsymbol{\mu}_{y}^{(r)}
\right\|_2.
\end{align*}
Using $ W_{r,i}(O)= T_r \cdot p_i^{O}(r)$, we can conclude that
\begin{align*}
    \max_{i\in R_y}W_{r,i}(O)\leq T_r\cdot (\sqrt{\eta_r}
+
\left\|
\boldsymbol{p}^{(O)}(r)-\boldsymbol{\mu}_{y}^{(r)}
\right\|_2)
\end{align*}
Averaging first over the oracles reaching each leaf and then over the leaves, and applying Jensen's inequality together gives
\[
\begin{aligned}
&\sum_y \Pr[Y_r=y]\,
  \mathbb E\left[
    \max_{i\in R_y}W_{r,i}(O)
    \,\middle|\,Y_r=y
  \right]\\
&\qquad\le
T_r\left(
  \sqrt{\eta_r}
  +\mathbb E\left[
    \left\|\mathbf p^O(r)-\boldsymbol\mu_{Y_r}^{(r)}\right\|_2
  \right]
\right)\\
&\qquad\le
T_r\left(
  \sqrt{\eta_r}
  +\sqrt{
    \mathbb E\left[
      \left\|\mathbf p^O(r)-\boldsymbol\mu_{Y_r}^{(r)}\right\|_2^2
    \right]
  }
\right)\\
&\qquad\le 2T_r\sqrt{\eta_r}
=\tau.
\end{aligned}
\]

After completing all stages, let $Y$ denote the final transcript. For each final leaf $z$, let $\text{Free}_z$ be the set of coordinates that remain unfixed at $z$. For each layer $r$, let $y_r(z)$ denote the unique stage $r$ leaf on the path from the root to the leaf $z$. Since every later stage only queries additional coordinates,
\[\text{Free}_z\subseteq R_{y_r(z)}\]
Thus, for every oracle reaching $z$, 
\[
\max_{i \in \text{Free}_z} W_{r,i}(O)\leq \max_{i \in R_{y_r(z)}} W_{r,i}(O)
\]
Averaging over the final leaves and grouping them according to their stage-$r$ ancestor gives, 
\begin{align*}
    \sum_z \Pr[Y=z]\mathbb{E}\Big[\max_{i \in \text{Free}_z} W_{r,i}(O)| Y=z\Big]
    &\leq \sum_y \Pr[Y_r=y]\mathbb{E}\Big[\max_{i \in R_y} W_{r,i}(O)| Y_r=y\Big]\\
    &\leq \tau
\end{align*}
\paragraph{Counting the Classical Queries:}
By construction, for layer $r$ such that $2\leq r\leq d$, the classical decision tree has depth \[\left\lceil
\frac{48Q_{r-1}T_r^2\log 3}{\tau^2}
\right\rceil
R_{r-1}
\left(
\mathbf{T}_{<r},
\frac{\tau}{64T_r}
\right),
\]
and at most $\frac{2T_r}{\tau} $ additional classical queries are made. Adding these query costs over layers $2$ through $d$, together with the $\frac{T_1}{\tau}$ queries made when $r=1$, gives equation \ref{eq:depthbound}.

Now, we will complete the induction- let $Y$ be the final classical transcript, and for every final leaf $y$, define
$
m_{r,y}
:=
\mathbb{E}
\left[
\max_{i\in\text{Free}_y}W_{r,i}(O)
\mathrel{\big|}
Y=y
\right]$. The construction gives
$
\mathbb{E}[m_{r,Y}]
\leq
\tau$.
For every final leaf y, define its output label by \[g_y
:=
\mathbb{E}[f(O)\mid Y=y].
\]
and define the output of the decision tree to be $g(O)= g_{Y(O)}$. Then, $\mathbb{E}
\left[
(f(O)-g(O))^2
\right]= \mathbb{E}[\operatorname{Var}(f(O)|Y)]$.

 For every final leaf $y$, Lemma~\ref{lem:d-layer-variance} gives
\[\operatorname{Var}(f(O)\mid Y=y)
\leq
16Q_d^2
\sum_{r=1}^d\sqrt{m_{r,y}}
\]
Therefore, by Jensen's inequality, 
\begin{align*}
 \mathbb{E}_{O\sim O_U}[\operatorname{Var}(f(O)\mid Y)]   
&\leq
16Q_d^2
\sum_{r=1}^d
\mathbb{E}[\sqrt{
m_{r,Y}]
}\\
&\leq
16Q_d^2
\sum_{r=1}^d
\sqrt{
\mathbb{E}[m_{r,Y}]
}\\
&\leq
16Q_d^2d\sqrt{\tau}
\\
&=
\rho^2.
\end{align*}
This completes the induction.
\end{proof}

\subsection{Query complexity and almost-everywhere simulation}
\label{subsec:fixed-layers-complexity}

For integers $d,T\ge 1$ and for $0<\rho<1$, let $\overline R_d(T,\rho)\ge 1$ be the  least number such that every binary-output quantum algorithm with at most $d$ adaptive parallel-query layers and at most $T$ total queries admits, on every restricted subcube, a deterministic decision-tree approximation of depth at most $\overline R_d(T,\rho)$ and mean-square error at most $\rho^2$. By definition, $\overline R_d(T,\rho)$ is nondecreasing in $d,T , 1/\rho$.

Let $\mathbf{T}=(T_1,\ldots,T_d),Q_r:=\sum_{s=1}^r T_s. $ Therefore $Q_d\leq T$. For the parameter $
\tau:=\frac{(\rho)^4}{256d^2Q_d^4}
$ in Theorem \ref{thm:fixed-layers-l2} we have

\begin{equation}\label{eq:Q_bound}
\sum_{r=2}^d Q_{r-1}T_r^2
\le
Q_d\sum_{r=1}^dT_r^2
\le
T^3,
\end{equation}

as well as

\begin{equation}\label{eq:tau_bound}
\frac{1}{\tau^2}
\le
\frac{2^{16}d^4T^8}{\rho^8}
\end{equation}

and, for every $r\in [d]$,

\[
\frac{\tau}{64T_r}
\ge
\frac{\rho^4}{2^{14}d^2T^5}=\sigma.
\]
So we can conclude that
\[
R_{r-1}(\mathbf{T}_{<r}, \frac{\tau}{64T_r})\leq \overline R_{d-1}(T, \sigma).
\]
the recurrence in Theorem \ref{thm:fixed-layers-l2} gives 
\begin{align*}
\overline R_d(\mathbf{T},\rho)
&\leq{}
\left\lceil
\frac{T_1}{\tau}
\right\rceil +
\sum_{r=2}^d
\left[
\left\lceil
\frac{48Q_{r-1}T_r^2\ln 3}{\tau^2}
\right\rceil
\overline R_{r-1}
\left(
\mathbf{T}_{<r},
\frac{\tau}{64T_r}
\right)
+
\left\lceil
\frac{2T_r}{\tau}
\right\rceil
\right]\\
&\leq\Big[ \frac{48\ln 3}{\tau^2}\sum_{r=2}^d Q_{r-1}T_r^2+ \frac{2T}{\tau}+2d\Big]\cdot \overline R_{d-1}(T, \sigma)\\
&\leq \Big[C\cdot \frac{T^{16} d^8}{\rho^8}\Big]\cdot \overline R_{d-1}(T, \sigma)
\end{align*}
for some large enough constant $C$. The last inequality follows because $0<\rho<1$, and due to Equations \ref{eq:Q_bound} and \ref{eq:tau_bound}.

By Lemma \ref{lem:parallel-l2-base}, there exists a large enough exponent $e_1 $ such that 

\[
\overline R_1(T,\rho)
\le
C \cdot \Big(\frac{T^2}{\rho}\Big)^{e_1}.
\]

\begin{claim}
For $d\geq 2$, define the $e_d$ recursively such that $e_d= 4\cdot e_{d-1}+8$ . Then
\[
\overline R_d(\mathbf{T},\rho)\leq \Big(\frac{ C \cdot  d\cdot T^2}{\rho}\Big)^{e_d}
\]
\end{claim}
\begin{proof}
We will prove by induction. 
\paragraph{Base case,   $d=1:$} The claim holds for the case of $d=1$ by construction.
\paragraph{Induction hypothesis:} Assume that 
\[
\overline R_{d-1}(\mathbf{T}_{<d-1},\rho)\leq \Big(\frac{C \cdot (d-1)\cdot T^2}{\rho}\Big)^{e_{d-1}}
\]
\paragraph{Induction step:} Since $\sigma = \frac{\rho^4}{2^{14}d^2 T^5}$, 
\[
  \frac{(d-1)\cdot T^2}{\sigma}= \frac{2^{14}(d-1)\cdot d^2\cdot  T^7}{\rho^4}\leq 2^{14}\Big(\frac{ d T^2}{\rho}\Big)^{4}
\]
Then, we can conclude that
\begin{align*}
\overline R_d(\mathbf{T},\rho)&\leq \Big[C\cdot \frac{T^{16} d^8}{\rho^8}\Big]\cdot \overline R_{d-1}(T, \sigma)\\
&\leq \Big(\frac{C \cdot T^{2} d}{\rho}\Big)^8\cdot \overline R_{d-1}(T, \sigma)\\
&\leq \Big(\frac{C \cdot T^{2} d}{\rho}\Big)^8\cdot \Big(\frac{  C \cdot(d-1)\cdot T^2}{\sigma}\Big)^{e_{d-1}}\\
&\leq  \Big(\frac{C \cdot T^{2} d}{\rho}\Big)^{4 e_{d-1}+8}\\
&= \Big(\frac{C \cdot T^{2} d}{\rho}\Big)^{e_d}
\end{align*}

\end{proof}
Solving the recurrence, we can conclude that $e_d = 4^{d-1}e_1 + \frac{8}{3}(4^{d-1}-1)$. Therefore there is an absolute constant $c_0$ such that $e_d \leq c_0 \cdot 4^d$. Putting everything together, we get 
\[
\overline R_d(\mathbf{T},\rho)\leq \Big(\frac{C \cdot d \cdot T^2}{\rho}\Big)^{c_0 \cdot 4^d}\leq \Big(\frac{C \cdot d \cdot T}{\rho}\Big)^{c \cdot 4^d}.
\]

for some large enough constant $c$.

\begin{proof}[Proof of Theorem~\ref{thm:fixed-adaptive-layers}]
Apply Theorem~\ref{thm:fixed-layers-l2} with
$$
\rho
:=
\epsilon\sqrt{\delta}.
$$
Markov's inequality gives
$$
\Pr_{O\sim O_U}
\left[
|f(O)-g(O)|>\epsilon
\right]
\leq
\frac{
\mathbb{E}_{O\sim O_U}
\left[
(f(O)-g(O))^2
\right]
}{
\epsilon^2
}
\leq
\delta.
$$
The query bound follows from the preceding estimate for $R_d(T,\rho)$.
\end{proof}

\section{Acknowledgments}
Qipeng Liu is grateful to Shachar Lovett for helpful discussions.
Saachi Mutreja thanks Andrea Coladangelo, Tony Metger, Anand Natarajan, Avishay Tal, Kabir Tomer, John Wright, Henry Yuen and Mark Zhandry for many insightful discussions. 
\printbibliography

\end{document}